\documentclass[11pt,a4paper,fleqn]{article}

\usepackage{amsthm}
\usepackage{amssymb}
\usepackage{amstext}
\usepackage{amsmath}
\usepackage{url}
\usepackage{hyperref}
\usepackage[margin = 40pt,font=small,labelfont=bf]{caption}
\usepackage{arydshln}
\usepackage{pdfsync}
\usepackage[mac]{inputenc}
\usepackage[T1]{fontenc}
\usepackage{color}
\usepackage{algorithm}
\usepackage{algorithmic}
\usepackage{enumerate}
\usepackage{bbm}
\usepackage{ae,aecompl}
\usepackage{pdfpages}
\usepackage{epstopdf}
\usepackage{graphicx}
\usepackage{epsfig}
\usepackage{subfigure}

\numberwithin{equation}{section}

\newcommand{\thefont}[2]{\fontsize{#1}{#2}\fontshape{n}\selectfont}
\def\argmin{\mathop{\rm arg \; min}\limits}%

\theoremstyle{plain}
\newtheorem{prop}{Proposition}[section]
\newtheorem{coro}{Corollary}[section]
\newtheorem{theo}{Theorem}[section]
\newtheorem{lem}{Lemma}[section]

\theoremstyle{definition}
\newtheorem{defin}{Definition}[section]
\newtheorem{exe}{Example}[section]

\theoremstyle{remark}
\newtheorem{rem}{Remark}[section]

\newcommand{\cost}{\operatorname{K}_W}
\newcommand{\costn}{\operatorname{K}_W^{(n)}}
\newcommand{\costnb}{\operatorname{\bf K}_W^{(n)}}
\newcommand{\costh}{\operatorname{K}_X}
\newcommand{\costH}{\operatorname{H}_X}
\newcommand{\costhn}{\operatorname{K}_X^{(n)}}
\newcommand{\costhnb}{\operatorname{\bf K}_X^{(n)}}
\newcommand{\costHn}{\operatorname{H}_X^{(n)}}
\newcommand{\spann}{\operatorname{Sp}}

\newcommand{\bv}{\boldsymbol{v}}
\newcommand{\bu}{\boldsymbol{u}}
\newcommand{\bx}{\boldsymbol{x}}

\newcommand{\bnu}{\boldsymbol{\nu}}

\newcommand{\E}{{\mathbb E}}
\newcommand{\R}{{\mathbb R}}

\renewcommand{\P}{{\mathbb P}}

\newcommand{\M}{{\mathcal M}}
\newcommand{\U}{{\mathcal U}}

\newcommand{\LL}{L^{2}_{\mu}(\Omega)}
\newcommand{\VV}{V_{\mu}(\Omega)}

\newcommand{\CL}{\mathrm{CL}}
\newcommand{\WS}{W_2(\Omega)}

\newcommand{\GG}{\ensuremath{\mathcal G}}
\newcommand{\NN}{\ensuremath{\mathcal N}}

\newcommand{\1}{\rlap{\thefont{10pt}{12pt}1}\kern.16em\rlap{\thefont{11pt}{13.2pt}1}\kern.4em}

\title{Geodesic PCA in the Wasserstein space}

\author{ J\'{e}r\'{e}mie Bigot$^{1}$,  Ra\'ul Gouet$^{2}$, Thierry Klein$^{3}$ \& Alfredo L\'{o}pez$^{4}$  \\
\\  Institut de Math\'ematiques de Bordeaux et CNRS  (UMR 5251)$^{1}$  \\ Universit\'e de Bordeaux  \vspace{0.1cm}  \\ Depto. de Ingenier\'{\i}a Matem\'{a}tica and CMM (CNRS, UMI 2807)$^{2}$ \\ Universidad de Chile \vspace{0.1cm}\\ Institut de Math\'ematiques de Toulouse et CNRS  (UMR 5219)$^{3}$ \\ Universit\'e de Toulouse  \vspace{0.1cm}  \\ CSIRO Chile International Centre of Excellence$^{4}$}

\date{\today}

\begin{document}

\maketitle

\thispagestyle{empty}

\begin{abstract}
We introduce the method of  Geodesic Principal Component Analysis (GPCA) on the space of probability measures on the line, with finite second moment, endowed with the Wasserstein metric. We discuss the advantages of this approach, over a standard functional PCA of probability densities in the Hilbert space of square-integrable functions. We establish the consistency of the method by showing that the empirical GPCA converges to its population counterpart, as the sample size tends to infinity. A key property in the study of GPCA is the isometry between the Wasserstein space and a closed convex subset of the space of square-integrable functions, with respect to an appropriate measure. Therefore, we consider the general problem of PCA in a closed convex subset of a separable Hilbert space, which serves as basis for the analysis of GPCA and also has interest in its own right. We provide  illustrative examples on simple statistical models, to show the benefits of this approach for data analysis. The method is also applied to a real dataset of population pyramids.
\end{abstract}



\section*{Acknowledgments} The authors acknowledge the support of the French Agence Nationale de la Recherche (ANR) under reference ANR-JCJC-SIMI1 DEMOS, and the financial support from the French embassy in Chile and CONICYT PFCHA-2012. J. Bigot would like to thank the Center for Mathematical Modeling (CMM) and the CNRS for financial support and excellent hospitality while visiting Santiago, where part of this work was carried out. R. Gouet acknowledges support from Fondecyt grant 1120408 and project PFB-03-CMM.

\section{Introduction}

\subsection{Main goal of this paper}

The main goal of this paper is to define a notion of principal component analysis (PCA) of a family of  probability measures $\nu_1,\ldots,\nu_n$, defined on the real line $\R$.  In the case where the measures admit square-integrable densities $f_1,\ldots,f_n$, the standard approach is to use functional PCA (FPCA) (see e.g.\ \cite{MR650934,MR2168993, MR1389877}) on the Hilbert space $L^{2}(\R)$, of square-integrable functions, endowed with its usual inner product. This method has already been applied in \cite{MR2736564,MR1946423} for analysing the main modes of variability of a set of densities.

We briefly introduce elements of standard PCA in a separable Hilbert space $H$,  endowed with inner product $\langle \cdot, \cdot \rangle$ and norm $\| \cdot \|$. A PCA of the data $x_{1},\ldots,x_{n}$ in $H$ is carried out by diagonalizing the empirical covariance operator
$K x = \frac{1}{n} \sum_{i=1}^{n} \langle x_{i} - \bar{x}_{n}, x \rangle (x_{i} - \bar{x}_{n}), \; x \in H$,  where $\bar{x}_{n} = \frac{1}{n} \sum_{i=1}^{n} x_{i}$ is the Euclidean mean.

The eigenvectors of $K$, associated to the largest eigenvalues,  describe the principal modes of data variability around $\bar{x}_n$. The first principal mode of linear variation of the data is  defined by the $H$-valued curve $g :\R \to H$ given by
\begin{equation}
g_{t}= \bar{x}_{n} + t{\sigma}_{1}{w}_{1}, \; t \in \R, \label{eq:g}
\end{equation}
where ${w}_{1} \in H$ is the eigenvector corresponding to the largest eigenvalue ${\sigma}_{1} \geq 0$ of $K$. On the other hand, it is well known that PCA can be formulated as the problem of finding a sequence of nested affine subspaces, minimizing the sum of norms of projection residuals. In particular, $w_1$ is a solution of
\begin{equation}\label{eq:pca}
\min_{v \in H, \; \| v \| = 1} \sum_{i=1}^{n} d^2(x_i,S_v) = \min_{v \in H, \; \| v \| = 1} \sum_{i=1}^{n}  \| x_{i} - \bar{x}_{n} - \langle x_{i} - \bar{x}_{n} , v \rangle v  \|^2,
\end{equation}
where $S_{v} = \{\bar{x}_{n} + t v, \; t \in \R\}$ is the affine subspace through $\bar{x}_{n}$, with direction $v \in H$, and $d(x,S) =\inf_{x' \in S}\|x-x'\|$ denotes the distance from $x \in H$ to $S \subset H$.

We illustrate the strategy discussed above on the set of Gaussian densities $f_1,\ldots,f_4$, shown in Figure \ref{fig:exintro}. These densities are sampled from the following location-scale model, to be used throughout the paper as illustrative example: let $f_0$ be a density in $L^{2}(\R)$ and, for $(a_{i},b_{i})  \in (0,\infty) \times \R,\ i=1,\ldots,n$, we define $\nu_{i}$ as the  probability measure with density
\begin{equation}
f_{i}(x) := a_i^{-1}f_{0}\left( a_i^{-1}(x-b_{i})\right), \; x \in \R.\  
\end{equation}

This model is appropriate in many applications such as curve registration and signal warping, see e.g.\ \cite{BK12} and \cite{Gallon2011}. The main sources of variability in these densities are the variation in location along the $x$-axis, and the scaling variation. One of the purposes of this paper is to develop a notion of PCA, that has desirable and coherent properties with respect to this variability and the model. A first requirement is that the principal modes of variation be densities. Moreover, they should reflect the fact that the data vary in location and scale around $f_{0}$.

The densities displayed in Figure \ref{fig:exintro} represent an example of realizations of this model, with $f_0$ the standard normal density and $n=4$. Let us first consider the FPCA of this dataset.
 To that end we compute the Euclidean mean $\bar{f}_4$, shown in Figure \ref{fig:exintro}(e), a bi-modal density which is not a "satisfactory" average of the uni-modal densities $f_1,\ldots,f_4$.
  In Figure \ref{fig:linPC} we display the first mode of linear variation $g$, given by \eqref{eq:g}, and observe that it is not a "meaningful" descriptor of the variability in the data. Indeed, for $|t|$ sufficiently large, $g_{t}$ may take negative values and does not integrate to one, as illustrated in Figure \ref{fig:linPC}(a),(e),(f). Moreover, even for small values of $|t|$, $g_{t}$ does not represent the typical shape of the observed densities, as shown by  Figure \ref{fig:linPC}(c),(d).
  Therefore, the FPCA of densities in $L^{2}(\R)$  is not always appropriate as it may lead to principal modes of linear variation that are not coherent with the sources of variability observed  in the data (e.g.,\ sampled from a location-scale model). To overcome some of these issues, one could constraint the  first mode of variation to lie in the set of positive functions, integrating to one. However, such a constrained PCA would be computed via the $L^{2}(\R)$ norm, so the Euclidean mean $\bar{f}_4$ would stay unchanged and still not be satisfactory. We believe these drawbacks of FPCA are mainly due to the fact that the  Euclidean distance in $L^{2}(\R)$ is not  appropriate to perform PCA for densities. 

\begin{figure}[h!]
\centering
\subfigure[]{ \includegraphics[width=3cm]{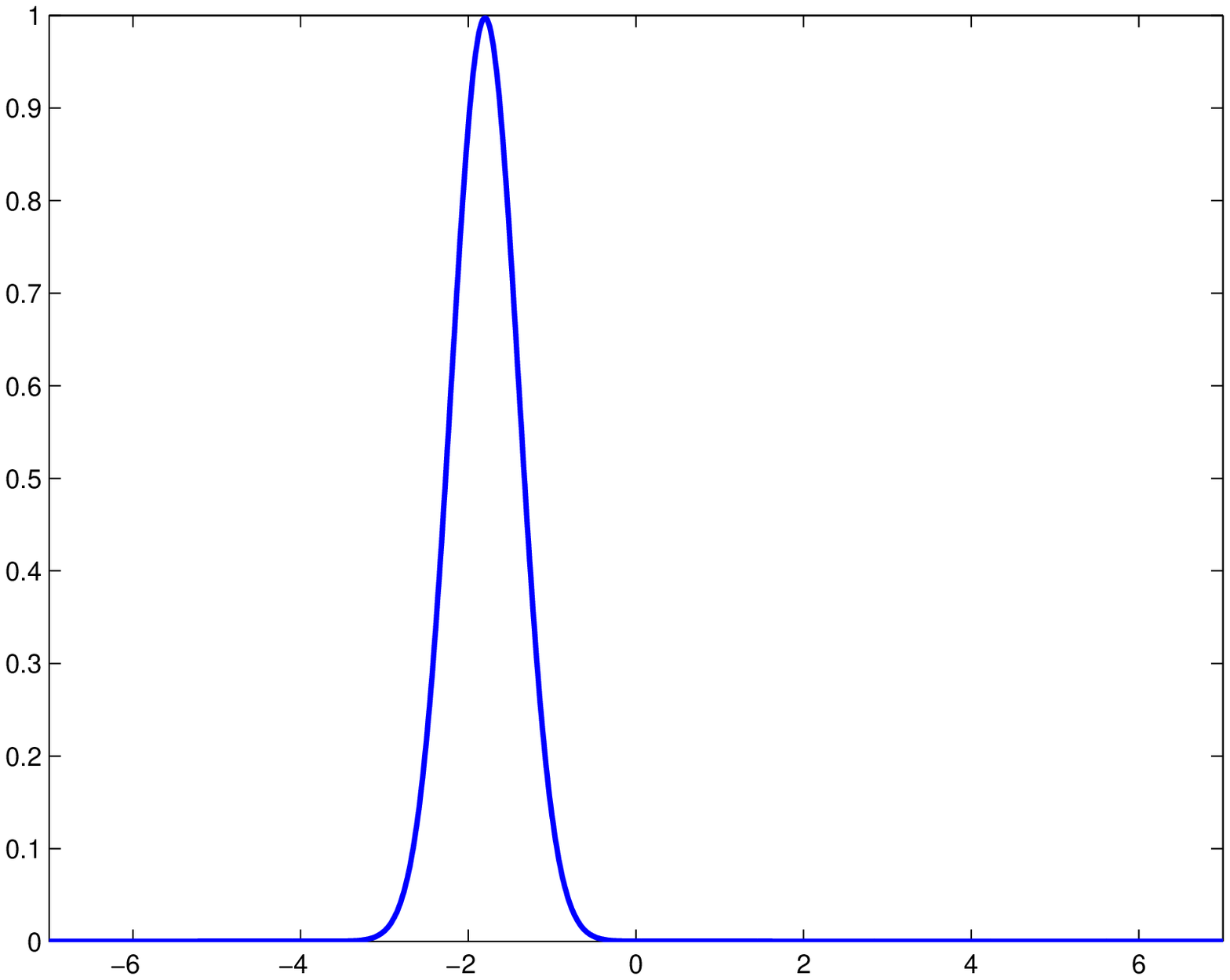} }
\subfigure[]{ \includegraphics[width=3cm]{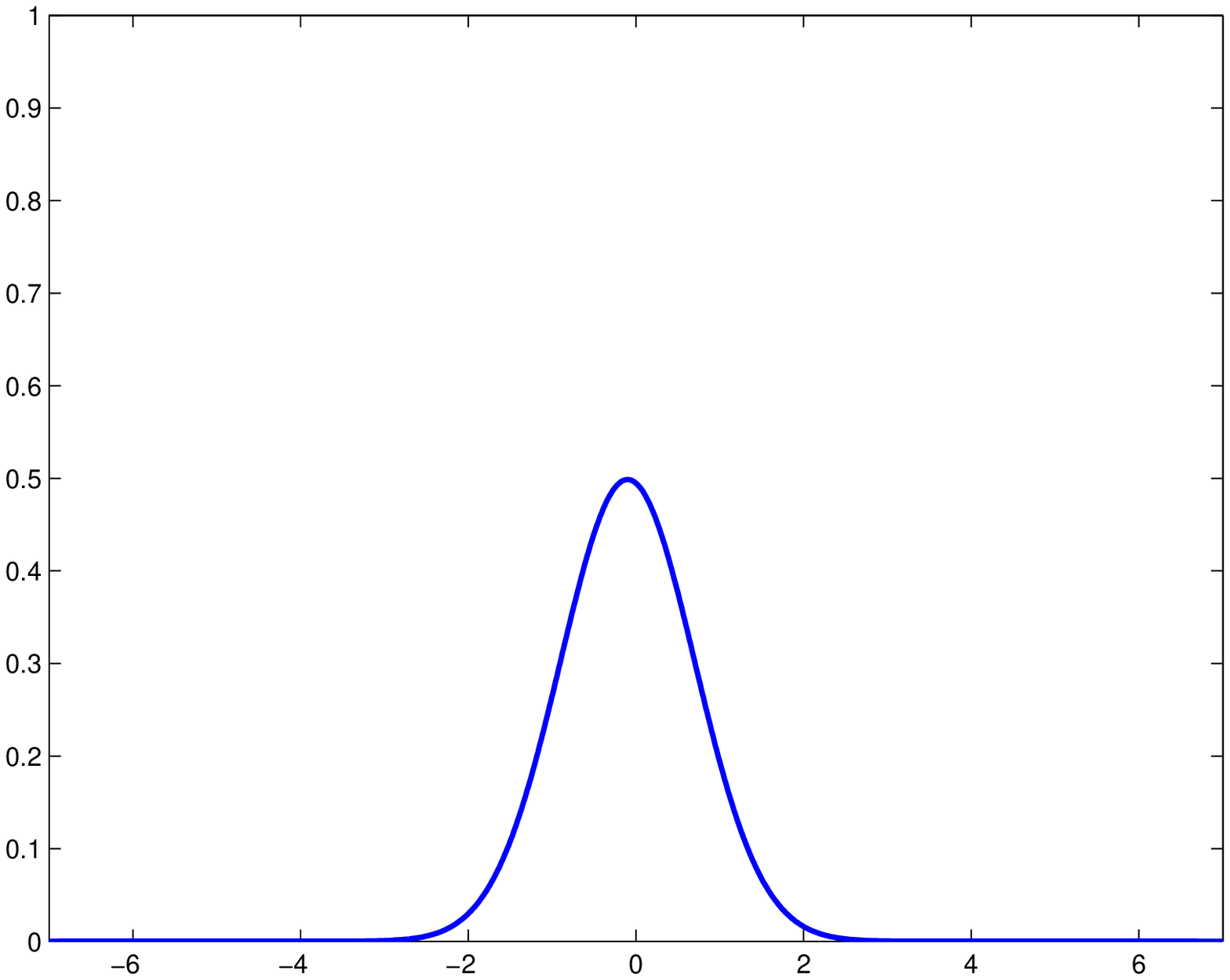} }
\subfigure[]{ \includegraphics[width=3cm]{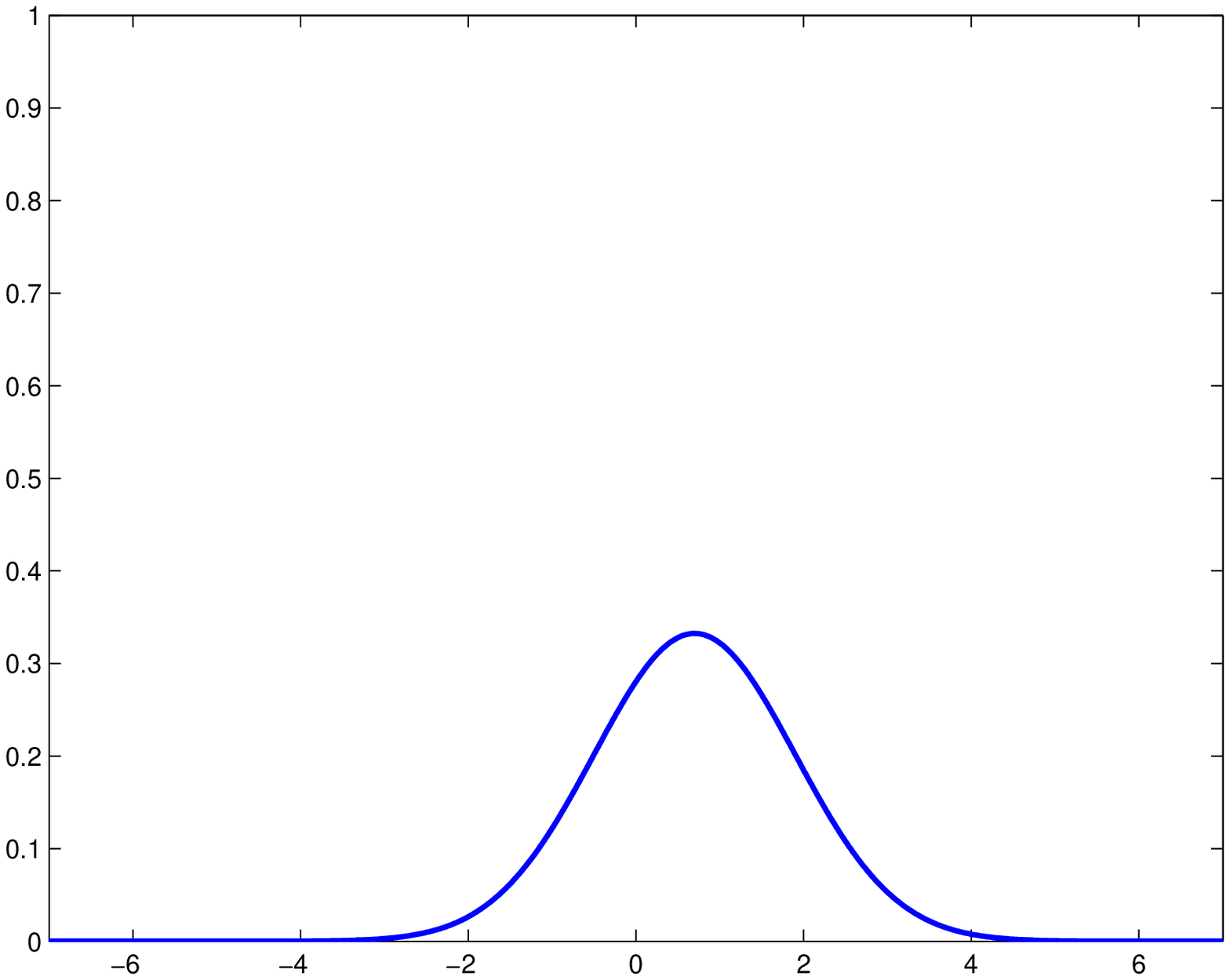} }
\subfigure[]{ \includegraphics[width=3cm]{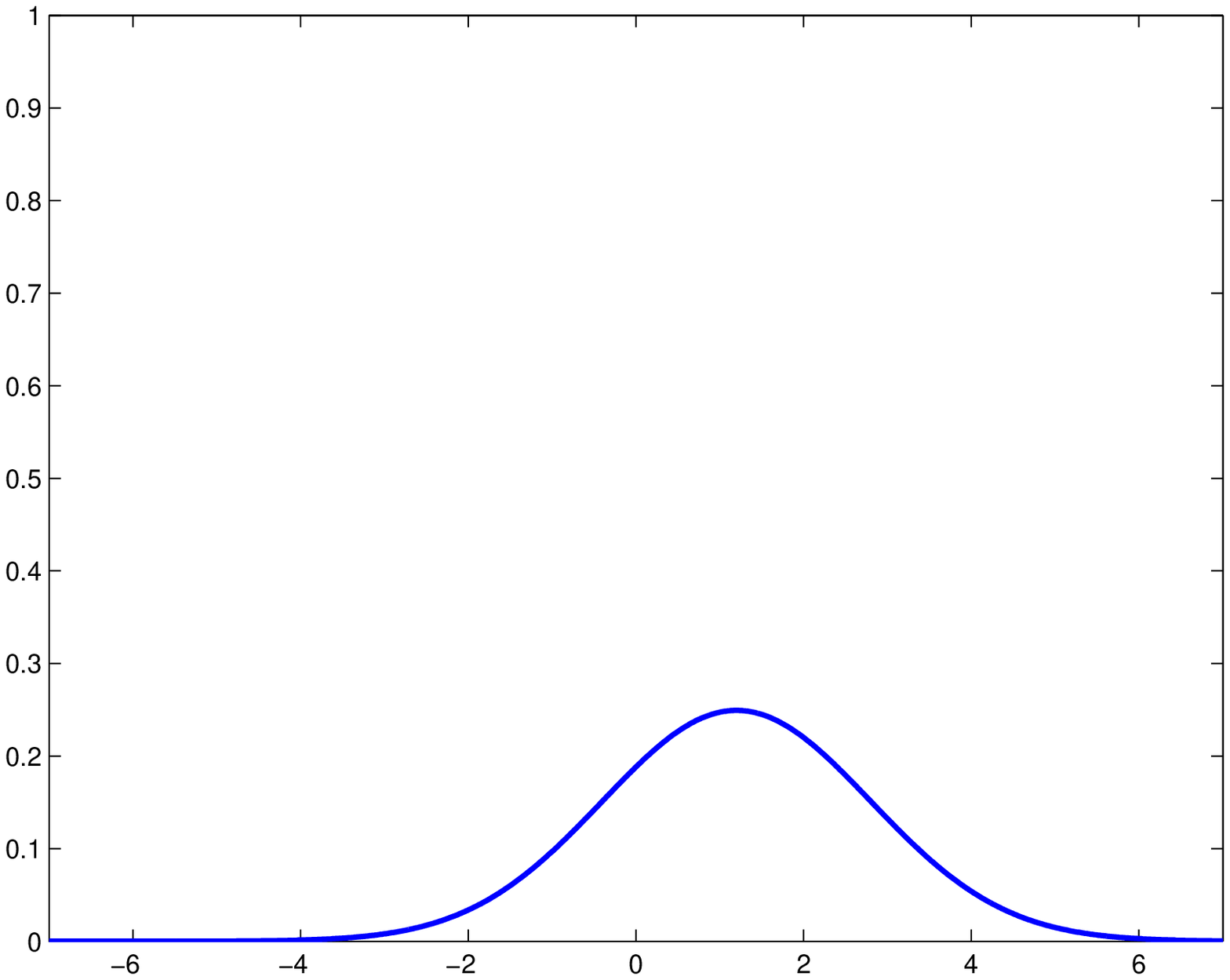} }

\subfigure[]{ \includegraphics[width=3cm]{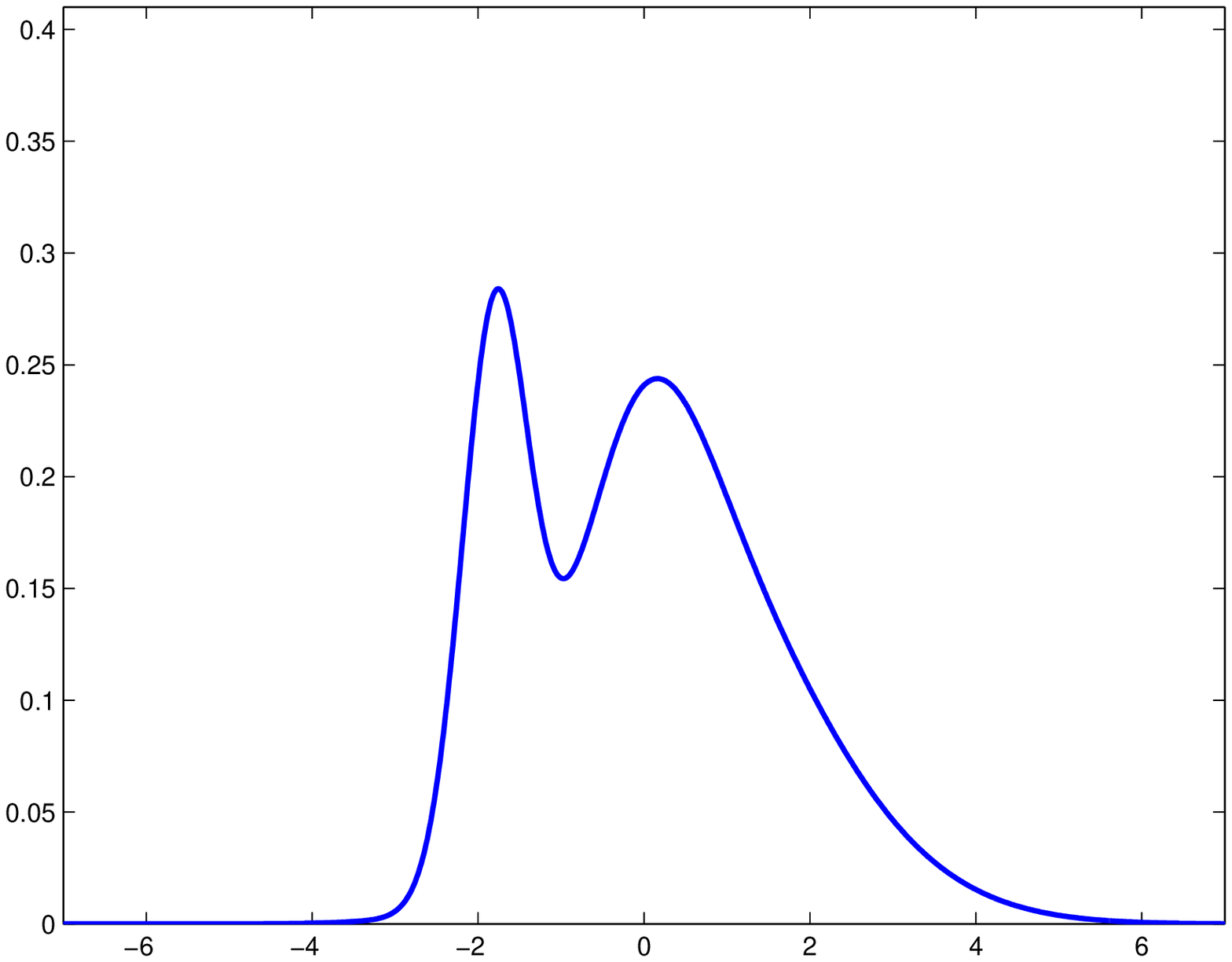} }
\subfigure[]{ \includegraphics[width=3cm]{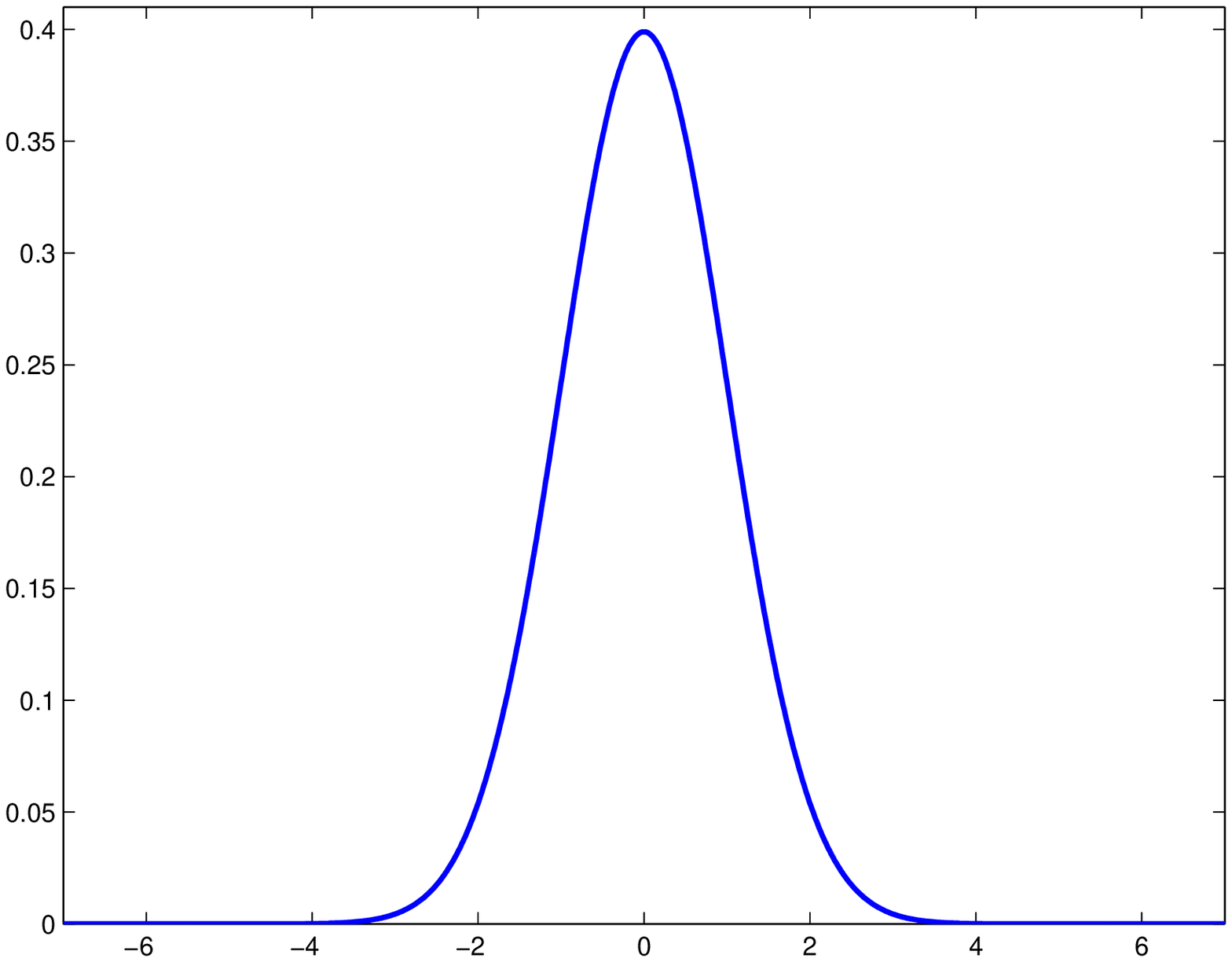} }

\caption{(a,b,c,d) Graphs of Gaussian densities $f_1,\ldots,f_4$, with different means and variances sampled from a location-scale model. (e) Euclidean mean of $f_1,\ldots,f_4$ in $L^2(\R)$. (f) Density of the barycenter $\bar{\nu}_{4}$ of $\nu_1,\ldots,\nu_4$ in the Wasserstein space $W_2(\R)$.} \label{fig:exintro}
\end{figure}

\begin{figure}[h!]
\centering

\subfigure[$g_{-2}$]{ \includegraphics[width=3cm]{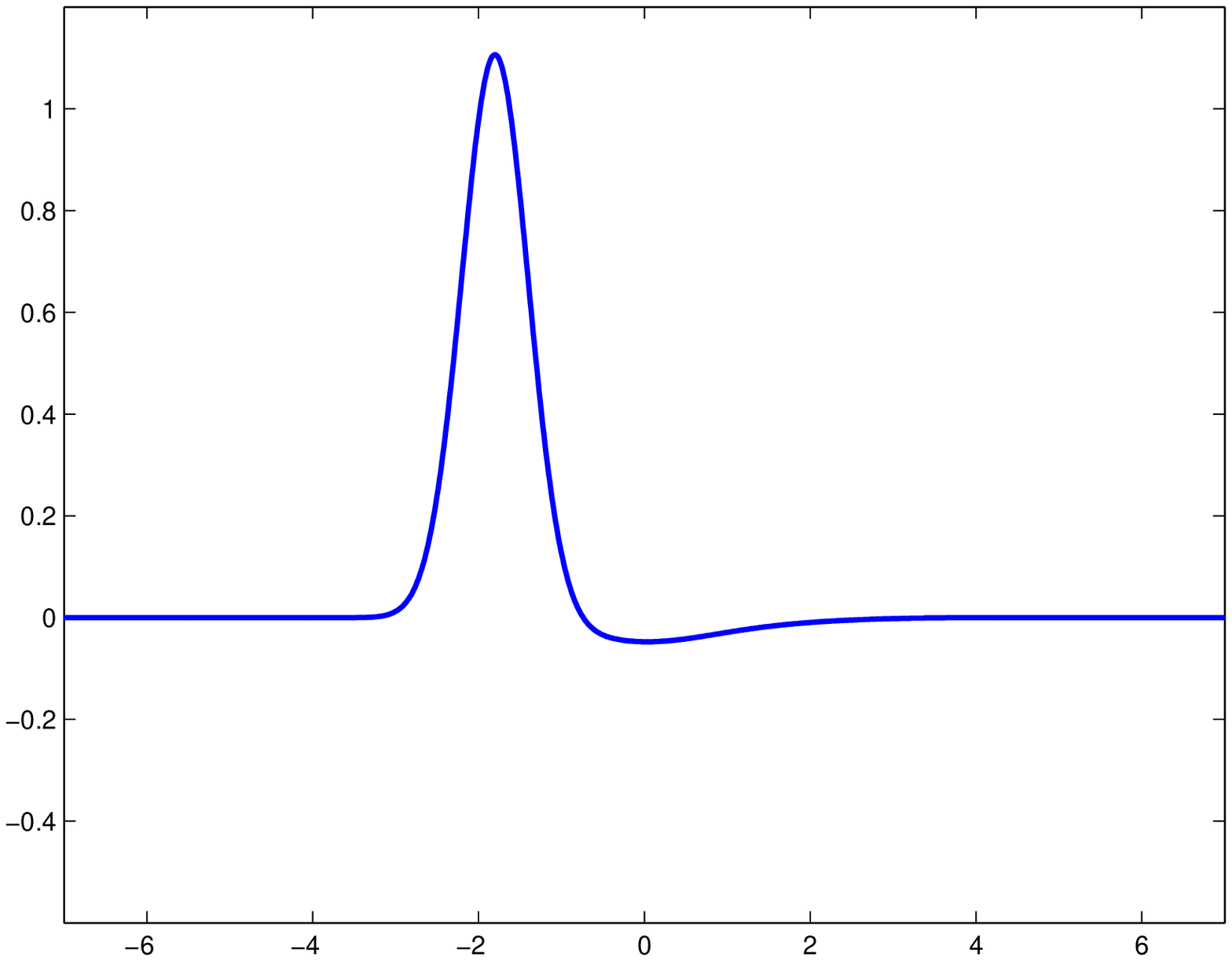} }
\subfigure[$g_{-1}$]{ \includegraphics[width=3cm]{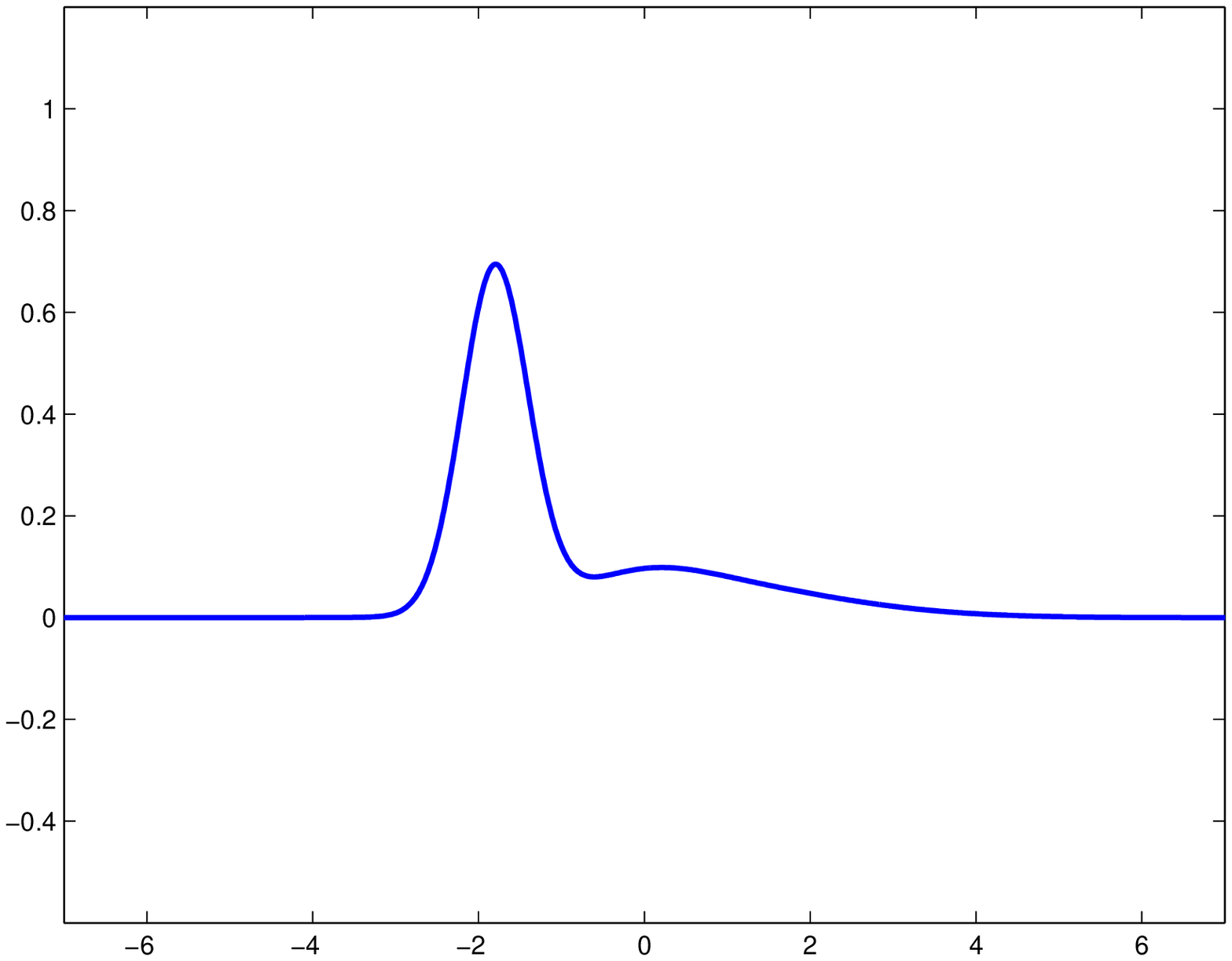} }
\subfigure[$g_{-0.5}$]{ \includegraphics[width=3cm]{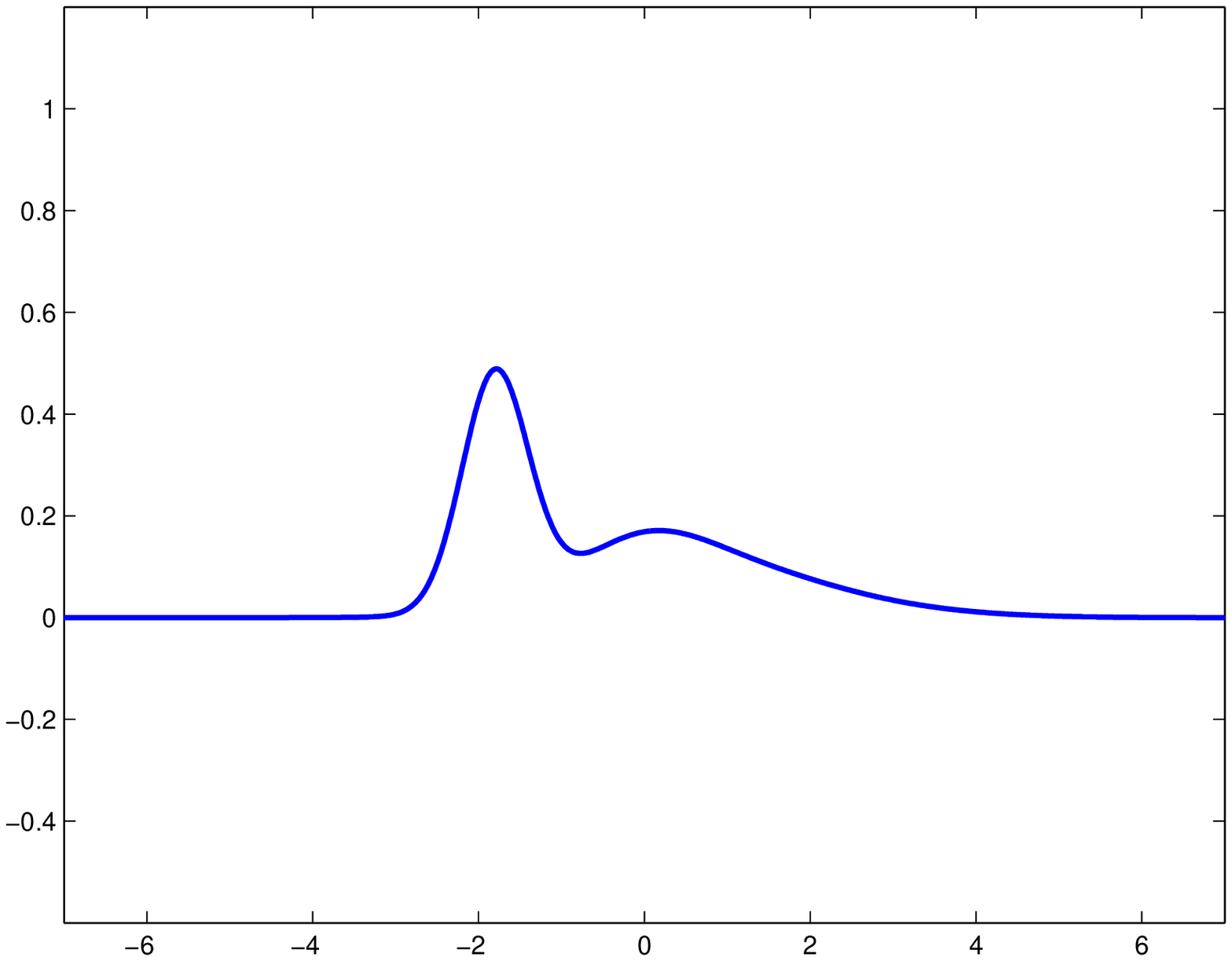} }

\subfigure[$g_{0.5}$]{ \includegraphics[width=3cm]{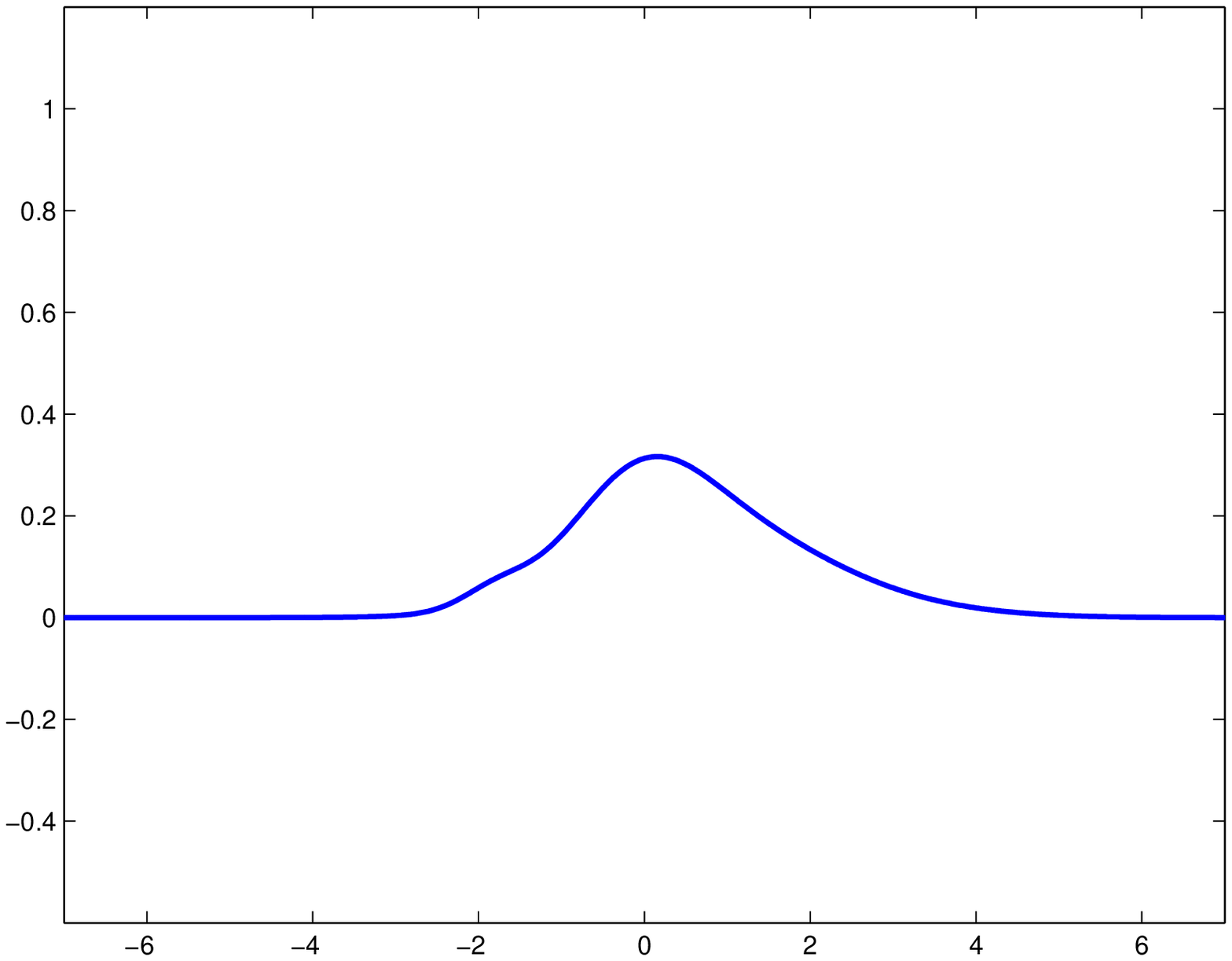} }
\subfigure[$g_{1}$]{ \includegraphics[width=3cm]{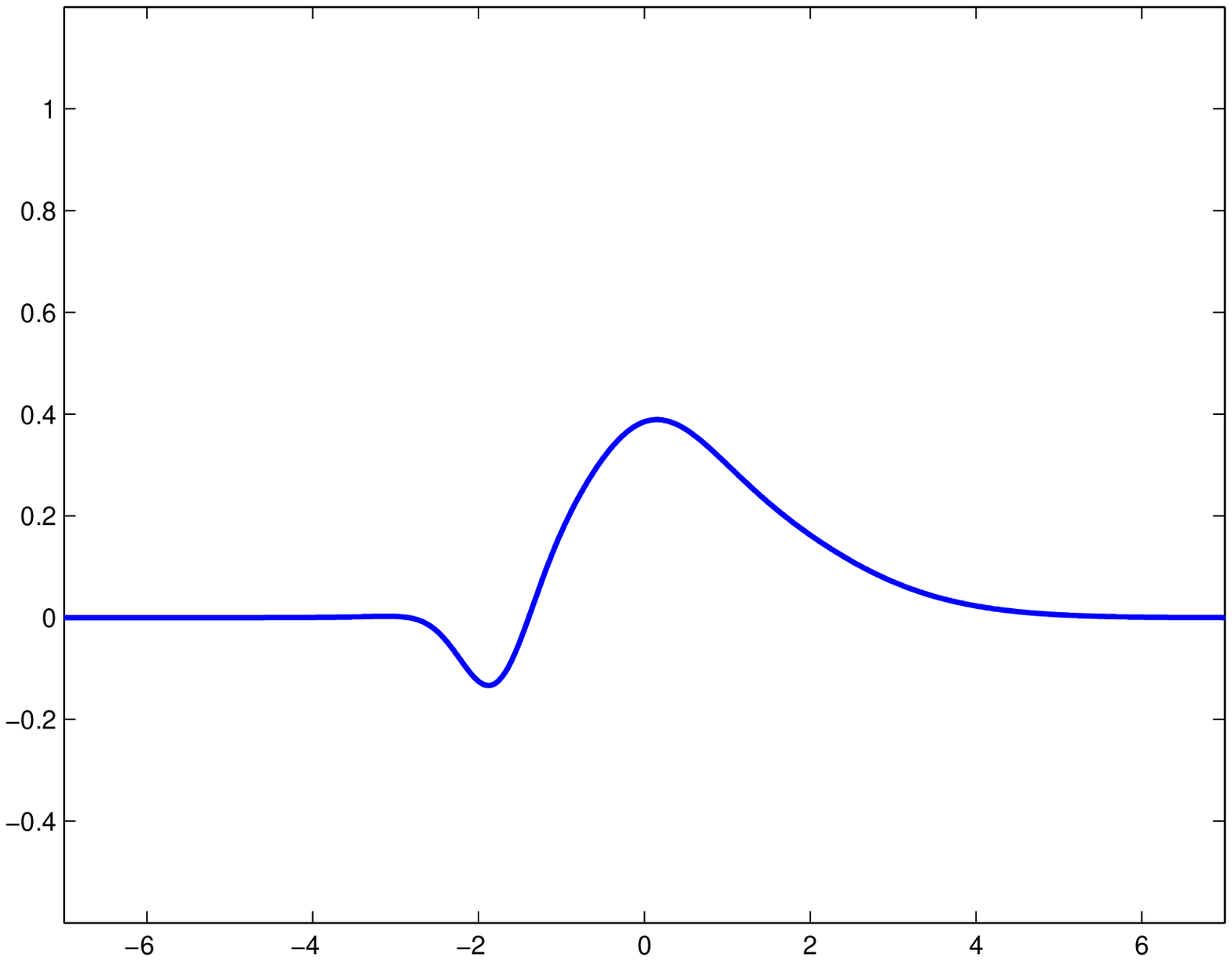} }
\subfigure[$g_{2}$]{ \includegraphics[width=3cm]{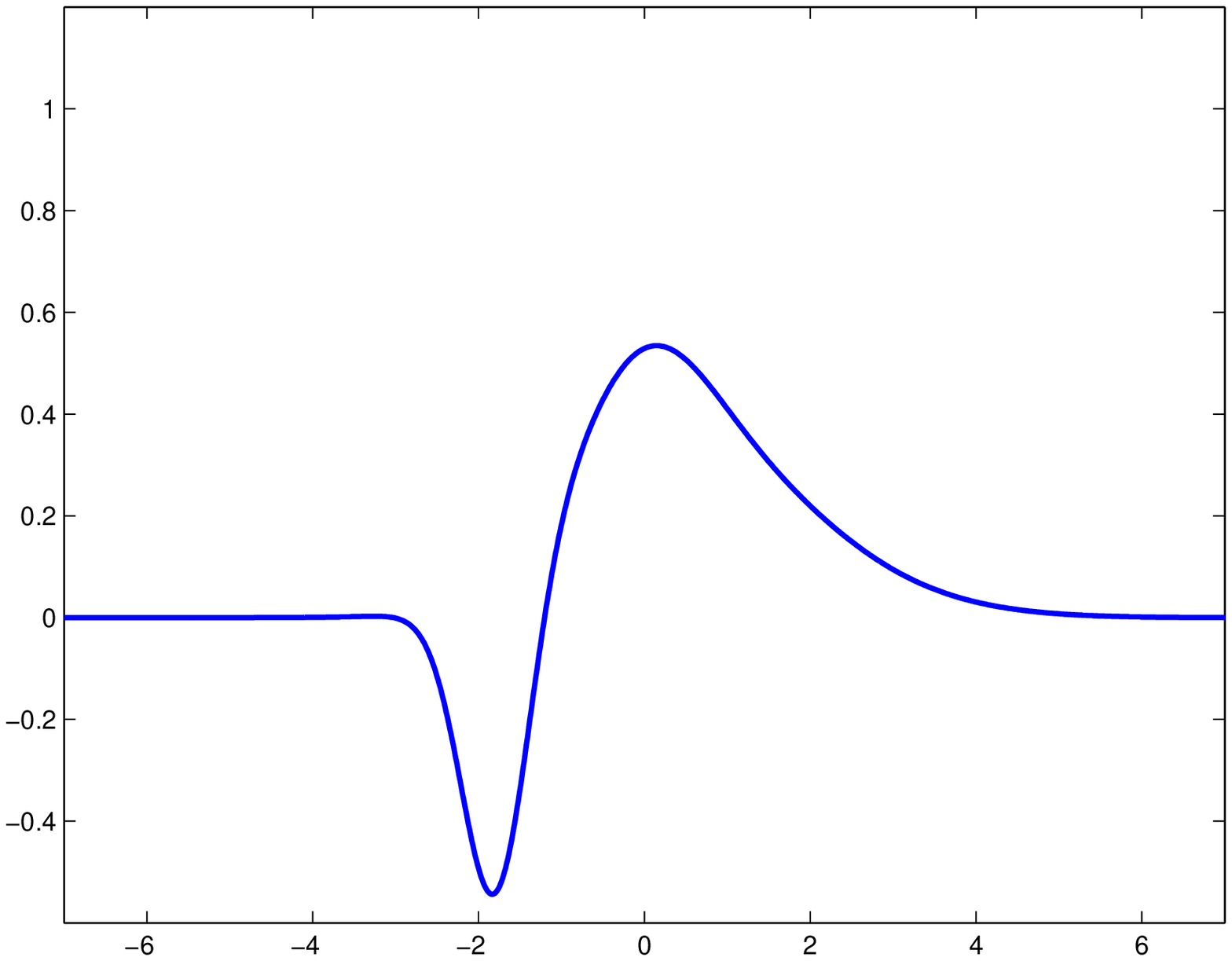} }

\caption{An example of functional PCA of densities. First principal mode of linear variation $g_{t}$ in $L^{2}(\R)$, for $-2 \leq t \leq 2$,  of the densities displayed in Figure \ref{fig:exintro}; see equation \eqref{eq:g}.} \label{fig:linPC}
\end{figure}

\subsection{Main contributions and organization of the paper}

In this paper we suggest to rather consider that $\nu_1,\ldots,\nu_n$  belong to the Wasserstein space $\WS$ of probability measures over $\Omega$, with finite second order moment, where $\Omega$ is $\R$ or a closed interval of $\R$. This space is endowed with the  Wasserstein distance, associated to the quadratic cost; see \cite{villani-topics} for an overview of Wasserstein spaces. In this setting it is not possible to  define a notion of PCA in the usual sense as $\WS$ is not a linear space. Nevertheless,  we show how to define a proper notion of Geodesic PCA (GPCA), by relying on the formal Riemannian structure of $\WS$, developed in \cite{ambrosio2004gradient}, that we describe in Section \ref{sec:Riemannian}. A first idea in that direction is related to the mean of the data, which is an essential ingredient in any notion of PCA. We propose to use the Fr\'echet mean (also called barycenter) as introduced in \cite{MR2801182}, with asymptotic properties studied in \cite{BK12}. It is significant that the barycenter of $\nu_1,\ldots,\nu_4$, in our example above, preserves the shapes of the densities; see Figure \ref{fig:exintro}(f).

Before precisely defining  GPCA in $\WS$, we display $\tilde{g}$  in Figure \ref{fig:geoPC}, the first principal mode of geodesic variation in $\WS$, of the data  displayed in Figure \ref{fig:exintro}; see equation \eqref{eq:gtilde}. GPCA clearly gives  a better  description of the variability in the data, compared to the results in Figure \ref{fig:linPC}, that correspond to the  first principal  mode of linear variation $g$ in $L^{2}(\R)$, given by \eqref{eq:g}.

\begin{figure}[h!]
\centering

\subfigure[$\tilde{g}_{-2}$]{ \includegraphics[width=3cm]{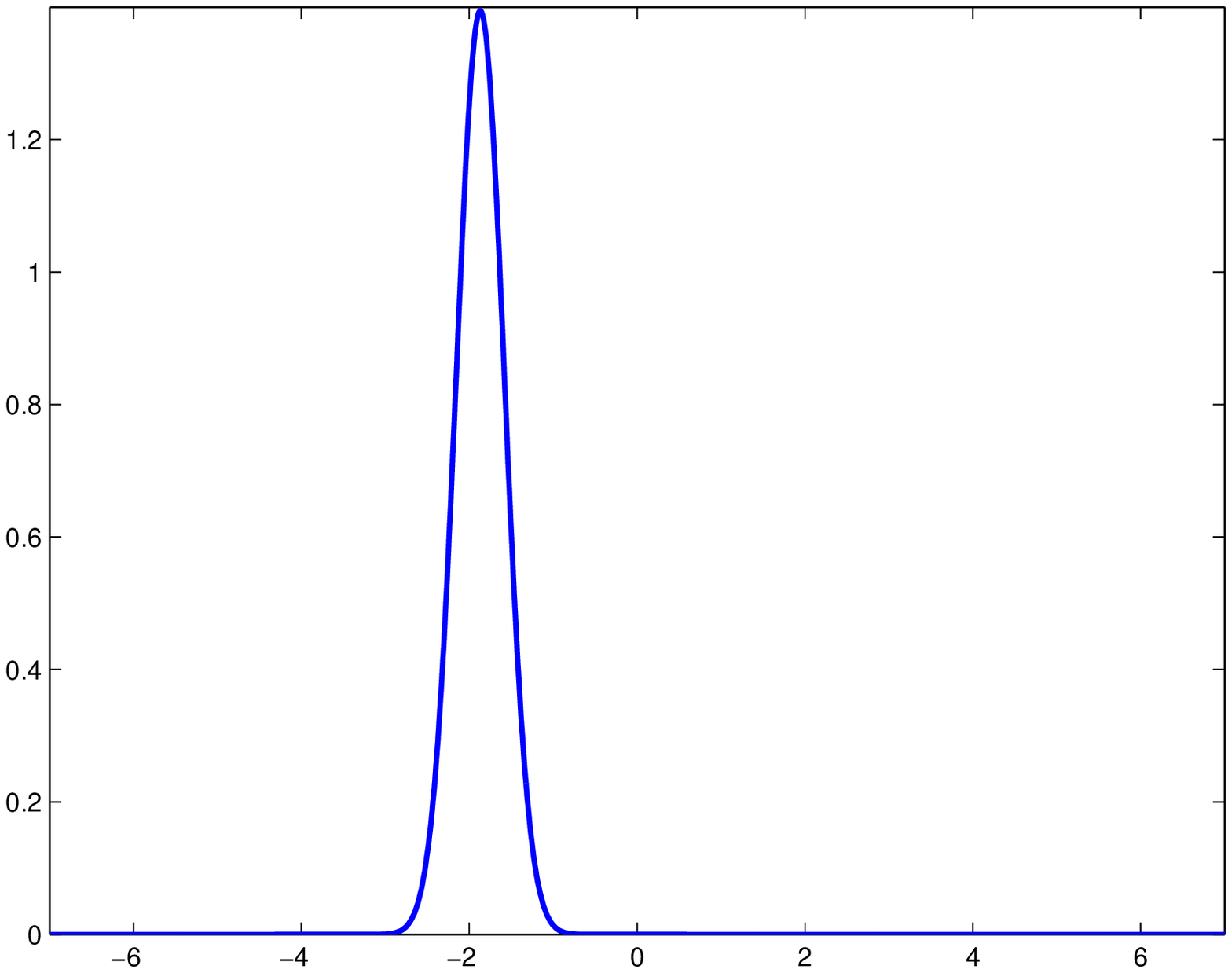} }
\subfigure[$\tilde{g}_{-1}$]{ \includegraphics[width=3cm]{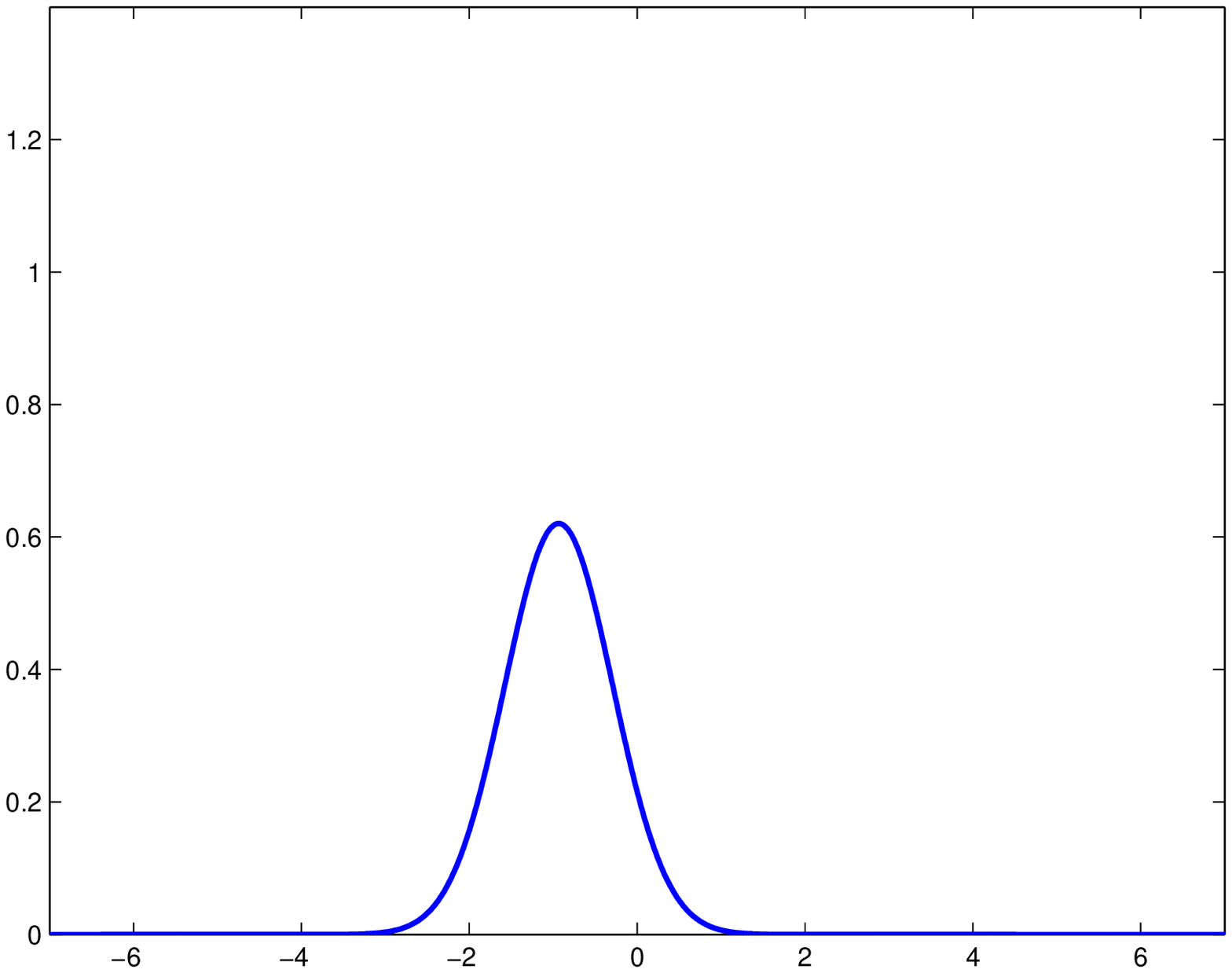} }
\subfigure[$\tilde{g}_{-0.5}$]{ \includegraphics[width=3cm]{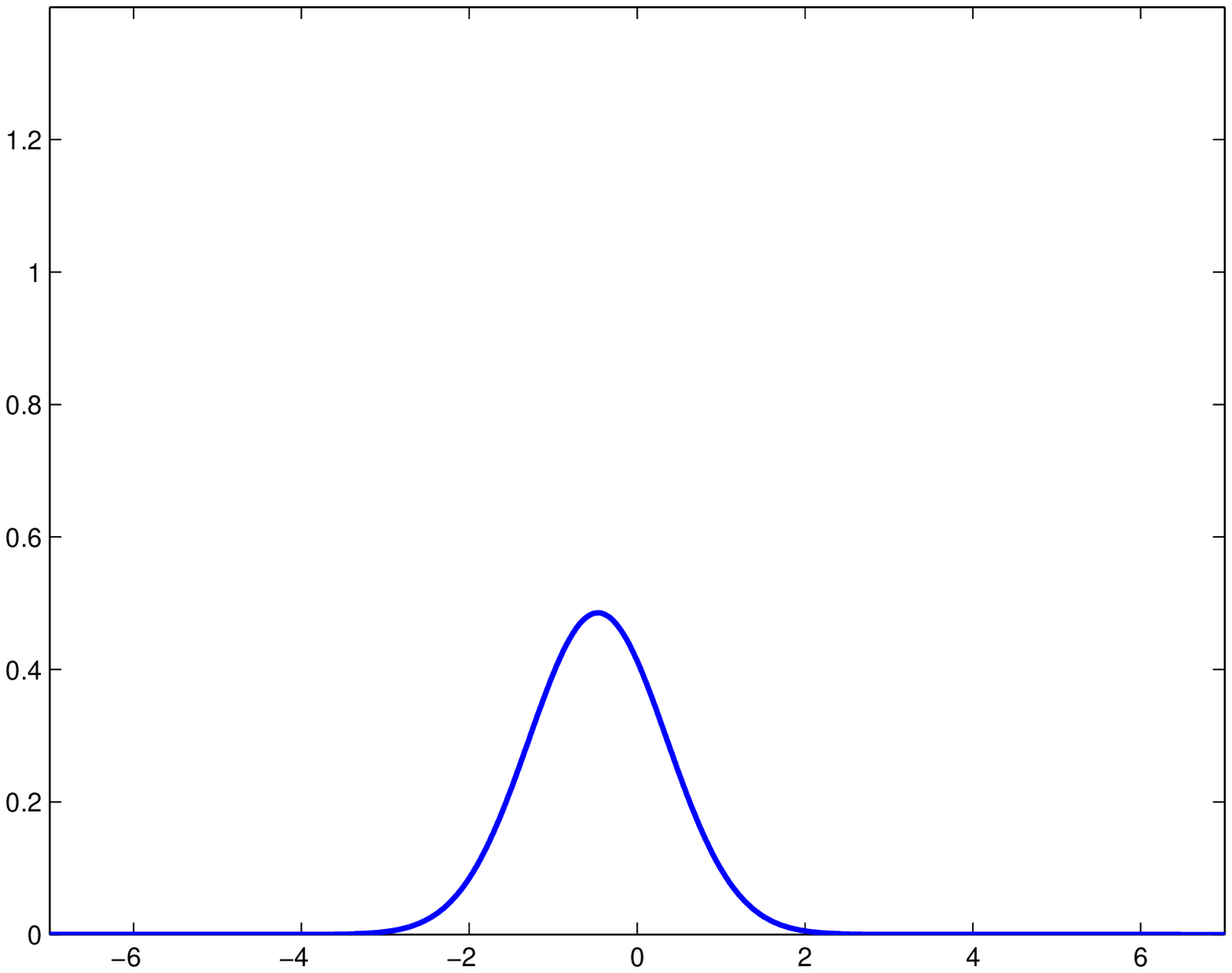} }

\subfigure[$\tilde{g}_{0.5}$]{ \includegraphics[width=3cm]{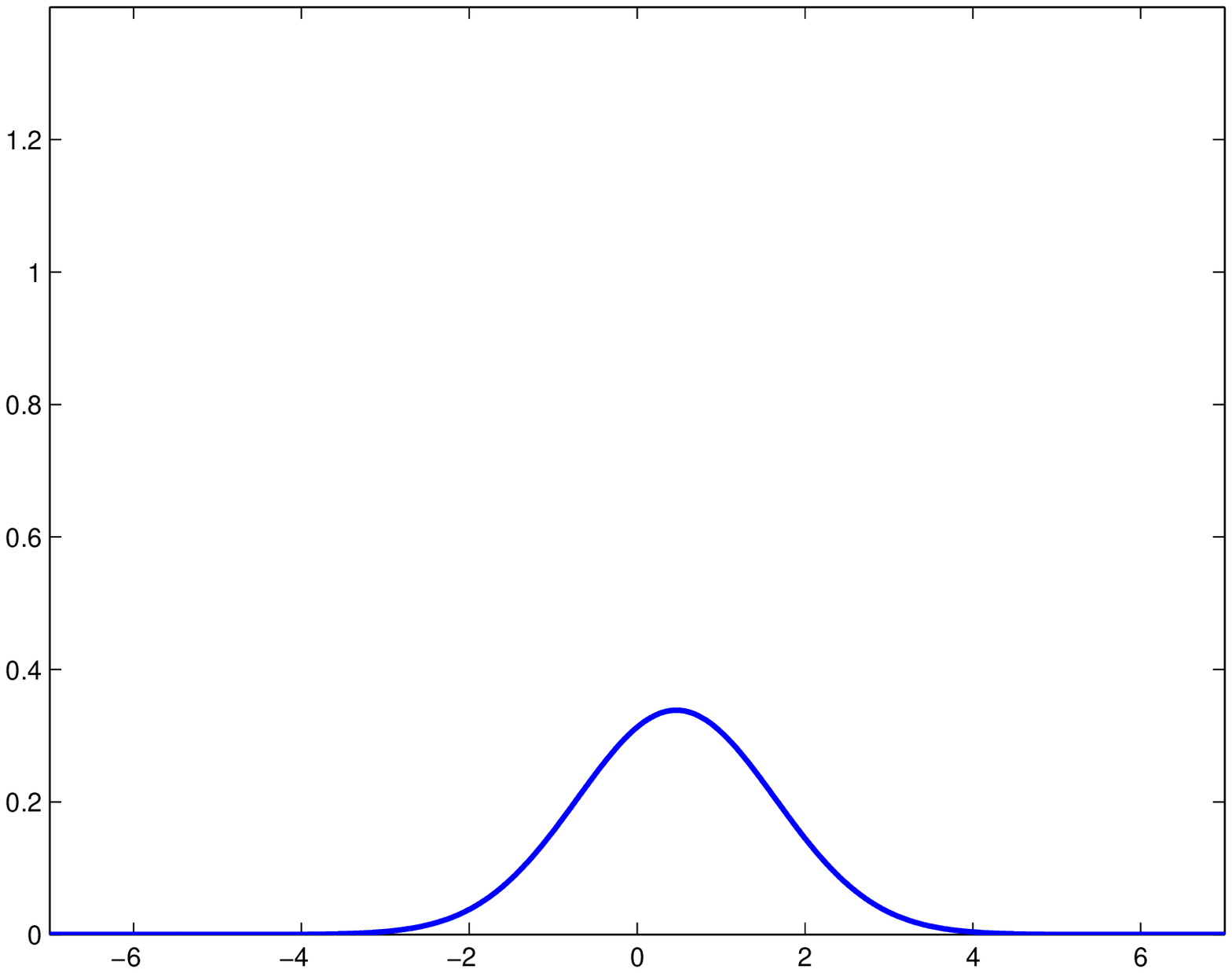} }
\subfigure[$\tilde{g}_{1}$]{ \includegraphics[width=3cm]{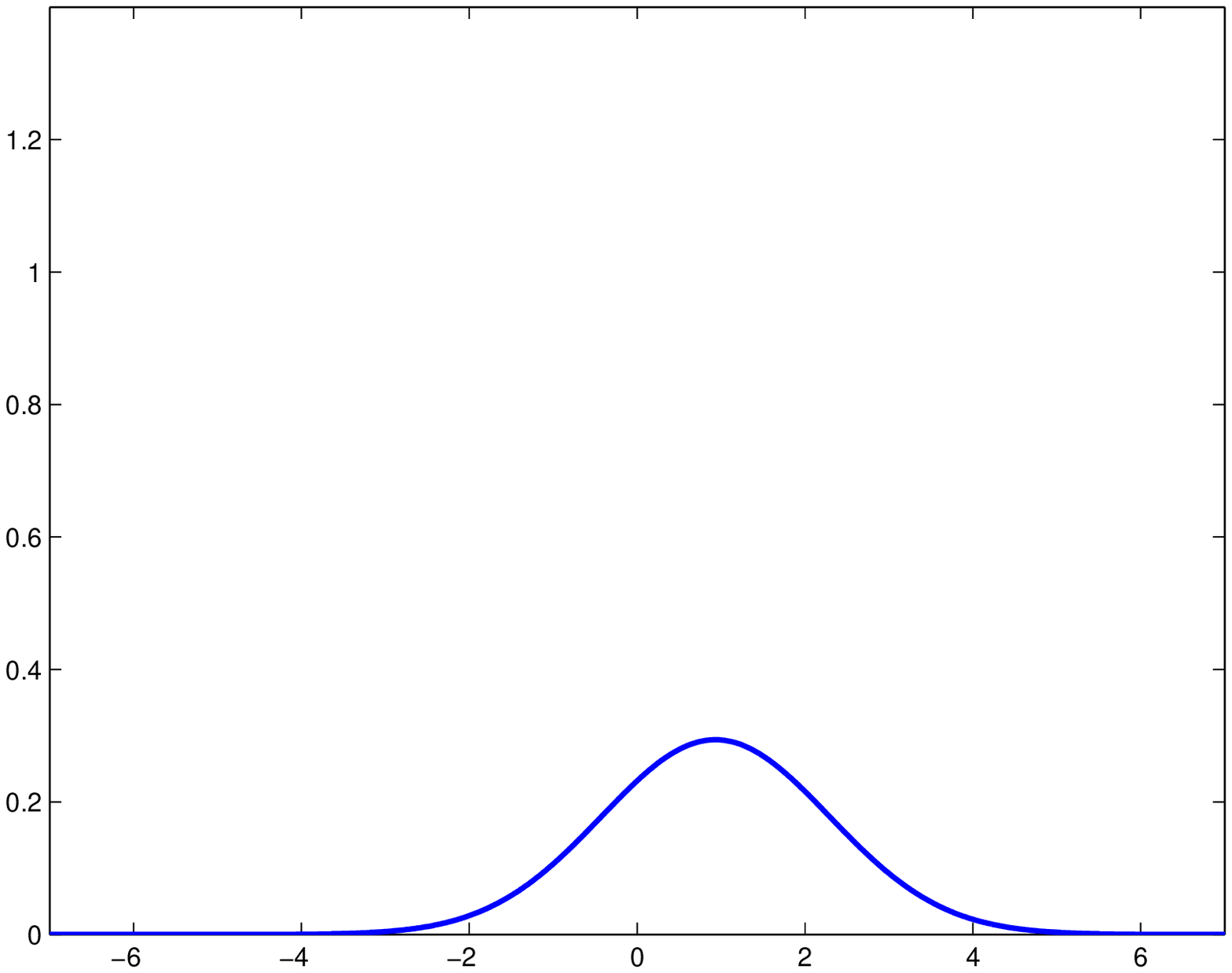} }
\subfigure[$\tilde{g}_{2}$]{ \includegraphics[width=3cm]{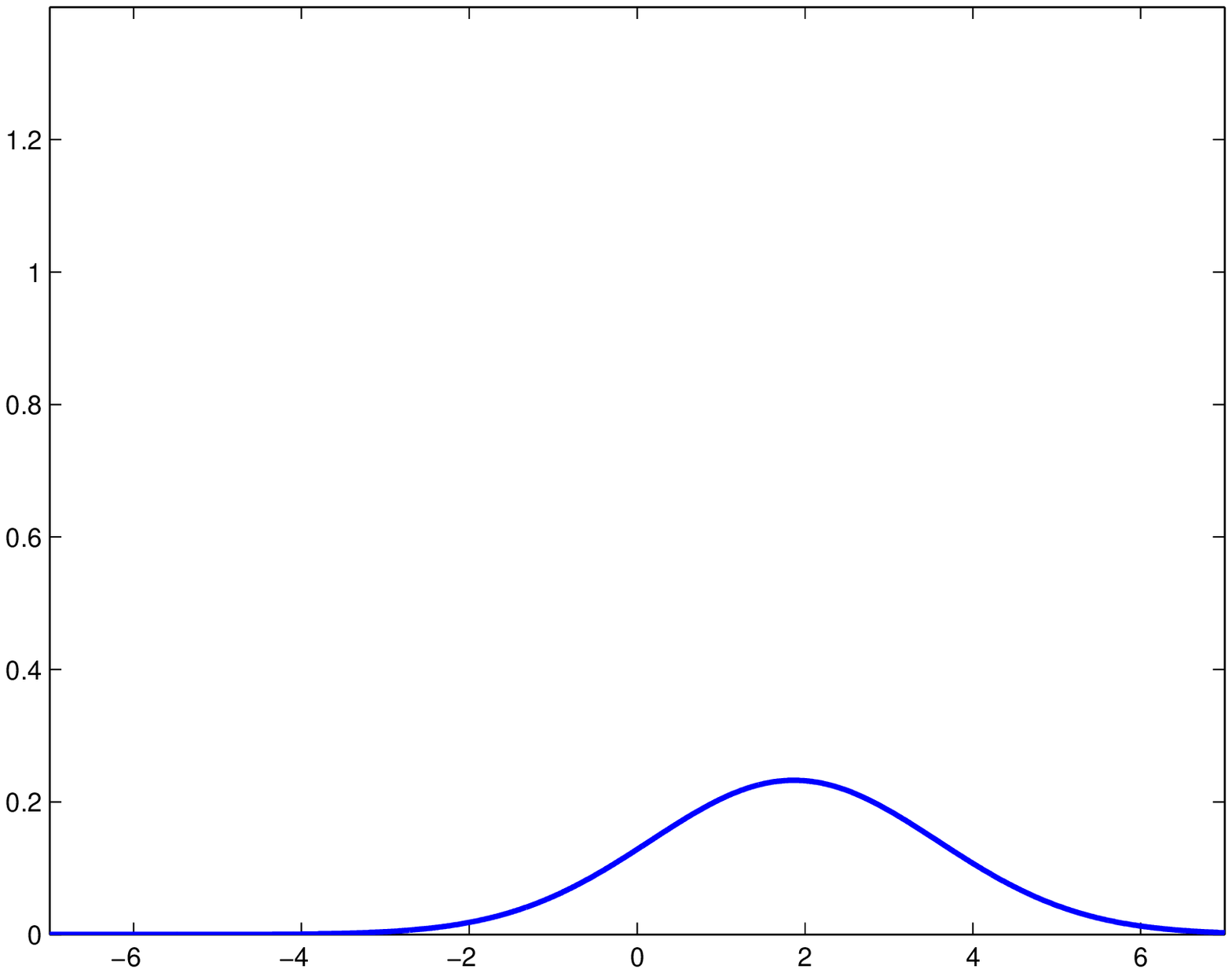} }

\caption{An example of GPCA of densities. First principal mode of geodesic variation $\tilde{g}_{t}$ in $W_2(\R)$, for $-2 \leq t \leq 2$,  of the densities displayed in Figure \ref{fig:exintro}; see \eqref{eq:gtilde}.}  \label{fig:geoPC}
\end{figure}

Our approach shares similarities with analogs of PCA for data belonging to a Riemannian manifold. There  is currently a growing interest in the statistical literature on the development of nonlinear analogs of PCA, for the analysis of data belonging to curved Riemannian manifolds; see e.g.\ \cite{geodesicPCA,citeulike:10314363,exactPGA}  and references therein. These methods, generally referred to as Principal Geodesic Analysis (PGA), extend the notion of classical PCA in Hilbert spaces. Nevertheless, as the Wasserstein space is not a Riemannian manifold, existing methods to perform a PGA cannot be directly applied to the setting of this paper.

The key property that we use to develop a notion of GPCA in the Wasserstein space is the isometry between $\WS$ and a closed convex subset of the Hilbert space of square-integrable functions $\LL$, with respect to an appropriate measure $\mu$; see Theorem \ref{theo:exp}. In this paper we thus  consider the statement of the general problem of PCA in a closed convex subset of a Hilbert space, which not only serves as basis for the analysis  of GPCA in $W_{2}(\Omega)$, but may also have interest in its own right, for further developments. For example, the notion of convex PCA introduced in this paper could be of interest, when probability distributions are characterized by observed parameters, belonging to some convex subset of an Euclidean space.

Throughout the paper, various notions from Riemannian geometry such as geodesic, tangent space, exponential and logarithmic maps, are used to illustrate the connection between our approach and PGA. However, the important issue here is not the geometry of $W_{2}(\Omega)$ but rather the use of these notions to state precisely the isometry between $W_{2}(\Omega)$ and a closed convex set of $\LL$. The GPCA in the Wasserstein space is then an application of these results. \\

The rest of the paper is organized as follows. In   Section \ref{sec:preliminaries}, we present the isometry between $\WS$ and a closed convex subset of $\LL$. We also recall basic definitions such as tangent space, geodesic, exponential and logarithmic maps in the Wasserstein space framework, having their analogs in the Riemannian setting. Section \ref{sec:cpca} is devoted to the definition and analysis of Convex PCA (CPCA) in a general framework. The main results on GPCA are gathered in Section \ref{sec:GPCA}. In Section \ref{sec:numerical} we describe some numerical aspects of GPCA on simulated examples, using simple statistical models. We also analyze a real dataset of population pyramids of 223 countries, for the year 2000. Section \ref{sec:consitency} is dedicated to the consistency of the empirical CPCA and GPCA, as the number of random data points tends to infinity.  We conclude the paper in Section \ref{sec:conclusion}, discussing the differences  between GPCA and existing PGA methods on Riemannian manifolds. We also mention potential extensions of this work. Finally, to make the paper self-contained, we collect in the Appendix some technical results about quantiles, geodesic spaces, Kuratowski convergence and $\Gamma$-convergence.

\section{Convexity of the Wasserstein space $\WS$ up to an isometry}\label{sec:preliminaries}

\subsection{The pseudo-Riemannian structure of $\WS$}\label{sec:Riemannian}

Let $\Omega$ be either the real line $\R$ or a closed interval of $\R$ and let $\WS$ be the set of probability measures over $(\Omega,{\cal B}(\Omega))$, with finite second moment, where ${\cal B}(\Omega)$ is the $\sigma$-algebra of Borel subsets of $\Omega$. 
For $\nu \in \WS$ and $T : \Omega \to \Omega$ (always assumed measurable),  we recall that the  push-forward measure $ T  \# \nu $ is defined by $( T  \# \nu  )(A) =  \nu\{x\in\Omega|T(x)\in A\}$,  for $A \in {\cal B}(\Omega)$. The cumulative distribution function (cdf) and the quantile function of $\nu$ are denoted respectively by $F_\nu$ and $F_\nu^{-}$; see Definition \ref{def:quantile}. If $\nu$ is absolutely continuous (a.c.), its density is denoted by $f_\nu$.

\begin{defin}
The quadratic Wasserstein distance $d_W$ in $\WS$ is defined by
\begin{equation*}
d_W^2(\nu_1,\nu_2) := \inf_{\pi \in {\Pi}(\nu_1,\nu_2)} \int |x-y|^2 d \pi(x,y),\; \nu_1, \nu_2\in\WS,
\end{equation*}
where $\Pi(\nu_1,\nu_2)$ is the set of probability measures on $\Omega \times \Omega$, with marginals $\nu_1$ and $\nu_2$.
\end{defin}
It can be shown that $\WS$ endowed with $d_W$ is a metric space, usually called Wasserstein space. For a detailed analysis of $\WS$, we refer to \cite{villani-topics}. In particular, the following formula, from Theorem 2.18 in \cite{villani-topics}, is important in the sequel:

\begin{equation}\label{eq:wdist1}
d_W^2(\nu_1,\nu_2) = \int_{0}^1 (F_{\nu_2}^{-}(y) - F_{\nu_1}^{-}(y))^2 dy.
\end{equation}
Also important is the following celebrated theorem (stated for measures on $\R^d$), from optimal transportation theory, due to Brenier \cite{Brenier91}.

\begin{theo}\label{theo:brenier}
Let $\mu, \nu \in W_2(\R^d)$ such that $\mu$ gives no mass to small sets, then
\begin{equation}\label{eq:wassestein}
d_W^2(\mu,\nu) = \inf_{T \in {\rm MP}(\mu,\nu) } \int_{\Omega} | T(x) - x |^2 d \mu(x),
\end{equation}
where ${\rm MP}(\mu,\nu)=\{ T: \R^d \to \R^d \; | \;  \nu = T  \# \mu \}$. Moreover, there exists $T^* \in  {\rm MP}(\mu,\nu)$ such that $d_W^2(\mu,\nu) = \int_{\Omega} | T^*(x) - x |^2 d \mu(x)$, characterized as the unique (up to a $\mu$-negligible set) element in ${\rm MP}(\mu,\nu)$ that can be represented, $\mu$-almost everywhere (a.e.), as the gradient of a convex function.
\end{theo}
Since we are in dimension $d=1$, $T^*$ in Theorem \ref{theo:brenier}, being the gradient of a convex function, is increasing. Observe also that $T^*$  may possibly be defined and be increasing only in a set of $\mu$ measure 1, but still $T^*\# \mu$ makes sense; see \cite{villani-topics}, page 67. Finally note that in $\R$ it suffices to assume $\mu$ atomless, that is, $F_\mu$ continuous.
Under the above stated conditions it is well known that $T^{*} = F_\nu^{-} \circ F_\mu$ and
\begin{equation}\label{eq:wdist2}
d_W^2(\mu,\nu)=\int_\Omega(F_\nu^{-} \circ F_\mu(x)-x)^2d\mu(x),
\end{equation}
with $F_\nu^-\circ F_\mu$ defined on the full $\mu$-measure set $A_\mu:=\{x\in\Omega|F_\mu(x)\in(0,1)\}$.

The $\WS$ space has a formal Riemannian structure described, for example, in \cite{ambrosio2004gradient}.  We provide some basic definitions, having their analogs in the Riemannian manifold setting.

From here onwards we consider that $\mu\in\WS$ is a reference measure, with continuous cdf $F_\mu$. Following \cite{ambrosio2004gradient}, we define the  tangent space at $\mu$ as the Hilbert space $\LL$ of real-valued, $\mu$-square-integrable functions on $\Omega$, equipped with the standard inner product $\langle\cdot,\cdot\rangle_\mu$ and norm $\| \cdot \|_{\mu}$. Furthermore, we define  the  exponential and the logarithmic maps at $\mu$, as follows.

\begin{defin} \label{def:explog}
Let ${\rm id}$ be the identity on $\Omega$. The exponential $\exp_{\mu}: \LL \to \WS$ and logarithmic $\log_{\mu}: \WS \to \LL$ maps are defined respectively as
\begin{equation}\label{eq:exp}
 \exp_{\mu}(v)=(v+{\rm id}) \# \mu   \quad \text{ and } \quad \log_{\mu}(\nu)=F_\nu^{-} \circ F_\mu - {\rm id}.
\end{equation}
\end{defin}
\begin{rem}
(a) $ \exp_{\mu}(v) \in \WS$, for any $v \in \LL$, since
\begin{equation*}
\int x^2 d \exp_{\mu}(v)(x)  = \int (x + v(x))^2 d \mu(x) \leq 2 \int x ^2 d \mu(x) + \int v^2(x) d \mu(x) < + \infty.
\end{equation*}
(b) By Theorem \ref{theo:brenier} and \eqref{eq:wdist2}, $\log_\mu(\nu)$ is unique ($\mu$-a.e.) and belongs to $\LL $ since $\| \log_{\mu}(\nu)\|^2_{\mu} = d_W^2(\mu,\nu) < + \infty$, for all $\nu \in  \WS$. But, as commented after \eqref{eq:wdist2},  $ \log_{\mu}(\nu)$ is only defined on $A_\mu$. Finally, the continuity of $F_\mu$ implies $\exp_\mu(\log_\mu(\nu))=\nu$.
\end{rem}
\begin{exe}
We illustrate the notions of exponential and logarithmic maps, using again the location-scale model. For $\mu_0 \in W_2(\R)$ a.c.\ and $(a,b)  \in (0,\infty) \times \R$, let $\nu^{(a,b)}$ be the  probability measure, with cdf and density respectively given by
\begin{equation}
F^{(a,b)}(x) := F_{\mu_0}\left( (x-b)/a\right),\quad f^{(a,b)}(x) :=  f_{\mu_0}\left( (x-b)/a\right)/a,\; x \in \R. \label{eq:locsclmodel}
\end{equation}
From \eqref{eq:exp},
$
\log_{\mu}(\nu^{(a,b)})(x) = [F^{(a,b)}]^{-} \circ F_\mu(x) - x$ and
$\log_{\mu_0}(\nu^{(a,b)})(x) = (a - 1) x + b$.
Therefore, letting $v(x) = (a - 1) x + b$, we have
\begin{equation}\label{eq:expnuab}
\exp_{\mu_0}(v) = \nu^{(a,b)}.
\end{equation}
\end{exe}
In the setting of Riemannian manifolds, the exponential map at a given point is a local homeomorphism from a neighborhood of the origin in the tangent space to the manifold. However, this is not the case for $\exp_{\mu}$ defined above, as it is possible to find two arbitrarily small functions in $\LL$, with equal exponentials, see e.g.\ \cite{ambrosio2004gradient}. On the other hand, we show that $\exp_{\mu}$ is an isometry when  restricted to a specific set of functions defined below.

\subsection{Isometry between $\WS$ and a closed convex subset of $\LL$}
We consider below the image of $\WS$ under the logarithmic map, denoted $V_\mu(\Omega)$, which is shown to be a closed convex subset of $\LL$. We also prove that $\exp_\mu$, restricted to $\VV$, is an  isometry. These are crucial properties needed to define and to compute the GPCA in $\WS$.

\begin{theo}\label{theo:exp}
The exponential map $\exp_{\mu}$ restricted to $\VV:=\log_\mu(\WS)$ is an isometric homeomorphism, with inverse $\log_{\mu}$.
\end{theo}

\begin{proof}
Let  $\nu \in \WS$ then, from Theorem \ref{theo:brenier}, $F_\nu^{-} \circ F_\mu$ is the unique $\mu$-a.e. increasing map (see Definition \ref{def:incr}), such that $(F_\nu^{-} \circ F_\mu) \# \mu = \nu$. In other words, $v := \log_{\mu}(\nu) =F_\nu^{-} \circ F_\mu - {\rm id}$ is the unique element in $\VV$ such that  $\exp_{\mu}(v) = \nu$. The isometry property follows from \eqref{eq:wdist1} because
$d_W^2(\nu_1,\nu_2)=\int_{0}^1 (F_{\nu_2}^{-}(y) - F_{\nu_1}^{-}(y))^2 dy=\|F_{\nu_1}^-\circ F_\mu-F_{\nu_2}^-\circ F_\mu\|_\mu^2=\|  \log_{\mu}(\nu_1) -  \log_{\mu}(\nu_2) \|^2_{\mu}$, for any $\nu_1,\nu_2 \in \WS$.
\end{proof}

\begin{prop}\label{prop:clcv}
The set $\VV:=\log_\mu(\WS)$ is closed and convex in $\LL$.
\end{prop}
\begin{proof} Let $(\nu_n)$ be a sequence in $\WS$, such that $\log_\mu(\nu_n)\to v \in\LL$. Then $F_{\nu_n}^-\circ F_\mu\to v+{\rm id}$ and, because $F_\mu$ is continuous, we have $F_{\nu_n}^-\to w\in L^2(0,1)$ (the space of square-integrable functions with respect to the Lebesgue measure on $(0,1)$). From Proposition \ref{prop:closedconvexquantile}, there exists $\nu\in\WS$ such that $w=F_\nu^{-1}$ a.e. and so, $F_{\nu_n}^-\circ F_\mu\to F_{\nu}^-\circ F_\mu$ in $\LL$, that is, $\log_\mu(\nu_n)\to \log_\mu(\nu)\in\VV$.
Convexity follows also from Proposition \ref{prop:closedconvexquantile} because, for $\lambda \in [0,1]$, there exists $\nu_\lambda\in \WS$ such that
$\lambda \log_\mu(\nu_1)+(1-\lambda)\log_\mu(\nu_2)=(\lambda F_{\nu_1}^-+(1-\lambda)F_{\nu_2}^-)\circ F_\mu-{\rm id}=F_{\nu_\lambda}^-\circ F_\mu-{\rm id}\in\VV.$
\end{proof}
\begin{rem}\label{rem:VV}
The space $\VV$ can be characterized as the set of functions $v \in \LL$ such that $T:={\rm id}+v$ is $\mu$-a.e. increasing (see Definition \ref{def:incr}) and that $T(x) \in \Omega$, for $x \in \Omega$.
\end{rem}

\subsection{Geodesics in $\WS$}\label{sec:Wgeo}

 A general overview of geodesics in a metric space is given in the Appendix.  In this section, we consider  the notion of geodesic in $\WS$, as given in Definition \ref{def:geodesic}. A direct consequence of Corollary \ref{coro:gamma}, Proposition \ref{prop:clcv} and Theorem \ref{theo:exp} is that geodesics in $\WS$ are exactly the image under $\exp_{\mu}$ of straight lines in $\VV$. In particular, given two measures in $\WS$, there exists a unique shortest path connecting them. This property is stated in the following lemma.

\begin{lem}\label{lem:gamma}
Let $\gamma : [0,1] \to \WS$ be a curve and let $v_0 := \log_{\mu}(\gamma(0))$, $v_1 := \log_{\mu}(\gamma(1))$. Then $\gamma$ is a geodesic if and only if $\gamma(t) = \exp_{\mu}((1-t)v_0+tv_1)$, for all $t \in [0,1]$.
\end{lem}

\begin{exe} \label{ex:geoW2}
To illustrate Lemma \ref{lem:gamma}, let us consider again the location-scale  model \eqref{eq:locsclmodel}. Then one has
$v_0(x) := \log_{\mu_0}(\nu^{(1,0)}) = 0 \mbox{ and } v_1(x) := \log_{\mu_0}(\nu^{(a,b)}) = (a - 1) x + b, \;  x \in \R$.
From Lemma \ref{lem:gamma}, the curve $\gamma : [0,1] \to \WS$, defined by
\begin{equation*}
\gamma(t) = \exp_{\mu_0}((1-t)v_0+tv_1) = \exp_{\mu_0}( t(a - 1) x + tb ) = \nu^{(a_t , b_t )} , \; t \in [0,1],
\end{equation*}
is a geodesic such that $\gamma(0)=\mu_0 = \nu^{(1,0)}$ and $\gamma(1)= \nu^{(a,b)}$, where $a_t = 1-t + ta$ and $b_t = t b$. Moreover, for each time $ t \in [0,1]$, the measure $\gamma(t)$ admits the density
\begin{equation}
f^{(a_t , b_t )}(x) = a_t^{-1} f_{0}\left( a_t^{-1}(x - b_t)\right), \; x \in \R. \label{eq:bfunt}
\end{equation}
In Figure \ref{fig:exgeodesic} we display the densities $f^{(a_t , b_t )}$ for some values of $t \in [0,1]$, with $\mu_0$ the standard Gaussian measure, $a=0.5$ and $b = 2$.
\end{exe}

\begin{figure}[h!]
\centering

\subfigure[$f^{(a_t , b_t )}$ for $t = 0$]{ \includegraphics[width=3cm]{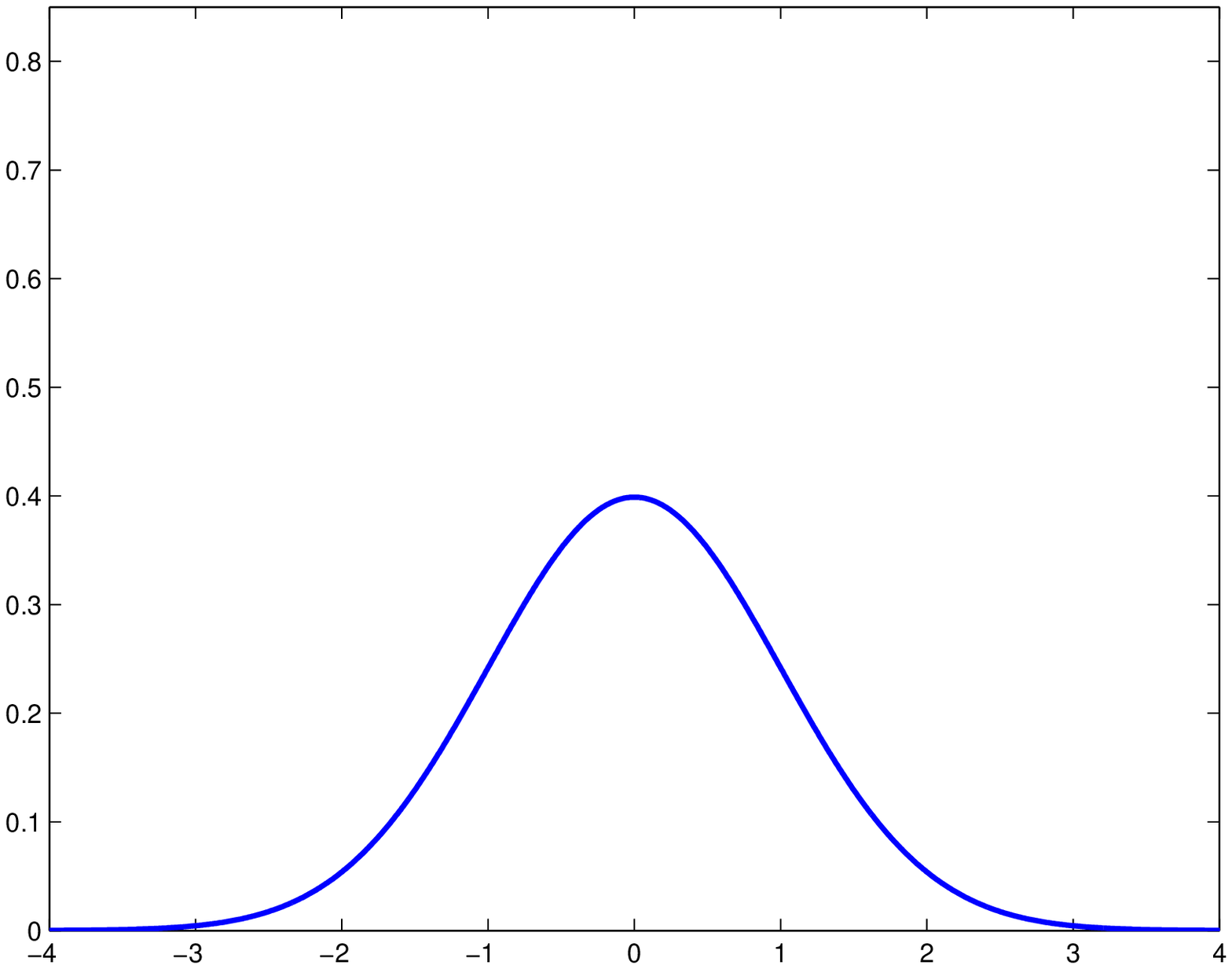} }
\subfigure[$f^{(a_t , b_t )}$ for $t = 0.2$]{ \includegraphics[width=3cm]{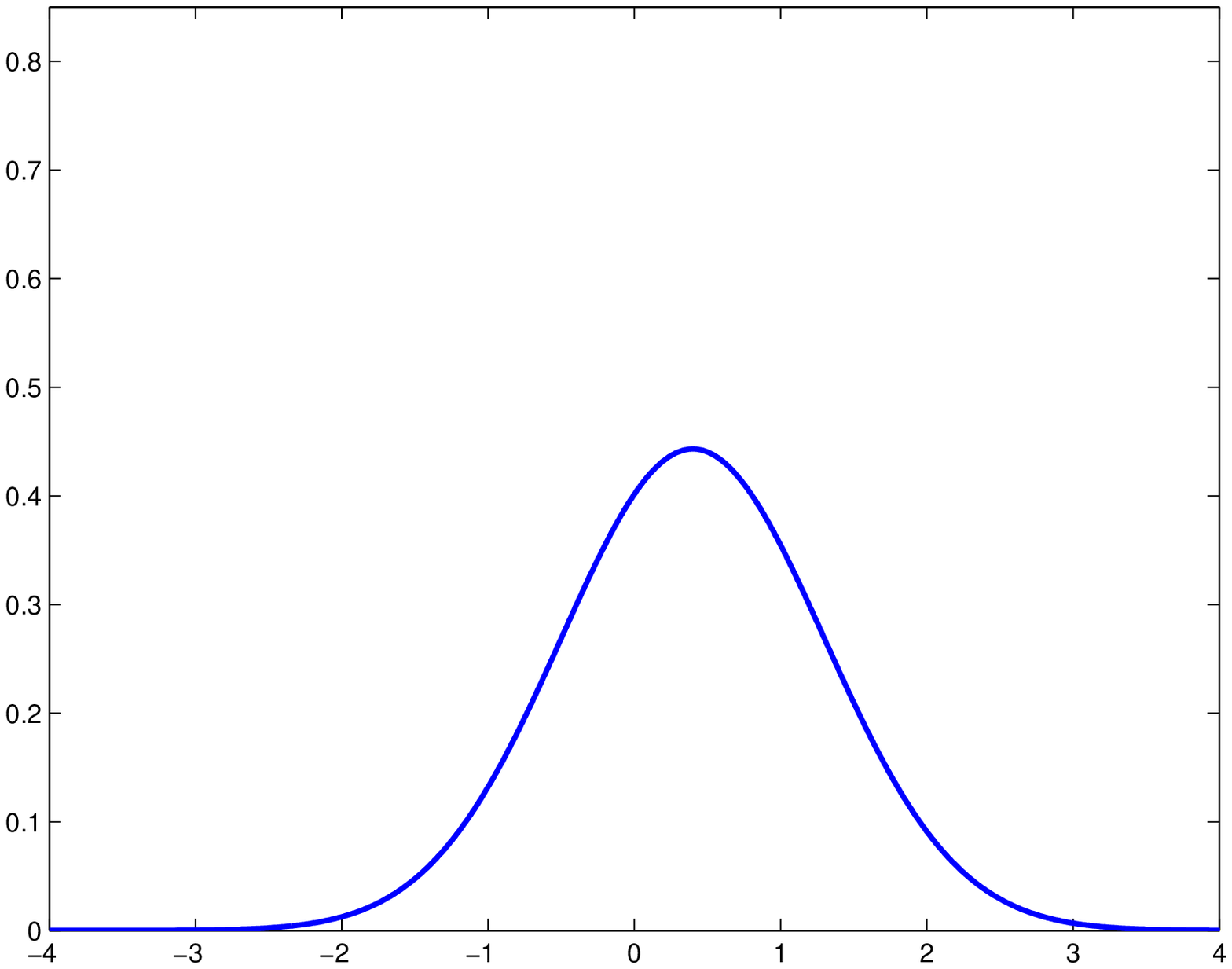} }
\subfigure[$f^{(a_t , b_t )}$ for $t = 0.4$]{ \includegraphics[width=3cm]{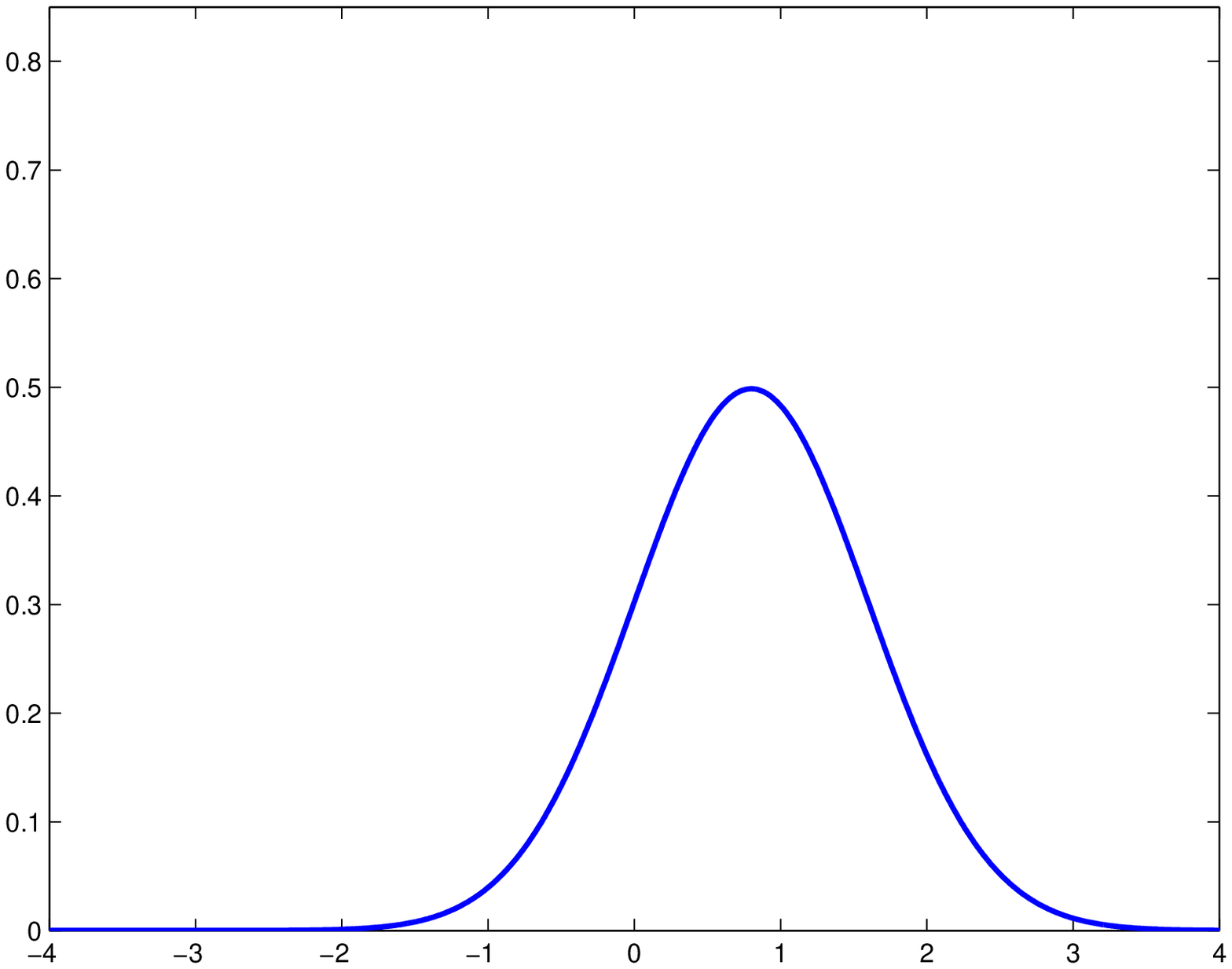} }

\subfigure[$f^{(a_t , b_t )}$ for $t = 0.6$]{ \includegraphics[width=3cm]{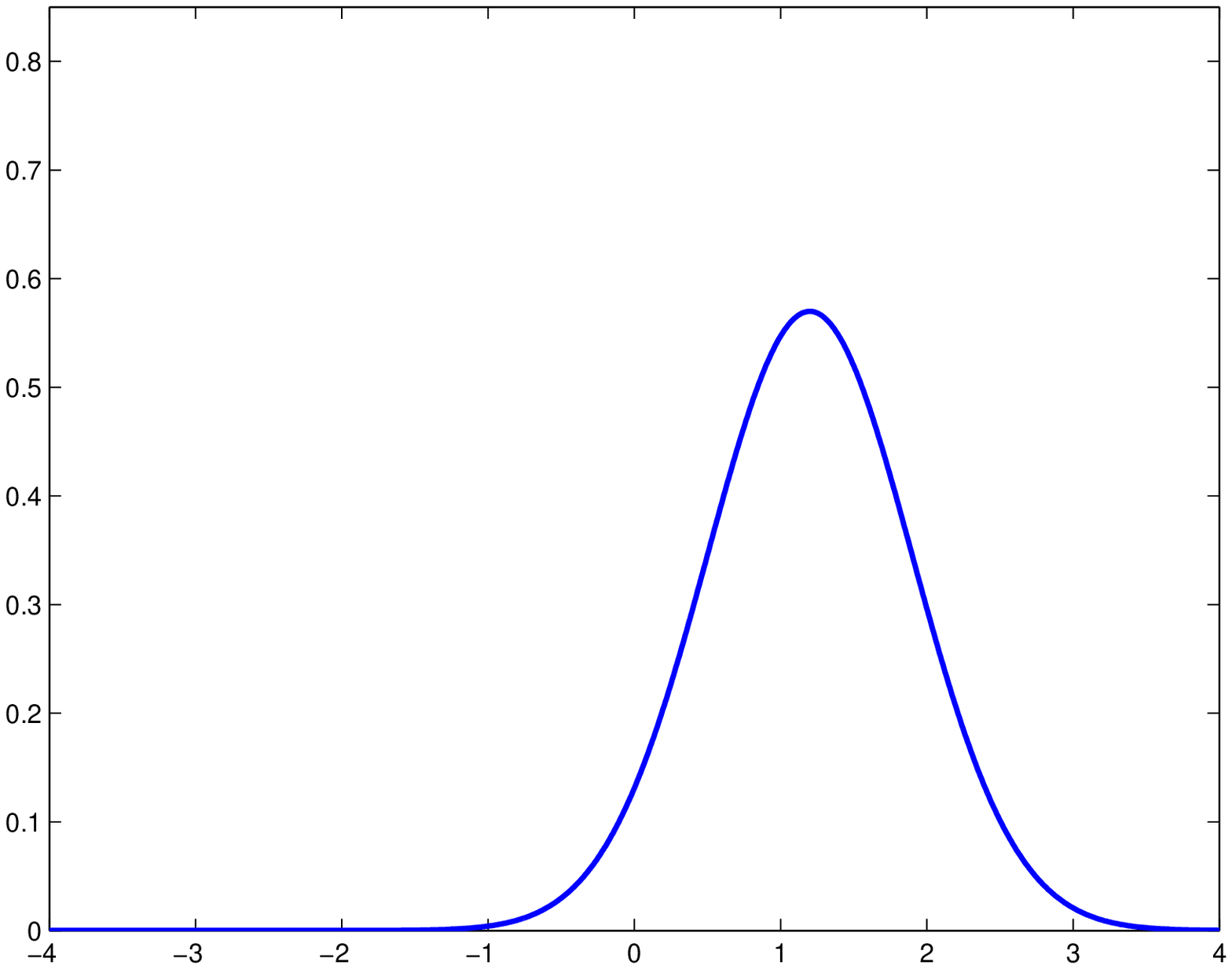} }
\subfigure[$f^{(a_t , b_t )}$ for $t = 0.8$]{ \includegraphics[width=3cm]{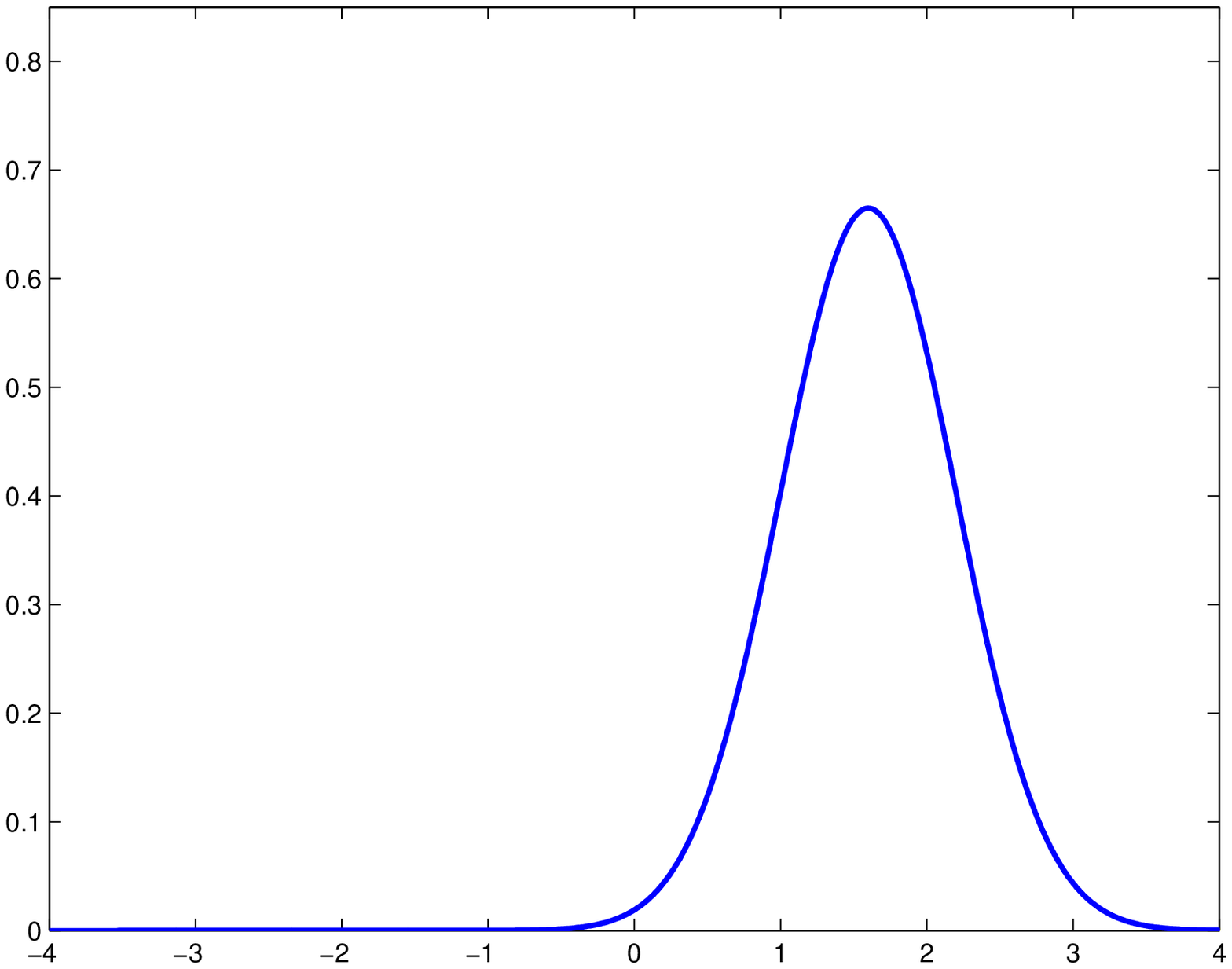} }
\subfigure[$f^{(a_t , b_t )}$ for $t = 1$]{ \includegraphics[width=3cm]{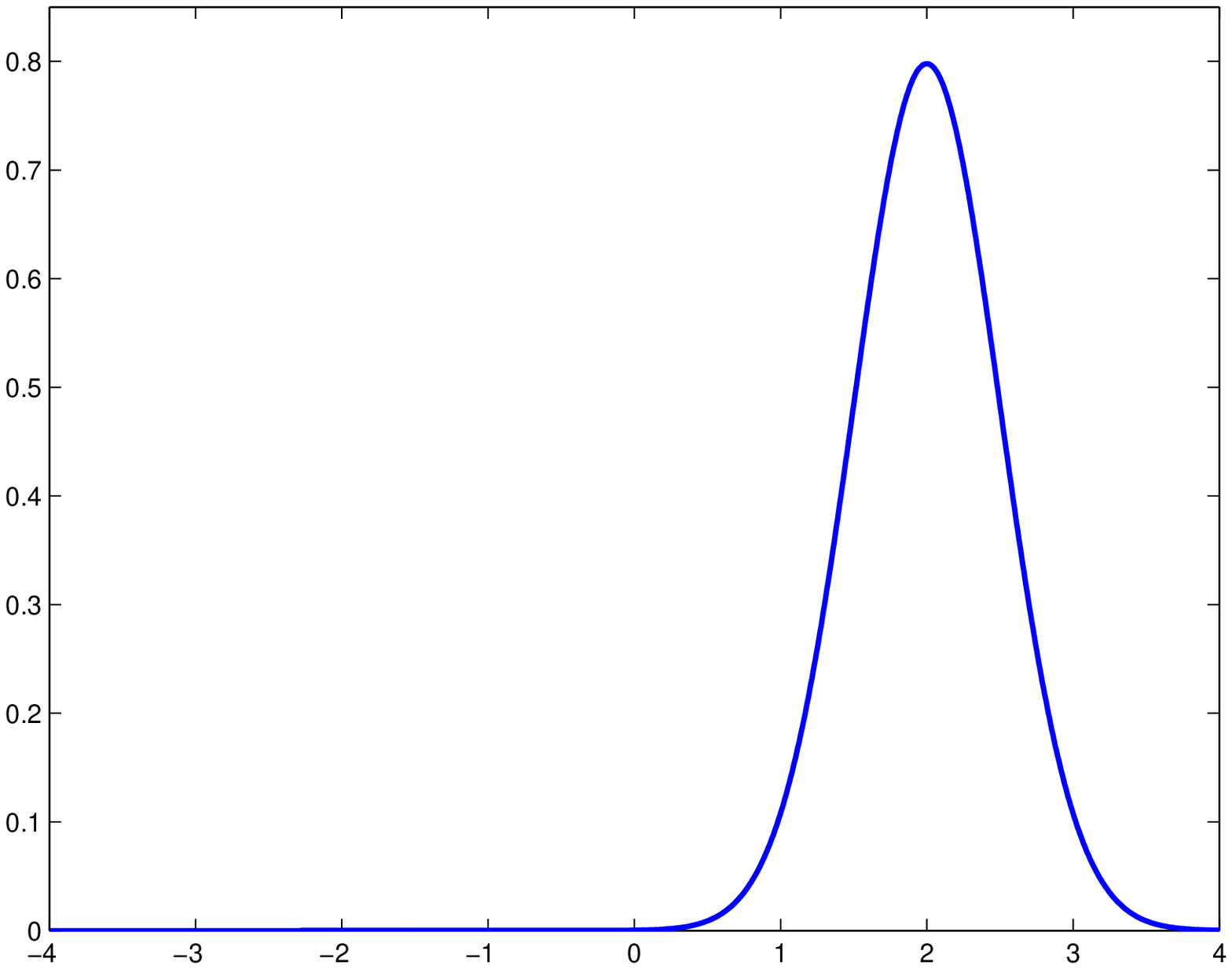} }

\caption{Visualization of the densities $f^{(a_t , b_t )}$ associated to the geodesic curve $\gamma(t) = \nu^{(a_t , b_t )}$ in $W_2$, described in Example \ref{ex:geoW2}, with $a=0.5$ and $b = 2$, in the case where $\mu=\mu_0$ is the standard Gaussian measure.
} \label{fig:exgeodesic}
\end{figure}
\begin{rem}
By  Lemma \ref{lem:gamma}, $\WS$ endowed with the Wasserstein distance $d_W$ is a geodesic space. Moreover, we have the following  corollary.
\end{rem}
\begin{coro}\label{coro:geoconvex}
A set $G  \subset \WS$ is geodesic (in the sense of Definition \ref{def:geodesicspace}) if and only if $\log_{\mu}(G)$ is convex.
\end{coro}

\begin{defin}\label{def:dimG}
Let $G\subseteq \WS$ be geodesic. The dimension of $G$, denoted $\dim(G)$, is defined as the dimension of the smallest affine subspace of $\LL$ containing $\log_{\mu}(G)$.
\end{defin}

\begin{rem}
$\dim(G)$ does not depend on the the reference measure $\mu$. Indeed, $\mu' \in \WS$ (atomless) and $E$ an affine subspace of $\LL$, such that $\log_{\mu}(G)\subseteq E$. It is easy to see that $ \log_{\mu'} \circ \exp_{\mu}:\LL \to L^2_{\mu'}(\Omega)$ is affine, therefore $\log_{\mu'} \circ \exp_{\mu}(E)$ is an affine subspace of $L^2_{\mu'}(\Omega)$ containing $\log_{\mu'}(G)$ and $\dim(E)=\dim(\log_{\mu'} \circ \exp_{\mu}(E))$. Observe also that, if $\gamma : [0,1] \to \WS$ is a geodesic, then $\gamma([0,1])$ is a geodesic space of dimension $1$.
\end{rem}

\section{Convex PCA}\label{sec:cpca}
We have shown in Section \ref{sec:preliminaries} that $\WS$ is isometric to the closed convex subset $\VV$, of the Hilbert space $\LL$. As can be seen in Section \ref{sec:GPCA}, the notion of GPCA in $\WS$ is strongly linked to a PCA constrained to $\VV$. It is then natural to develop a general strategy of convex-constrained PCA, in a general Hilbert space. This method, which we call Convex PCA (CPCA), could be applicable beyond the GPCA in $\WS$. We introduce the following  notation:
\ \\
- $H$ is a separable Hilbert space, with inner product $\langle \cdot, \cdot \rangle$ and norm $\| \cdot \|$.
\ \\
-  $d(x,y) := \|x-y\|$ and $d(x,E) := \inf_{z \in E}d(x,z)$, for $x,y \in H, E \subset H$.
\ \\
- $X$ is a closed convex subset of $H$, equipped with its Borel $\sigma$-algebra ${\cal B}(X)$.
\ \\
-   $\bx$ is an $X$-valued random element, assumed square-integrable, in the sense that $\E\|\bx \|^2< + \infty$, with expected value $\E \bx $.
\ \\
- $x_0\in X$ is a reference element and $k\ge1$ an integer.

\begin{rem}
$(\E \| \bx \|)^2\leq \E\|\bx \|^2 < + \infty$ and so, $\E \bx \in H$. It is well-known that $\E \bx$ is characterized as the unique element in $H$ satisfying $\langle \E \bx , x \rangle = \E \langle \bx , x \rangle$, for all $x \in H$, and also, as the unique element in $\arg \min_{y \in  X}  \E   d^{2}(\bx, y)$. Hence $\E \bx$ can be seen as a  natural notion of average in X.
\end{rem}

\subsection{Principal convex components}\label{sec:pcs}

\begin{defin}\label{def:costh}
For ${C} \subset X$, let $\costh(C)= \E d^2(\bx,C) $.
\end{defin}
\begin{rem}
Note that $\costh(C)$ is the  expected value of the squared residual of $\bx$ projected onto ${C}$, necessarily finite since $\bx$ is assumed square-integrable. Observe also that $\costh$ is monotone, in the sense that $\costh(C) \geq \costh(B), \mbox{ if } C \subset B$.
\end{rem}

\begin{defin} \label{def:CLC}
Let
\ \\
(a) $\CL(X)$ be the metric space  of  nonempty, closed subsets of $X$, endowed with the Hausdorff distance $h$ (see Definitions \ref{def:Haussdorff}, \ref{def:CLX}),
\ \\
(b) $\mathrm{CC}_{k}(X)$ be the family of convex sets $C\in\CL(X)$, such that  $\dim(C) \leq k$, where $\dim(C)$ is the dimension of the smallest affine subspace of $H$ containing $C$, and
\ \\
(c) $\mathrm{CC}_{x_0,k}(X)$ be the family of sets $C\in\mathrm{CC}_{k}(X)$, such that $x_0 \in C$.
\end{defin}

\begin{prop}\label{prop:cost}
If $X$ is compact, then $\costh$ is continuous on $\CL(X)$.
\end{prop}

\begin{proof}
Let  $C_n, C\in\CL(X), n\ge1$, such that $h(C_n,C)\to0$, and observe that $d^2(\bx,C_n)$ is a.s. bounded by the diameter of $X$.
Then,  by Proposition \ref{prop:dKcont} and the dominated convergence theorem, $\costh(C_n) \to \costh(C)$.
\end{proof}


\begin{prop}\label{prop:CLXcompact}
If X is compact, then $\CL(X), \mathrm{CC}_{k}(X)$ and $\mathrm{CC}_{x_0,k}(X)$ are compact.
\end{prop}

\begin{proof}
The compactness of $\CL(X)$ is proved in \cite{Price40} and \cite{Hai} and so we proceed with $\mathrm{CC}_{k}(X)$ and $\mathrm{CC}_{x_0,k}(X)$.
Let $C_n\in\mathrm{CC}_{k}(X), n\ge1$, and $C \in \CL(X)$, such that $h(C_n,C)\to0$. Then, from Blaschke's selection theorem in Banach spaces (see \cite{Price40}, \cite{Hai}), $C$ is convex.

 Let us check by contradiction that $\dim(C) \leq k$.  Assume that   that $\dim(C)>k$, then there exists linearly independent elements $x_1,\ldots,x_{k+1} \in C$  or, equivalently, with Gram determinant $\det(GM)\not=0$ (the Gram matrix $GM$ has elements $GM_{i,j}=\langle x_i,x_j \rangle$, $i,j=1,\ldots,k+1$). Observe that $h(C_n,C) \to 0$ implies that $C_n \to C$ in the sense of Kuratowski (see Remark \ref{rem:kurato}). By Definition \ref{def:Klim}(i), there exist $x_{1,n},\ldots,x_{k+1,n} \in C_n$, for every $n \geq 1$, such that $x_{j,n}\to x_j$, for $j=1,\ldots,k+1$. But as $\dim(C_n) \leq k$, the Gram determinant $\det(GM_n)$ of $x_{1,n},\ldots,x_{k+1,n}$ is zero. Also, it is easy to see that $\det(GM_n) \to \det(GM)$, which implies that $\det(GM)=0$, a contradiction. We conclude that $\mathrm{CC}_{k}(X)$ is closed, hence compact, as it is a subset of the compact space $\CL(X)$. Finally, observe that if $x_0 \in C_n$, for all $n \geq 1$, then $x_0 \in C$, by Definition \ref{def:Klim}(ii).  So $\mathrm{CC}_{x_0,k}(X)$ is also closed, thus compact.
\end{proof}

We define two notions of principal convex component (PCC), nested and global, and prove their existence. In the nested case, the definition is inductive and is motivated by the usual characterization of PCA, in terms of a nested sequence of optimal linear subspaces.

\begin{defin}\label{def:GGX}
(a) A  $(k, x_0)$-global principal convex component (GPCC) of $\bx$ is a set $C_k\in \GG_{x_0,k}(X) := \argmin_{C \in \mathrm{CC}_{x_0,k}(X)} \costh(C)$.
\ \\
(b) A $(k, x_0)$-nested principal convex component (NPCC) of $\bx$ is a set $C_k\in\NN_{x_0,k}(X):=\mkern-18mu\argmin_{C \in \mathrm{CC}_{x_0,k}(X), C \supset C_{k-1}}\mkern-18mu \costh(C), k\ge2$, with $C_1\in \GG_{x_0,1}(X)$.
\end{defin}

\begin{theo}\label{theo:GGnonemptyX}
If X is compact, then $\GG_{x_0,k}(X)$ and $\NN_{x_0,k}(X)$ are nonempty.
\end{theo}

\begin{proof}
The result for $\GG_{x_0,k}(X)$  is a direct consequence of Propositions \ref{prop:cost} and \ref{prop:CLXcompact}. We show that $\NN_{x_0,k}(X)\ne\emptyset$  by induction on $k$: first observe that $\NN_{x_0,1}(X) = \GG_{x_0,1}(X)\neq\emptyset$ and suppose that $C_{k-1}\in\NN_{x_0,k-1}(X) \neq \emptyset, k\ge2$. Furthermore, let $B_n \in \CL(X)$, such that $C_{k-1} \subset B_n, n \geq 1$, and $K$-$\lim B_n= B \in \CL(X)$ (the notation $K$-$\lim$ denotes convergence in the sense of Kuratowski, see Appendix \ref{sec:Klim} for a precise definition, where it is also recalled, that since $X$ is compact, the convergence with respect to the Hausdorff distance is equivalent to convergence in the sense of Kuratowski). It is clear that $C_{k-1} \subset B$, hence  ${\cal C}_{k-1}:=\{C \in \CL(X)\; | \;C \supset C_{k-1} \}$ is closed and so, by Proposition \ref{prop:CLXcompact}, $\{C \in \mathrm{CC}_{x_0,k}(X)\; | \;C \supset C_{k-1} \}={\cal C}_{k-1}\cap \mathrm{CC}_{x_0,k}(X)$ is closed, thus compact. Finally, Proposition \ref{prop:cost} implies $\NN_{x_0,k}(X)\neq \emptyset$.
\end{proof}

\begin{rem}\label{rem:compact}
For $k=1$ the notions of GPCC and NPCC coincide. However, this might not be the case for $k \geq 2$.
\end{rem}

\begin{defin}\label{def:bxn}
 Given $x_1, \ldots, x_n \in X$, we denote by $\bx^{(n)}$ the (square-integrable) $X$-valued random element such that $\P(\bx^{(n)} \in A)= \frac{1}{n} \sum_{i=1}^n \mathbbm{1}_{A}(x_i)$, for any $A\in\cal B(X)$,
 where $\mathbbm{1}_{A}$ is the indicator function of $A$.
\end{defin}

\begin{defin}\label{def:empPCC}
The empirical GPCC and NPCC are defined as in Definition \ref{def:GGX}, with $\bx$ replaced by $\bx^{(n)}$. The empirical version of $\costh$ is $\costhn(C) := \E  d^2(\bx^{(n)},C)  = \frac{1}{n} \sum_{i=1}^n  d^2(x_i,C)$.
\end{defin}
\subsection{Formulation of CPCA as an optimization problem in $H$}


\begin{defin}\label{def:costast}
For $\U=\{u_1,\ldots,u_k\} \subset H$, let\ \\
(a) $\spann(\U)$ be the subspace spanned by $u_1,\ldots,u_k$,
\ \\
(b) $C_{\U} = (x_0 + \spann (\U)) \cap X\in \mathrm{CC}_{x_0,k}(X)$ and
\ \\
(c) $\costH(\U) := \costh(C_{\U} )$.
\end{defin}

To simplify notations in Definition \ref{def:costast}, we write $\spann(u)$, $\costH(u)$ or $C_u$ whenever $\U=\{u\}$.
We show below that finding a GPCC can be formulated as an optimization problem in $H^k$.

\begin{prop}\label{prop:costastG}
Let $\U^{\ast}=\{ u^{\ast}_1,\ldots,u^{\ast}_k \}$ be a minimizer of $\costH$ over orthonormal sets $\U=\{ u_1,\ldots,u_k\} \subset H$, then $C_{\U^{\ast}} \in \GG_{x_0,k}(X)$.
\end{prop}

\begin{proof}
 For any $C \in \mathrm{CC}_{x_0,k}(X)$, there exists an orthonormal set $\U=\{u_1,\ldots,u_k \} \subset H$, such that $C \subset C_{\U}$. Thus, as $\costh$ is monotone, $\costh(C) \geq \costh(C_{\U}) = \costH(\U) \geq \costH(\U^{\ast}) = \costh(C_{\U^{\ast}})$, and the conclusion follows.
\end{proof}

The analogous result for NPCC is stated below. The proof, similar to that of Proposition \ref{prop:costastG}, is omitted.

\begin{prop}\label{prop:costastN}
Let $u^{\ast}_1,\ldots,u^{\ast}_k \in H$ such that $u^{\ast}_1 \in \argmin_{ u \in H, \|u\|=1} \costH(u)$ and, for $j=2,\ldots,k$,
let $u^{\ast}_j \in\mkern-18mu \argmin_{u \in \spann(u^{\ast}_1,\ldots,u^{\ast}_{j-1})^{\perp}, \|u\|=1}\mkern-18mu  \costH(u)$, where $\perp$ denotes orthogonal. Then $C_{\{u^{\ast}_1,\ldots,u^{\ast}_k\}} \in \NN_{x_0,k}(X)$.

\end{prop}
\begin{rem}\label{rem:costhastn}
The empirical version of $\costH$ is $\costHn(\U) := \costhn(C_{\U}) = \frac{1}{n} \sum_{i=1}^n d^2(x_i,C_{\U})$. A minimizer $\U^{\ast}=\{u_1^{\ast},\ldots,u_k^{\ast}\}$ of $\costHn(\U)$ leads to the construction of the empirical GPCC.
\end{rem}
In the following proposition we give a sufficient condition for the standard PCA on $H$ to be a solution of the CPCA problem. For the sake of simplicity, we  state the result only for GPCC. Given $x \in H$ and $C$ a closed convex subset of $H$, we denote by  $\Pi_{_C}x$ the projection of $x$ onto $C$.

\begin{prop}\label{prop:PCA2PGA}
Let  $\tilde{\U}=\{\tilde{u}_1,\ldots,\tilde{u}_k\} \subset  H$ be a set of orthonormal eigenvectors, associated to the $k$ largest eigenvalues of the covariance operator $K y = \E \langle \bx - x_0, y \rangle (\bx-x_0), \; y \in H.$ If $\Pi_{x_0+\spann(\tilde{\U})} \bx \in X \text{ a.s.}$, then $C_{\tilde{\U}} \in \GG_{x_0,k}(X)$.
\end{prop}
\begin{proof}
It is well known that $\tilde{\U}$ is minimizer of $\costh(x_0+\spann(\U))= \E \left( \| \bx - \Pi_{x_0+\spann (\U)}\bx \|^2 \right)$, over orthonormal sets $\U=\{ u_1,\ldots,u_k \} \subset H$. Further, $\costH(\U) =  \E \left( \| \bx - \Pi_{(x_0+\spann (\U)) \cap X}\bx \|^2 \right)$ and, since by hypothesis $\Pi_{x_0+\spann(\tilde{\U})} \bx  \in  X$, we have $\costH(\tilde{\U}) =\costh(x_0+\spann(\tilde{\U}))$.

Also, the monotonicity of $\costh$ implies
$\costh(x_0+\spann(\U))\le\costh((x_0+\spann(\U))\cap X)=\costH (\U)$. Finally, from the relations above, we get
\begin{equation*}\costH(\tilde{\U})=\costh(x_0+\spann(\tilde{\U}))\le \costh(x_0+\spann(\U))\le\costH (\U),\end{equation*}
which means that $\tilde{\U}$ is a minimizer of $\costH(\U)$ over orthonormal sets $\U=\{ u_1,\ldots,u_k \} \subset H$. Finally, from Proposition \ref{prop:costastG} we obtain the result.
\end{proof}

\begin{rem}\label{rem:concentrated}
(a) We can informally say that, if the data is sufficiently concentrated around the reference element $x_0$, then the CPCA in $X$ is simply obtained from the standard PCA in $H$. In particular, if there exists a ball $B(x_0,r)$, with center $x_0$ and radius $r>0$, such that $\bx\in B(x_0,r) \subset X$, a.s., then the hypothesis of Proposition \ref{prop:PCA2PGA} is satisfied. Indeed, $\| \Pi_{x_0+\spann(\tilde{\U})} \bx - x_0 \|=\| \Pi_{x_0+\spann(\tilde{\U})} (\bx-x_0) \| \leq \| \bx-x_0 \| \leq r$ and so $\Pi_{x_0+\spann(\tilde{\U})} \bx\in X$.

(b) The previous condition of data concentration is quite strong. However, obtaining weaker conditions ensuring $\Pi_{x_0+\spann(\tilde{\U})} \bx \in X \text{ a.s.}$, seems to be a difficult problem.

(c) If we replace $\bx$ by $\bx^{(n)}$, we obtain the empirical version of Proposition \ref{prop:PCA2PGA}. In this case, if  $\tilde{\U}=\{\tilde{u}_1,\ldots,\tilde{u}_k\} \subset  H$ are orthonormal eigenvectors associated to the $k$ largest eigenvalues of the empirical covariance operator $K y = \frac{1}{n}\sum_{i=1}^n \langle x_i - x_0, y \rangle (x_i-x_0), \; y \in H,$
and if $\Pi_{x_0+\spann(\tilde{\U})}x_i \in X$, for $i=1,\ldots,n$,  then $G_{\tilde{\U}}$ is an empirical GPCC.

(d) In this section we have used an arbitrary reference element $x_0\in X$. However, a natural choice for $x_0$ would be $\E \bx$ or $\bar{x}_n:= \E \bx^{(n)}$, in the empirical case.
\end{rem}

\section{Geodesic PCA}\label{sec:GPCA}

We consider $\WS$ equipped with the Borel $\sigma$-algebra ${\cal B}(W)$, relative to the Wasserstein metric. Also, $\bnu$ denotes a $\WS$-valued random element, assumed square-integrable in the sense that $\E d_W^2(\bnu,\lambda)< + \infty$, for some (thus for all) $\lambda\in \WS$. As in Section \ref{sec:preliminaries}, we assume that $\mu\in\WS$ is atomless, thus $F_\mu$ is continuous.

\subsection{Fr\'echet Mean}

A natural notion of average in $\WS$ is the Fr\'echet mean, studied in \cite{BK12} in a more general setting. 
In what follows we define and give some  properties of the population Fr\'echet mean $\nu^*$ of $\bnu$. Our results are stated in dimension one, that is, in $W_2(\R)$. The higher dimensional case is more involved and we refer to \cite{MR2801182,BK12} for further details.

Observe that if $\bu$ is a $\LL$-valued random element, such that $\E\|\bu\|_{\mu}< + \infty$, then its expectation $\E\bu$ is given by $(\E\bu)(x) = \E\bu(x)$, for all $x \in \R$. Clearly $\|\E\bu\|_{\mu} \leq \E\| \bu \|_{\mu}<\infty$, hence $\E\bu \in \LL$. Also, if $\P(\bu \in \VV)=1$, then $\E\bu \in \VV$.
\begin{prop}\label{prop:frechetmean}\ \\
(i) There exists a unique $\nu^* \in {\cal W}:= \arg \min_{\nu \in  \WS}  \E d_W^{2}(\bnu, \nu)$, called the Fr\'echet mean of $\bnu$.
\ \\
(ii) $\nu^* = \exp_{\mu}(\E\bv)$, where  $\bv=\log_{\mu}(\bnu)$.
\ \\
(iii) $F_{\nu^*}^-=\E (F_{_{\bnu}}^-)$, where $F_{_{\bnu}}$ is the (random) cdf of $\bnu$.
\ \\
(iv) If $F_{\bnu}$ is continuous a.s., then $F_{\nu^*}$ is continuous.
\end{prop}

\begin{proof} ({\it i, ii}) Let ${\cal L}=\argmin_{u \in  \LL} \E \| \bv-u \|^2_{\mu}$, ${\cal V}=\argmin_{u \in  \VV} \E \| \bv-u \|^2_{\mu}$. From Theorem \ref{theo:exp}, $\E\| \bv - u \|^2_{\mu}= \E d_W^2(\bnu,\exp_{\mu}(u))$, for all $u \in \LL$. Therefore
$\inf_{\nu \in  \WS} \E  d_W^{2}(\bnu,\nu) = \inf_{u \in  \VV} \E \| \bv-u \|^2_{\mu}$, and $u^*\in{\cal V}$ if and only if
$\exp_\nu(u^*)\in{\cal W}$.

On the other hand, $\E\bv\in\VV$ is the unique element of  ${\cal L}$, hence the unique element in ${\cal V}$. Thus, by Theorem \ref{theo:exp}, $\nu^* = \exp_{\mu}(\E\bv)$ is the unique element of ${\cal W}$.
\ \\
({\it iii}) From ({\it ii}) and  \eqref{eq:exp}, we have the chain of equalities
$F_{\nu^*}^{-}\circ F_{\mu}=\log_\mu(\nu^*)+{\rm id}=\E \bv+{\rm id}=\E( \bv+{\rm id})=\E(\log_\mu(\bnu)+{\rm id})=\E(F_{_{\bnu}}^-\circ F_{\mu})=\E(F_{_{\bnu}}^-)\circ F_{\mu}$, which implies $F_{\nu^*}^{-}=\E(F_{_{\bnu}}^-)$ because $F_\mu$ is continuous.
\ \\
({\it iv}) Observe that $F_\nu$ is continuous if and only if $F_\nu^-$ is strictly increasing. So, if
$F_{\bnu}^-(y)-F_{\bnu}^-(x)>0$ a.s., for $x<y$, then $F_{\nu^*}^-(y)-F_{\nu^*}^-(x)=\E(F_{\bnu}^-(y)-F_{\bnu}^-(x))>0$, that is, $F_{\nu^*}$ is continuous.
\end{proof}
\begin{rem}
It is interesting to see, from Proposition \ref{prop:frechetmean}(ii), that $\exp_{\mu}(\E(\log_{\mu}(\bnu)))$ does not depend on $\mu$.
\end{rem}

\subsection{Principal geodesics}\label{sec:pgs}

In this section we present definitions and results similar to those of Section \ref{sec:cpca}; $k$ denotes a positive integer and $\nu_0\in\WS$ is a reference measure.

\begin{defin}\label{def:cost}
For $\nu \in \WS$, $G \subset \WS$, let $d_W(\nu,{G}) = \inf_{\lambda \in {G}}d_W(\nu,\lambda)$ and $\cost(G) := \E d_W^2(\bnu,G)$.
\end{defin}

\begin{defin}\label{def:CL}
Let
\ \\
(a) $\CL(W)$ be the metric space  of  nonempty, closed subsets of $\WS$, endowed with the Hausdorff distance $h_{W_2}$, and
\ \\
(b) $\mathrm{CG}_{\nu_0,k}(W) = \left\{  G \in \CL(W) \; | \; \nu_0 \in G, \; \mbox{$G$ is a geodesic set and $\dim(G) \leq k$}   \right\},\; k\ge1$.
\end{defin}

The notions of  global and nested principal  geodesics of $\bnu$ with respect to $\nu_0$, are presented below, followed by the main existence result. In the case of the nested geodesics, the definition is inductive. The proof depends on the relation between GPCA and CPCA in $\VV$.

\begin{defin}\label{def:GG}
(a) A $(k,\nu_0)$-global principal geodesic (GPG) of $\bnu$ is a set $G_k\in\GG_{\nu_0,k}(W) := \argmin_{G \in \mathrm{CG}_{\nu_0,k}(W)} \cost(G)$.
\ \\
(b) A $(k,\nu_0)$-nested principal geodesic  (NPG) of $\bnu$ is a set $G_k\in\mkern-18mu\argmin_{G \in \mathrm{CG}_{\nu_0,k}(W), G \supset G_{k-1}}\mkern-18mu \cost(G),k\ge2$, with $G_1\in\GG_{\nu_0,k}(W)$.
\end{defin}

\begin{theo}\label{theo:GGnonempty}
If $\Omega$ is compact, then $\GG_{\nu_0,k}(W)$ and $\NN_{\nu_0,k}(W)$ are nonempty.
\end{theo}
\begin{proof}
As $\Omega$ is compact, $\WS$ is also compact (see \cite{villani-topics}) and so is $\VV$, by Theorem  \ref{theo:exp}. Then, Theorem \ref{theo:GGnonemptyX} and Propositions \ref{prop:CPCA2GPCA_G}, \ref{prop:CPCA2GPCA_N} ensure the existence of GPG and NPG.
\end{proof}

\begin{rem}
As commented in Remark \ref{rem:compact} for CPCA, the notions of GPG and NPG are not equivalent, except obviously for $k=1$. 
\end{rem}

\subsection{Empirical Fr\'echet mean and principal geodesics}\label{sec:empPGA}

\begin{defin}\label{def:bnun}
Given $\nu_1, \ldots, \nu_n \in \WS$, we denote by $\bnu^{(n)}$ the $\WS$-valued random element, such that $\P(\bnu^{(n)} \in A)= \frac{1}{n} \sum_{i=1}^n \mathbbm{1}_{A}(\nu_i)$, for any $A \in{\cal B}(W)$.
\end{defin}

\begin{defin}
The empirical Fr\'echet mean of $\nu_1, \ldots, \nu_n \in \WS$, denoted by $\nu_n^*$, is defined, following Proposition \ref{prop:frechetmean}, as the Fr\'echet mean of $\bnu^{(n)}$ defined above. Equivalently, $\nu_n^*$ is the unique element of $\argmin_{\nu \in  \WS} \frac{1}{n}\sum_{i=1}^n d_W^{2}(\nu_i,\nu)$.
\end{defin}

\begin{prop}\label{prop:empfrechetmean}
Let $\nu_1, \ldots, \nu_n \in \WS$. Then, the following formula holds
\begin{equation}\label{eq:qavg}
F_{\nu_n^*}^{-} = \frac{1}{n} \sum_{i=1}^n F_{\nu_i}^{-}.
\end{equation}
\end{prop}
\begin{proof}
The result is a direct consequence of Proposition \ref{prop:frechetmean}(iii).
\end{proof}
\begin{rem}
Formula \eqref{eq:qavg}  is known in statistics as quantile averaging; see \cite{Zhang2011,Gallon2011}. A detailed characterization of $\nu_n^*$  can be found in \cite{MR2801182}, for measures supported on $\R^d$, $d \geq 1$.
\end{rem}
\begin{defin}\label{def:GGn}
The empirical GPG and NPG are defined as in Definition \ref{def:GG}, with $\bnu$ replaced by $\bnu^{(n)}$.
\end{defin}
\begin{rem} (a) A natural choice for the reference measure $\nu_0$ is the Fr\'echet mean $\nu^*$, which is atomless thanks to Proposition \ref{prop:frechetmean}({\it iv}).
\ \\
(b) In the empirical case $\cost$ is given by $\costn(G) := \E d_W^2(\bnu^{(n)},G) = \frac{1}{n} \sum_{i=1}^n  d_W^2(\nu_i,G)$.
\end{rem}

\subsection{Formulation of GPCA as CPCA in $\VV$}\label{sec:GPCAequivCPCA}

Recall that geodesic sets in $\WS$ are the image under the exponential map $\exp_{\mu}$, of convex sets in $\VV$ (see Corollary \ref{coro:geoconvex}). Thus, the GPCA in $\WS$ can be formulated as a CPCA in $\VV$, as shown in  this section. CPCA is applied to $H=\LL$, $X=\VV$, $x_0=\log_{\mu}(\nu_0)$ and $\bx = \log_{\mu}(\bnu)$. In this setting $\costh(C) = \E  d_{\mu}^2(\bx,C), \; C \subset \VV$.

The following proposition shows that the search of GPG in $\WS$ is equivalent to the search of GPCC in $\VV$. The same principle applies to NPG.

\begin{prop}\label{prop:CPCA2GPCA_G}
Let $\GG_{\nu_0,k}(W)$ be the set of GPG of $\bnu$ and $\GG_{x_0,k}(\VV)$ be the set of GPCC of $\bx=\log_{\mu}(\bnu)$. Then $\GG_{\nu_0,k}(W)=\exp_{\mu}\left(\GG_{x_0,k}(\VV)\right)$.
\end{prop}

\begin{proof}
From Corollary \ref{coro:geoconvex} we have $\mathrm{CG}_{\nu_0,k}(W)=\exp_{\mu}\left(\mathrm{CC}_{x_0,k}(\VV)\right)$. On the other hand, from Theorem  \ref{theo:exp} and Definition \ref{def:cost},  $\cost(G) = \E  d_W^2(\bnu,G)  = \E  d_{\mu}^2(\bx,\log_{\mu}(G)) $. Therefore,
$\cost(G) = \costh (\log_{\mu}(G))$, for $G \subset \WS$. The result follows from Theorem \ref{theo:exp}.
\end{proof}

 \begin{prop}\label{prop:CPCA2GPCA_N}
Let $\NN_{\nu_0,k}(W)$ be the set of NPG of $\bnu$ and $\NN_{x_0,k}(\VV)$ the set of NPCC of $\bx=\log_{\mu}(\bnu)$. Then $\NN_{\nu_0,k}(W)=\exp_{\mu}\left(\NN_{x_0,k}(\VV)\right)$.
\end{prop}

\section{Numerical examples of GPCA in $W_2(\R)$}\label{sec:numerical}

In Section \ref{sec:conc} we show an example of concentrated data, such that  Proposition \ref{prop:PCA2PGA} can be applied and the problem of finding GPG is reduced to standard PCA on the logarithms; see Remark \ref{rem:concentrated}(a). In Section \ref{sec:spread} we exhibit ``spread-out data'', where the GPG cannot be obtained from standard PCA.

\subsection{Concentrated data} \label{sec:conc}

We consider the set of  probabilities $\nu_1,\ldots,\nu_4$, with  densities $f_1,\ldots,f_4$, displayed in Figure \ref{fig:exintro}. These measures satisfy the location-scale model \eqref{eq:locsclmodel}, with $\mu_0$ being the standard Gaussian measure and the values of $a_i$ and $b_i$ given in Table \ref{table:concvals}.

\begin{table}[htbp]
	\centering
		\begin{tabular}{lrrrr}
		\hline \\
		[-2ex]
		$i$ & $1$ & $2$ & $3$ & $4$  \\
		\hline
		$a_i$   & $ 0.4$   & $0.8$   & $1.2$ & $1.6$    \\
		$b_i$   & $-1.8$   & $-0.1$  & $0.7$ & $1.2$    \\
		\hline
		\end{tabular}
		\caption{Values of parameters for concentrated data.}
	\label{table:concvals}
\end{table}
The Fr\'echet mean $\nu^*_4$ of $\nu_1,\ldots,\nu_4$ is computed using the quantile average formula \eqref{eq:qavg}, from which we obtain the density $g^*_{4}$ of  $\nu^*_4$ (Figure \ref{fig:exintro}(f)), given by
\begin{equation*}
g^*_{4}(x) = f^{(\bar{a}_{4},\bar{b}_{4})}(x) = f_{\mu_0}\left( (x-\bar{b}_{4})/\bar{a}_{4}\right)/\bar{a}_{4}=f_{\mu_0}(x), \; x \in \R,
\end{equation*}
where $\bar{a}_{4} = 1$ and $\bar{b}_{4} = 0$ are the arithmetic means of the parameters $a_i, b_i$, and so, $ \nu^*_4 = \mu_0$. Observe that  the measures $\nu_1,\ldots,\nu_4$ are concentrated around their Fr\'echet mean, in the sense that their expectations and variances are not too far from those of $\nu_{4}^*$ (see Figure \ref{fig:exintro}).

We apply Propositions \ref{prop:PCA2PGA} and \ref{prop:CPCA2GPCA_G} to compute an empirical first GPG, with both $\mu$ and $\nu_0$ equal to $\mu_0$. Let $w_1$ be the eigenvector associated to the largest eigenvalue of the empirical covariance operator
$K v = \sum_{i=1}^4 \langle v_i, v\rangle v_i/4, \; v \in L^2_{\mu_0}(\R)$, where
\begin{equation*}
v_ i (x) =  \log_{\mu_0 } (\nu_i )(x) = \left(  a_i  - 1 \right) x + b_i  , \;i=1,\ldots,4; x \in \R.
\end{equation*}
 Given that the $v_{i}\in A\subseteq L^2_{\mu_0}(\R)$, the subspace of affine functions (generated by the identity and the constant function $1$, which are orthonormal in $L^2_{\mu_0}(\R)$), the operator $K$ can be identified with the $2 \times 2$ matrix $M =  \sum_{i=1}^4 V_i'V_i/4$
with $V_{i} = (  a_i  - 1 , b_i )' \in \R^2$, $1 \leq i \leq 4$. Therefore, $w_1\in A$ and $w_1(x) = \alpha_1 x + \beta_1$, where  $W_1:=(\alpha_1,\beta_1)'=(0.36,0.93)' \in \R^2$ is the eigenvector associated to the largest eigenvalue of  $M$. In other words, computing $w_1$ simply amounts to calculating the first eigenvector associated to the standard PCA of the $V_{i} \in \R^{2}$, which represent the slope and intercept  parameters of the functions $v_{i}$. In Figure \ref{fig:exlinGPCA}  we display the vectors $V_{i}$ (circles), together with the linear space spanned by $W_1$ (dash-dot line), which corresponds to the first principal direction of variation of this dataset.

Affine functions $u(x)= \alpha x + \beta$ in $V_{\mu_0}(\R)$ are represented by points $(\alpha,\beta)'\in\R^{2}$, with $\alpha \geq -1$, which is the region to the right of the vertical dashed line in Figure \ref{fig:exlinGPCA}. Hence, it can be seen from the projections of the $V_{i}$ onto the space spanned by $W_1$, that $\Pi_{\spann(w_1)} v_i  \in  V_{\mu_0}(\R) $, for $1 \leq i \leq 4$. Therefore, from Propositions \ref{prop:PCA2PGA} and \ref{prop:CPCA2GPCA_G}, the set of probability measures
\begin{equation*}
G_{1} = \left\{ \nu_{1,t}:=\exp_{\mu_0}(tw_1)\;|\; t \in \R, 1 + t \alpha_1 \geq 0 \right\},
\end{equation*}
is a first empirical GPG. From \eqref{eq:locsclmodel} and \eqref{eq:expnuab}, each  $\nu_{1,t}\in G_{1}$ admits the density
\begin{equation}
g_{1,t}(x) = f_{\mu_0}\left( (x-b_{1,t})/a_{1,t} \right)/a_{1,t}, \; x \in \R, \label{eq:gtilde}
\end{equation}
with $a_{1,t} = 1 + t \alpha_1$ and $b_{1,t} = t \beta_1$. In Figure \ref{fig:geoPC}, we display  the first principal mode of geodesic variation $g_{1,t}$, for $-2 \leq t \leq 2$, of the densities displayed in Figure \ref{fig:exintro}. As already mentioned, the GPCA  in $\WS$  gives a better  interpretation of the data variability, when compared to results from the  first principal  mode of linear variation of the densities in $L^{2}_{\mu_0}(\R)$, displayed  in Figure \ref{fig:linPC}.

\begin{figure}[h!]
\centering
{ \includegraphics[width=7cm,height=7cm]{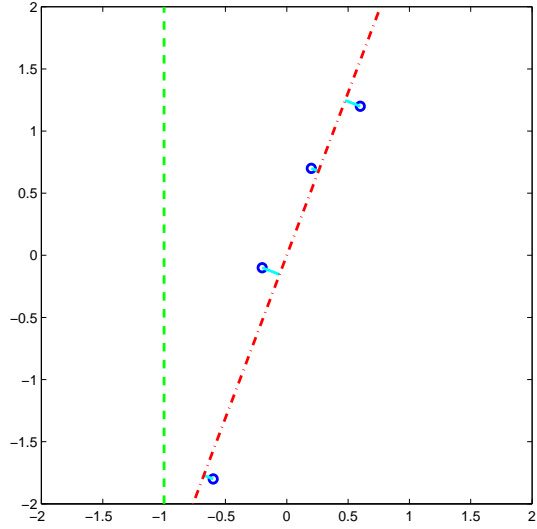} }

\caption{ Two-dimensional representation of the affine functions  $u(x)= \alpha x + \beta$ in $L^2_{\mu_0}(\R)$. The horizontal and vertical axes represent the slope  and the intercept parameters $\alpha, \beta$, respectively. Points to the right of the vertical dashed line at $\alpha = -1$, correspond to affine functions in $V_{\mu_0}(\R) $.  Circles represent the vectors $V_{i} = (  a_i  - 1 , b_i )'$, associated to the functions $v_ i (x) =  \left(  a_i  - 1 \right) x + b_i$, for $1 \leq i \leq 4$, corresponding to the measures with densities displayed in Figure \ref{fig:exintro}. The dash-dot line is the linear space spanned by the  first eigenvector $W_1$ from the standard PCA of $V_1,\ldots,V_4$.} \label{fig:exlinGPCA}
\end{figure}

\subsection{The case of spread-out data}\label{sec:spread}

We exhibit a case where standard PCA of logs in $L^2_{\mu}(\R)$ does not lead to a solution of GPCA in $\WS$. Measures $\nu_1,\ldots,\nu_4$ are as in Section \ref{sec:conc}, with parameters $a_i, b_i$ given in Table \ref{table:spreadvals}.
We have again $\bar{a}_{4}= 1$, $\bar{b}_{4} = 0$ and $\nu_4^*=\mu_0$. From Figure \ref{fig:spread}, we see that $\nu_1,\ldots,\nu_4$ are less concentrated around $\nu_4^*$, compared to the foregoing example (see Figure \ref{fig:exintro}).

\begin{table}[htbp]
	\centering
		\begin{tabular}{lrrrr}
		\hline \\
		[-2ex]
		$i$ & $1$ & $2$ & $3$ & $4$  \\
		\hline
		$a_i$   & $0.2$   & $0.2$   & $0.2$ & $3.4$    \\
		$b_i$   & $-3$   & $-1$  & $1$ & $3$    \\
		\hline
		\end{tabular}
		\caption{Values of parameters for spread-out data.}
	\label{table:spreadvals}
\end{table}

  \begin{figure}[h!]
\centering
\subfigure[]{ \includegraphics[width=3cm]{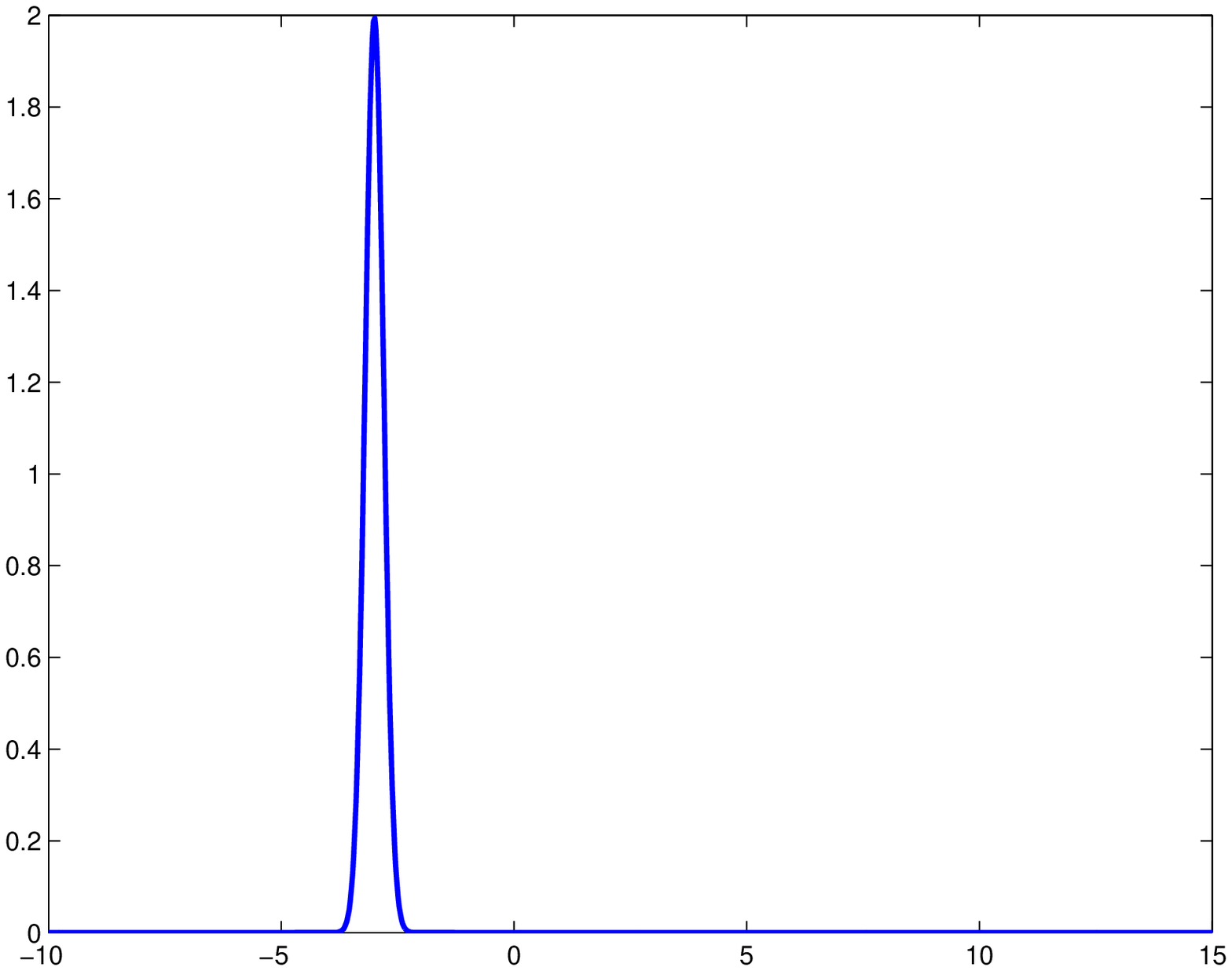} }
\subfigure[]{ \includegraphics[width=3cm]{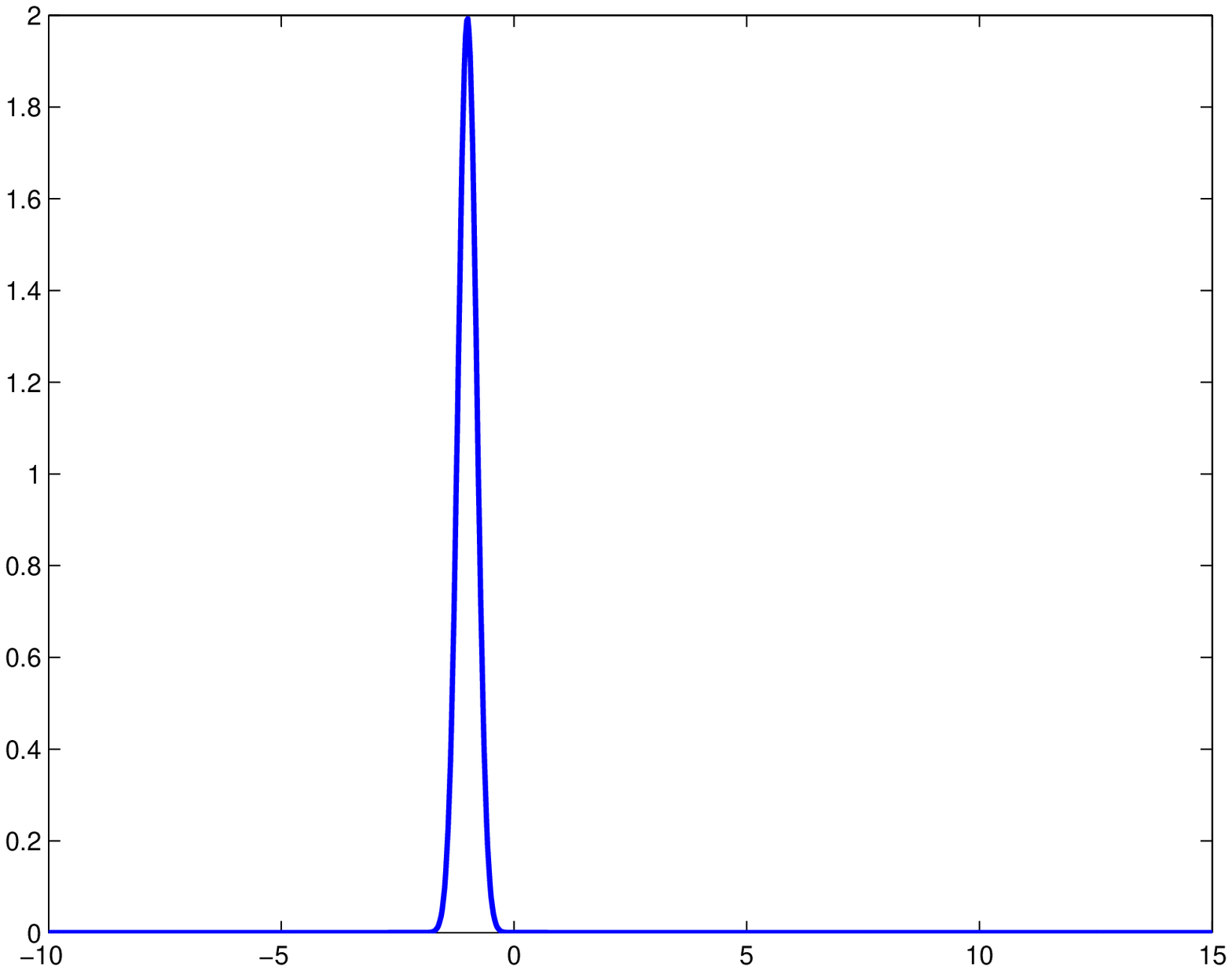} }
\subfigure[]{ \includegraphics[width=3cm]{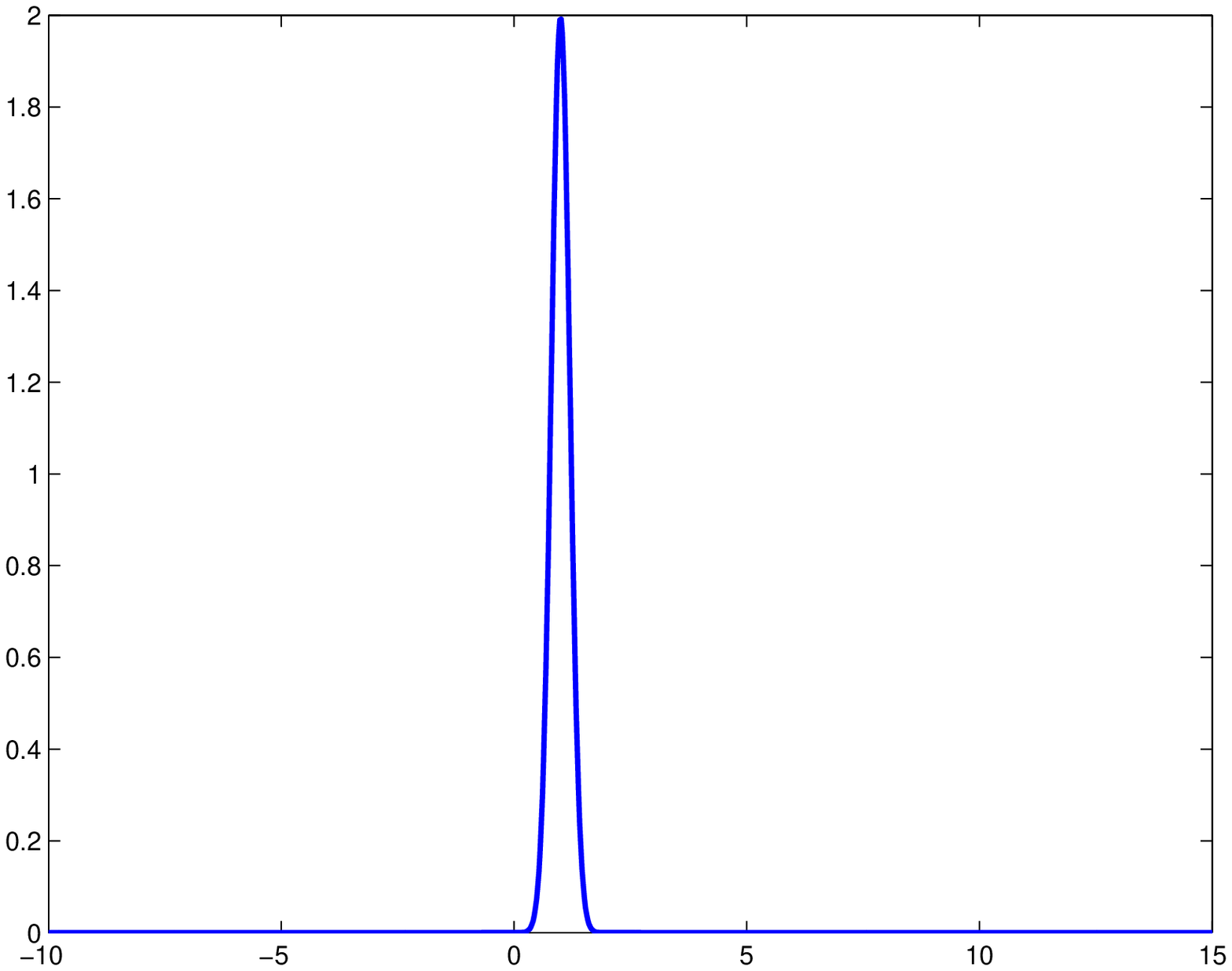} }
\subfigure[]{ \includegraphics[width=3cm]{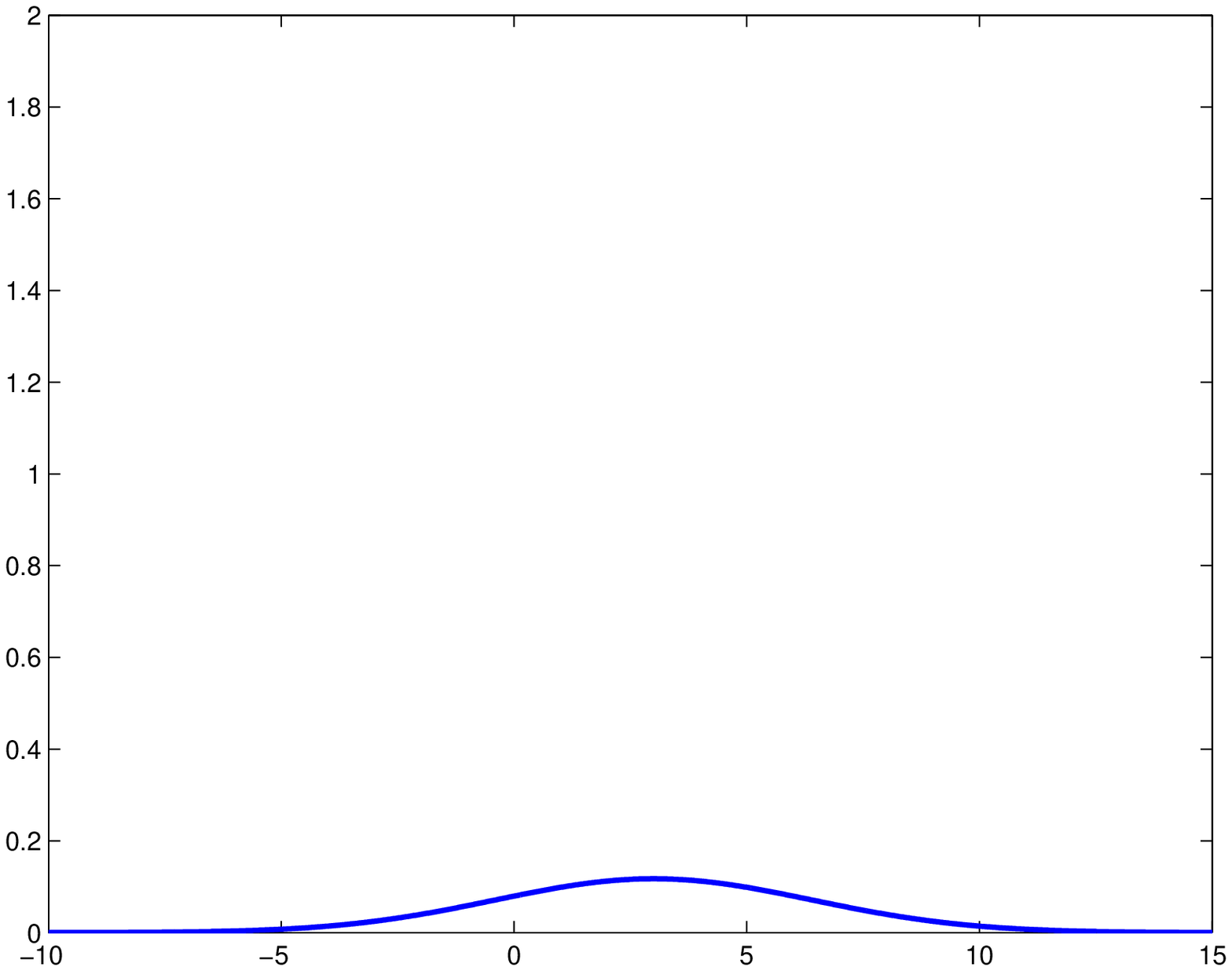} }

\subfigure[]{ \includegraphics[width=3cm]{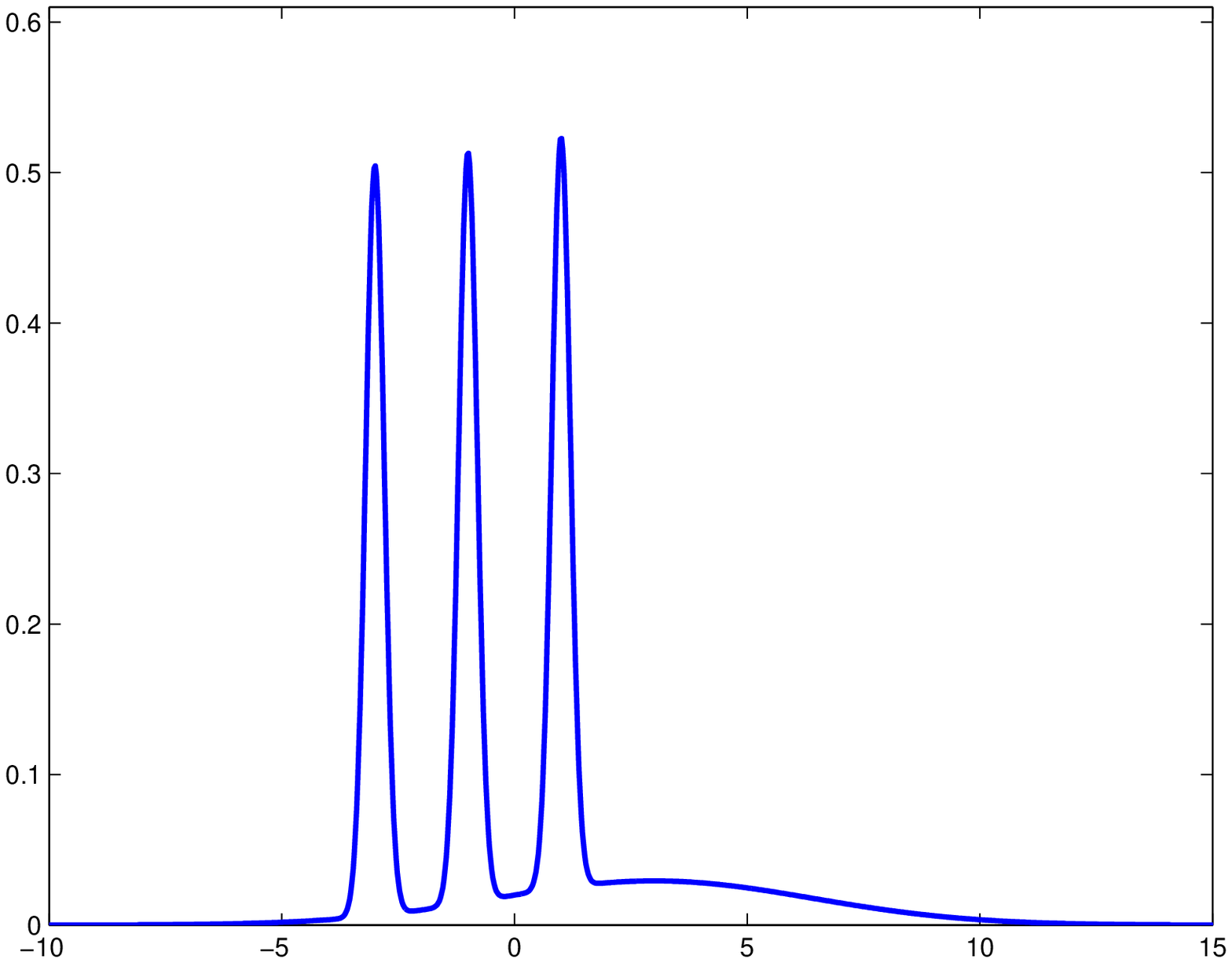} }
\subfigure[]{ \includegraphics[width=3cm]{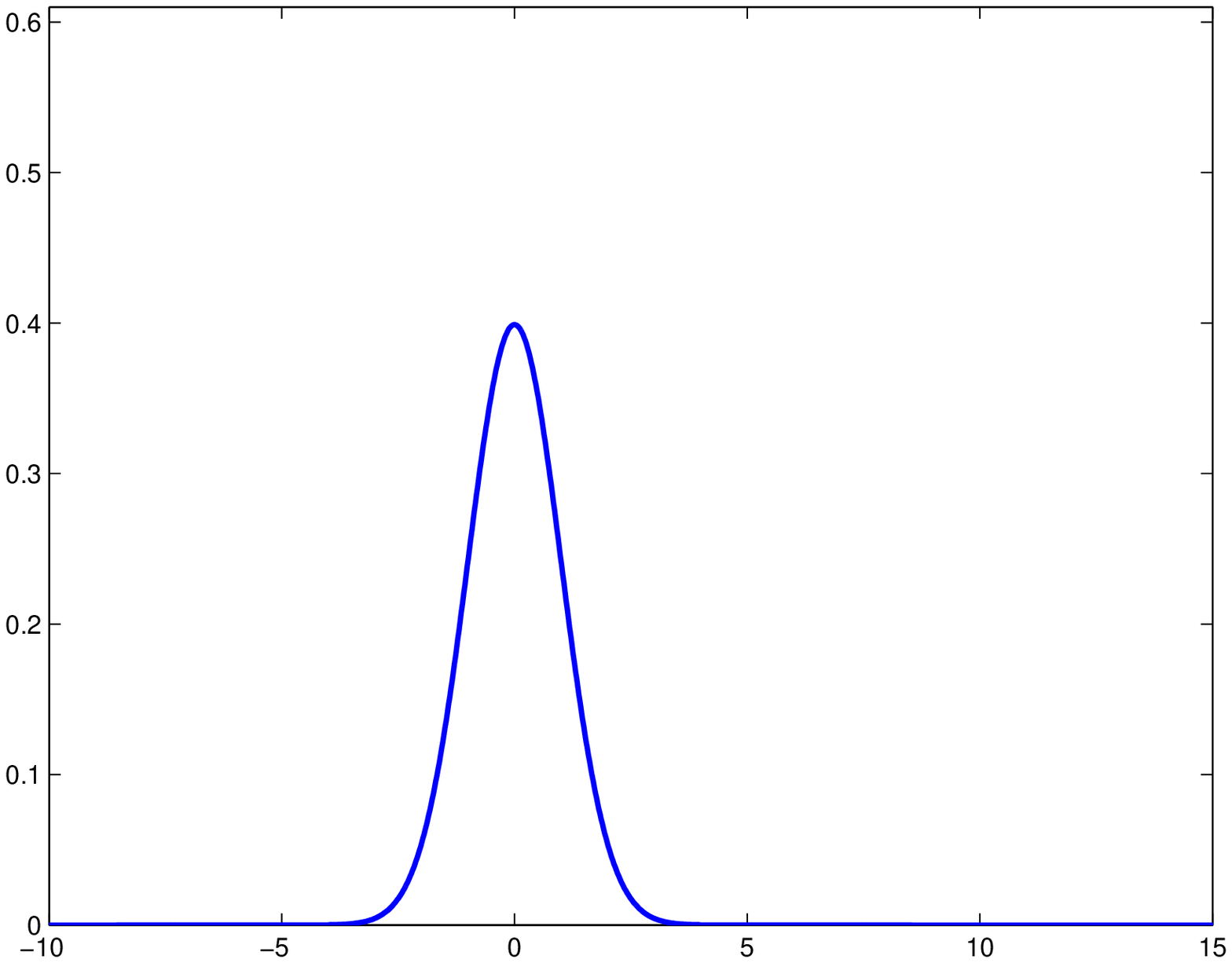} }

\caption{Gaussian densities $f_1,\ldots,f_4$ from the location-scale model \eqref{eq:locsclmodel}, with means and variances given in Table \ref{table:spreadvals}. (e) Euclidean mean of the densities in $L^{2}(\R)$. (f) Density of the barycenter $\nu^*_{4}\in\WS$ of $\nu_1,\ldots,\nu_4$, with densities  $f_1,\ldots,f_4$.} \label{fig:spread}
\end{figure}

As for concentrated data, we first perform a standard PCA on the logarithms in $V_{\mu_0}(\R)$. In what follows, we keep the same notation as in Section \ref{sec:conc}.  In Figure \ref{fig:exnonlinGPCA} we display the vectors $V_i$ and the linear space $\spann(W_1)$.  From the projections of the vectors $V_{i}$ onto the space spanned by $W_1$, it can be seen that  $ \Pi_{\spann(w_1)} v_1  \notin  V_{\mu_0}(\R) $. Therefore, the condition $\Pi_{x_0+\spann(\tilde{\U})} \bx \in X \text{ a.s.}$ in Proposition \ref{prop:PCA2PGA} is not satisfied. Thus, one cannot conclude that $G_1$ is a first empirical GPG.

Now, in order to show that $G_1$ is not a GPG, it suffices to find $G_1^*\in{\rm CG}_{\mu_0,1}(W)$, such that $\costn(G^*)<\costn(G_1^*)$. To that end we perform a CPCA of $v_1,\ldots,v_4$, with $X=A \cap V_{\mu_0}(\R)$ and reference element $x_0=0$. By Proposition \ref{prop:costastG}, this amounts to solving 
\begin{equation}\label{eq:costA}
\min_{u \in A, \|u\|_{\mu_0}=1} \costHn(u)=\frac{1}{n} \sum_{i=1}^n d^2_{\mu_0}(v_i,\spann(u)\cap X).
\end{equation}
On the other hand, letting $Y = \{(\alpha,\beta)' \in \R^2 \; | \; \alpha \geq -1 \}$,  \eqref{eq:costA} is equivalent to

\begin{equation}\label{eq:costU}
\min_{U \in \R^2, \|U\|=1} \operatorname{H}^{(n)}_{Y}(U):=\frac{1}{n} \sum_{i=1}^n d^2(V_i,\spann(U)\cap Y),
\end{equation}
where $d, \|\cdot\|$ are the Euclidean distance and norm in $\R^2$ and $V_{i} = (  a_i  - 1 , b_i )' \in \R^2$, $1 \leq i \leq 4$. We have numerically found a unique minimizer $W_1^{\ast} = (  \alpha_1^{\ast} , \beta_1^{\ast} )$ of \eqref{eq:costU} and so,  $w_1^{\ast}(x) = \alpha_1^{\ast} x + \beta_1^{\ast}$ is the unique minimizer of \eqref{eq:costA}. 

Letting $G_1^*:= \left\{ \nu^{\ast}_{1,t}:=\exp_{\mu_0}(tw_1^{\ast})\;|\; t \in \R,  1 + t \alpha_1^{\ast} \geq 0 \right\}\in{\rm CG}_{\mu_0,1}(W)$, we find that $G_1^*\ne G_1$ and $\costn(G_{1}^{\ast})  < \costn(G_{1})$. Indeed, from Figure \ref{fig:exnonlinGPCA} it can be seen that $W_1^{\ast} \neq W_1$ and also that $\operatorname{H}^{(n)}_{Y}(W^*_1)  < \operatorname{H}^{(n)}_{Y}(W_1)$.

\begin{rem}
For this example of spread-out data, it should be noted that $G_{1}^*$ is not necessarily the first empirical GPG. Indeed,  $G_{1}^*$ is a minimizer of $\costn(G)$ over the sets $G \in \mathrm{CG}_{\mu_0,1}(W)$ such that $G \subset \{\nu^{(a,b)} \;|\; (a,b)  \in (0,\infty) \times \R \}$. Whether or not $G_1^*$ is a first empirical GPC remains as an open issue.
Problem  \eqref{eq:costU} is locally convex in a neighborhood of the (unique) optimum, after parametrization of the unitary sphere of $\R^2$ using polar coordinates.  In the general case of the optimization problems in Propositions \ref{prop:costastG} and \ref{prop:costastN}, issues such as convexity, uniqueness of solution and optimality conditions, need to be addressed.
\end{rem}

\begin{figure}[h!]
\centering
{ \includegraphics[width=7cm,height=7cm]{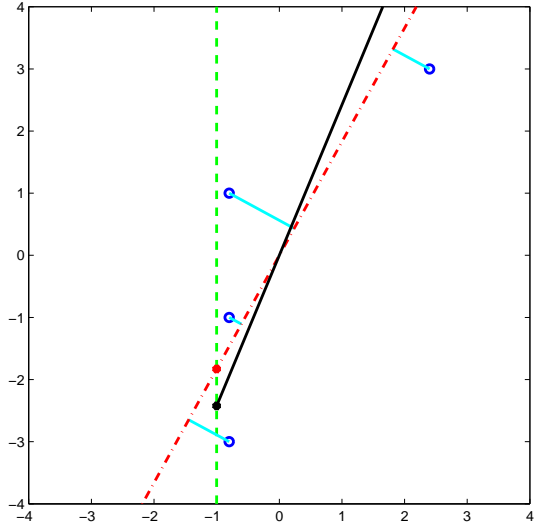} }

\caption{Same interpretation as Figure \ref{fig:exlinGPCA}. The dash-dot line is the linear space spanned by the  first eigenvector $W_1$ from the standard PCA of $V_1,\ldots,V_4$. The solid line is the convex set  $\spann(W_1^{\ast}) \cap Y$ where $W_1^{\ast} = (\alpha_1^*,\beta_1^{\ast})$ is the minimizer of \eqref{eq:costU}.  The dot on the solid line is the projection of $V_1 = (0.2,-3)$ onto  $\spann(W_1^{\ast}) \cap Y$, while  the dot on the dash-dot line is the projection of $V_1$ onto  $\spann(W_1) \cap Y$.} \label{fig:exnonlinGPCA}
\end{figure}

\subsection{Real data example: statistical analysis of population pyramids}\label{sec:realdataex}

We analyze a real dataset consisting of  histograms that represent the population pyramids of 223 countries for the year 2000. This dataset has been studied in \cite{MR2736564} using FPCA of densities. The data  are available from the International Data Base (IDB), produced by the International Programs Center, US Census Bureau (IPC, 2000), and they can be downloaded from the URL  {\tt http://www.census.gov/ipc/www/idb/region.php}. Each histogram in the database represents the relative frequency by age, of people living in a given country. Each bin in a histogram is an interval of one year,  and the last interval corresponds to people older than 85 years. The histograms are normalized so that their area is equal to one, and thus they represent a set of probability density functions.
In Figure \ref{fig:5hist}, we display the population pyramids of five countries.

\begin{figure}[h!]
\centering
\subfigure[Afghanistan]
{\includegraphics[width=0.22 \textwidth,height=0.18\textwidth]{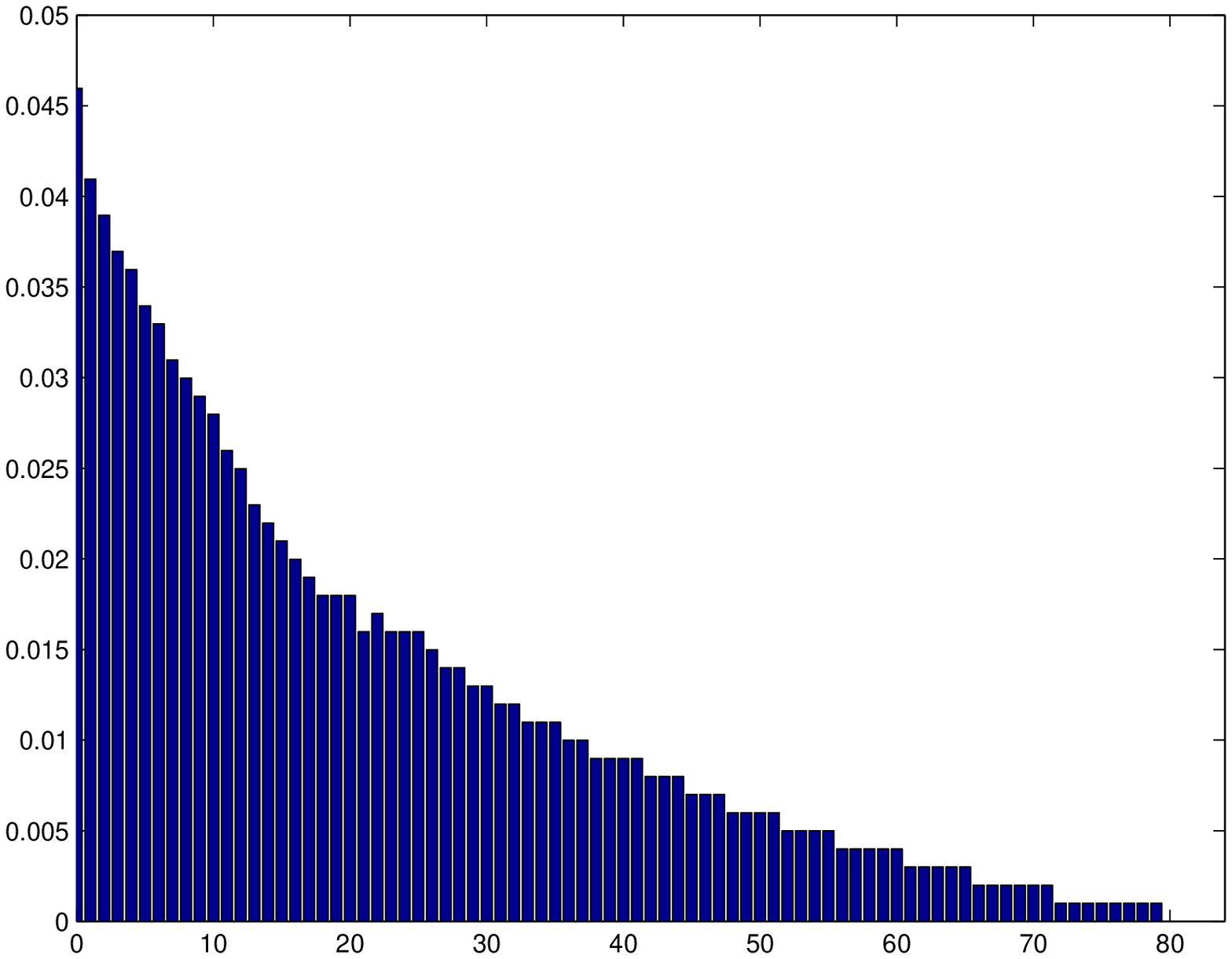}}
\subfigure[Angola]
{ \includegraphics[width=0.22 \textwidth,height=0.18\textwidth]{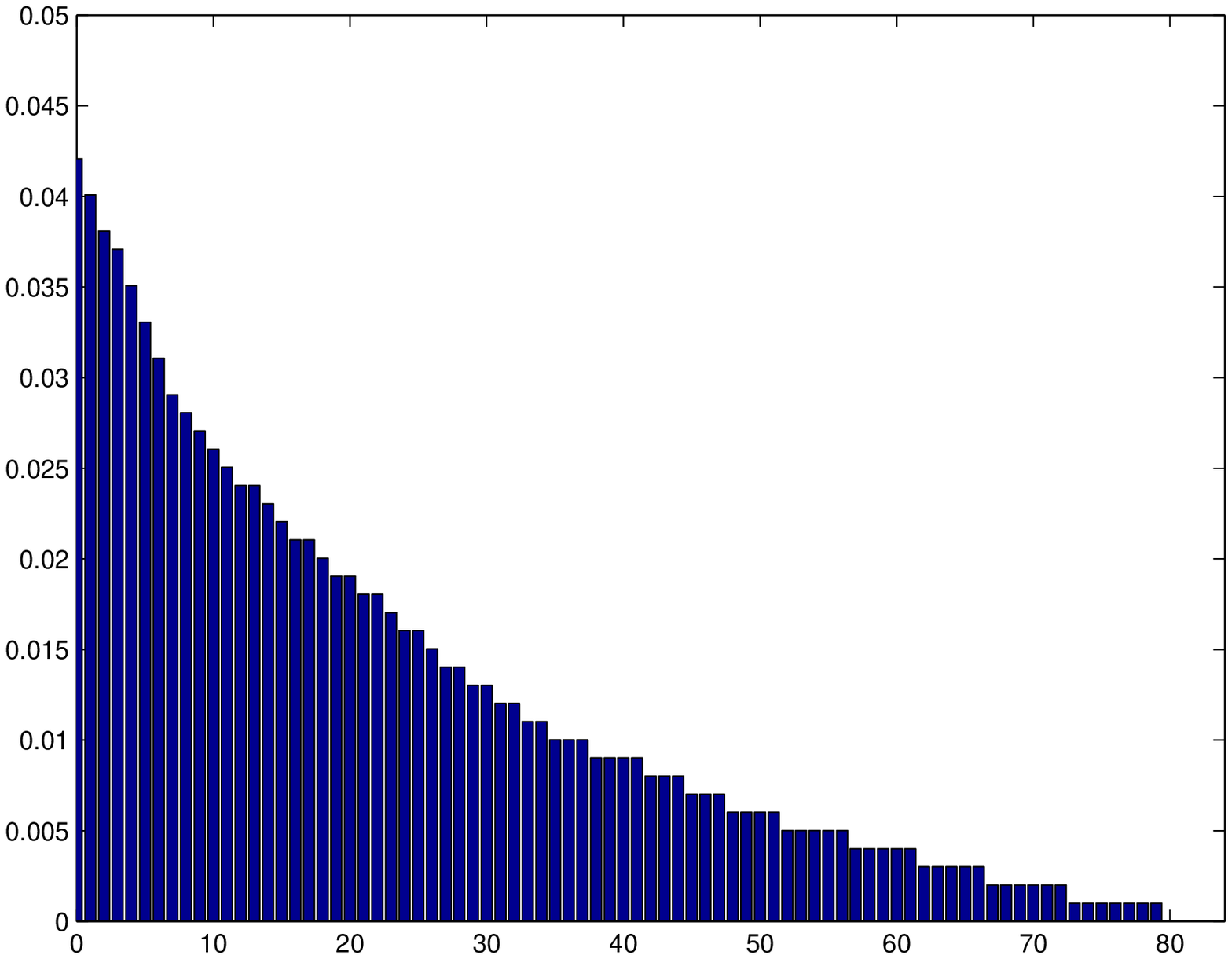}} \\

\subfigure[Australia]
{ \includegraphics[width=0.22 \textwidth,height=0.18\textwidth]{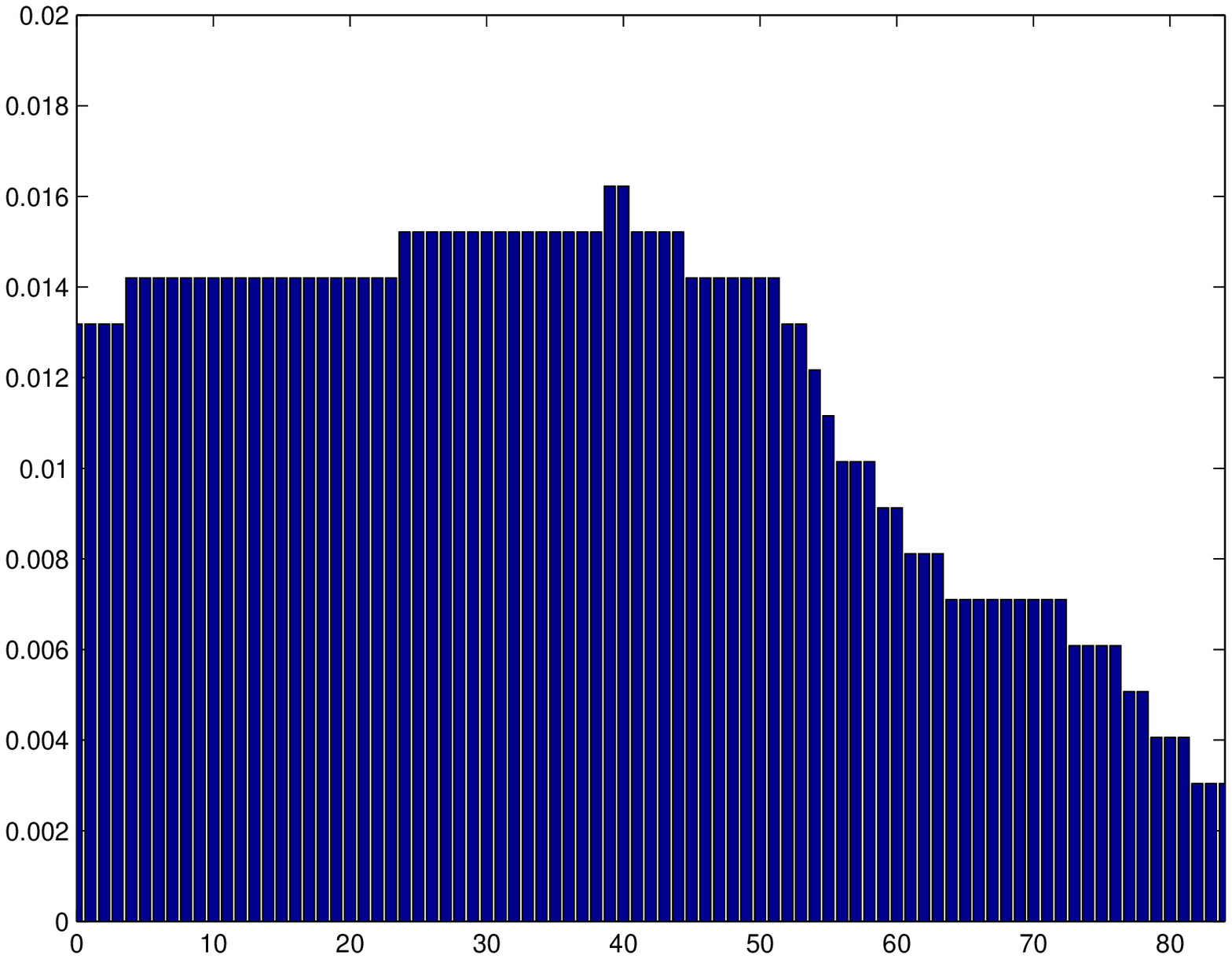}}
 \subfigure[Chile]
 {\includegraphics[width=0.22 \textwidth,height=0.18\textwidth]{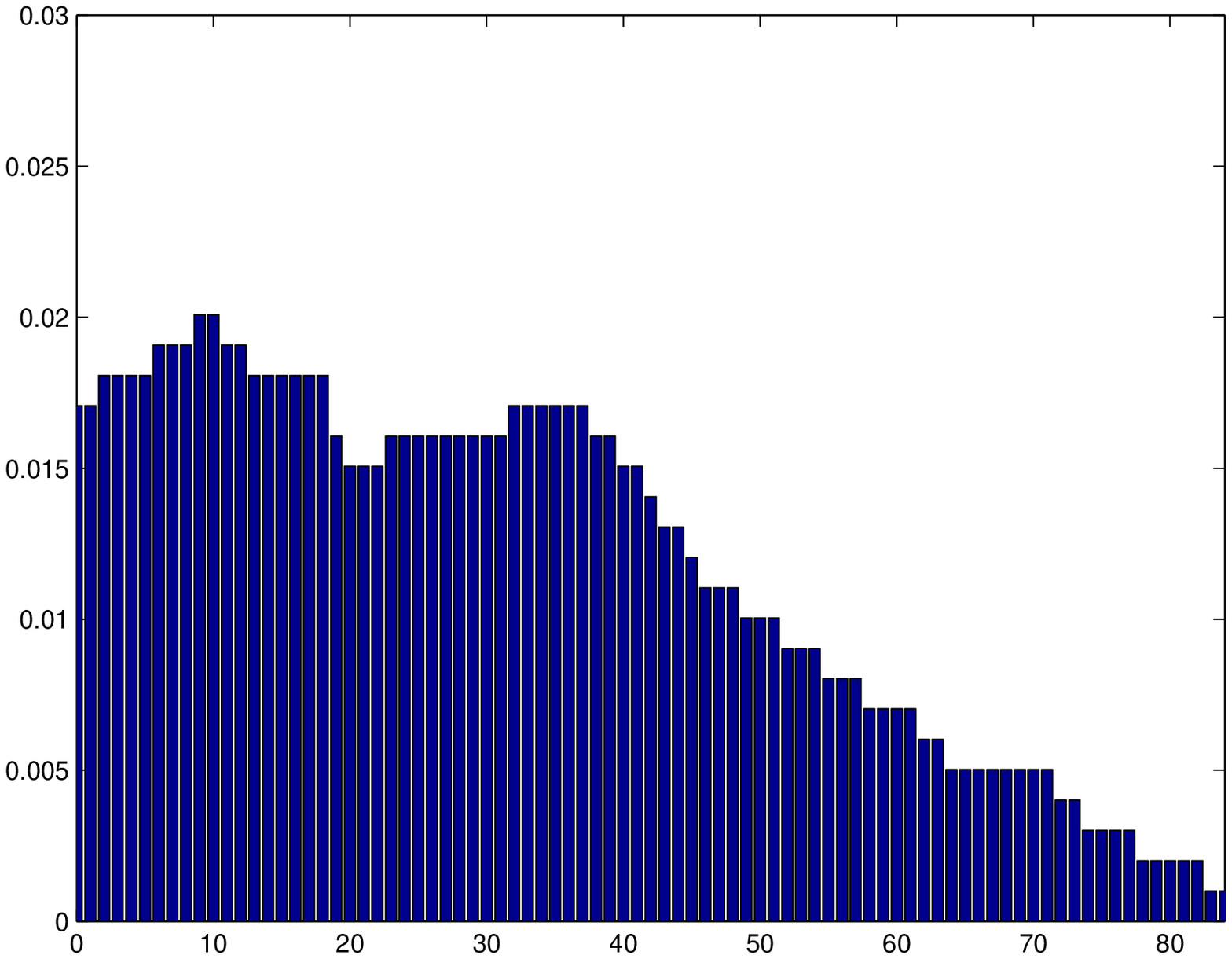}}
 \subfigure[France]
{ \includegraphics[width=0.22 \textwidth,height=0.18\textwidth]{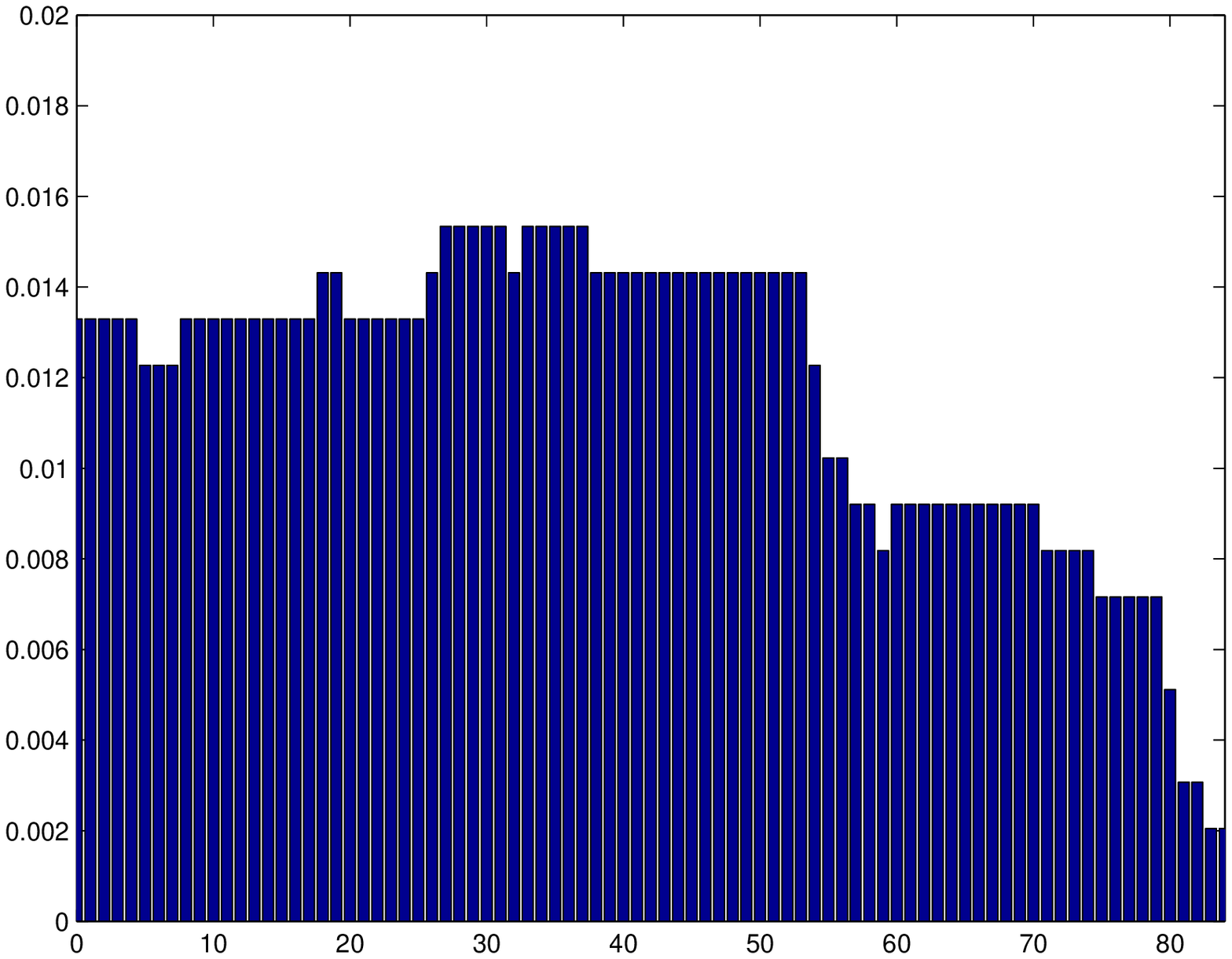}} \\
\caption{Population pyramids of 5 countries for the year 2000.} \label{fig:5hist}
\end{figure}

For the purpose of  summarizing this dataset in an efficient way, we propose to compare the results obtained using either FPCA or GPCA. Note that FPCA of histograms amounts to a standard multivariate PCA in the Euclidean space $\R^{p}$ with $p=85$. In Figure \ref{fig:FPCA_GPCA}(a), we display the projection of the data onto the first two principal components (PC) when performing FPCA. Note that 81 \% (resp.\ 8 \%) of variability is explained by the first  PC (resp.\ the second PC).

To perform GPCA we proceed as follows. First we compute the cdf of each histogram, which allows, from \eqref{eq:exp}, the computation of $v_ i =  \log_{\nu_n^* } (\nu_i ), i=1,\ldots,n=223$, where $\nu_i $ is the probability associated to the $i$-th histogram and $\nu_n^*$ is the Fr\'echet mean of these probability measures in $\WS$. Then, we perform the FPCA of the $v_ i$ in $L^2_{\nu_n^*}(\R)$ to compute the first two PC that we denote by $w_1$ and $w_2$. For this dataset, we notice that the conditions  $\Pi_{\spann(\{w_1,w_2\} )}v_i \in V_{\nu_{n}^*}$, for all $i=1,\ldots,n=223$, are satisfied. Therefore, Propositions \ref{prop:PCA2PGA} and \ref{prop:CPCA2GPCA_G}, the FCPA of data-logarithms in $L^2_{\nu_n^*}(\R)$   leads to a solution of GPCA in $\WS$. In Figure \ref{fig:FPCA_GPCA}(b), we display the projection in  $L^2_{\nu_n^*}(\R)$ of  $v_ i, i=1,\ldots,n$,  onto the first two PC  $w_1$ and $w_2$. Note that, when using GPCA,  96 \% (resp.\ 2 \%) of variability is explained by the first  PC (resp.\ the second PC). Hence, we achieve a better reduction of dimensionality by the use of GPCA. Using the representation in the Wasserstein space, one may conclude that this dataset is essentially one dimensional, in terms of variability around its Fr\'echet mean in $\WS$. In particular, this fact can be observed  in Figure \ref{fig:FPCA_GPCA_pays}, where we plot the projections of the five histograms displayed  in Figure \ref{fig:5hist}.
\begin{rem}\label{rem:randobs}
In this paper we assume that the input data consist of probabilities $\nu_1,\ldots,\nu_n$ belonging to $\WS$. However, in many applications we may have access, only to random observations from each of these probabilities. A natural strategy to address this issue is to estimate the associated densities by means of kernel estimators, for instance, and then a GPCA could be applied to the estimations. Possibly, more efficient estimators of principal geodesics could be obtained by adapting ideas from \cite{MR1946423} which would, however, require a simple representation of principal geodesics in terms of densities.
\end{rem}

\begin{figure}[h!]
\centering
 \subfigure[FPCA]
{ \includegraphics[width=5cm,height=4cm]{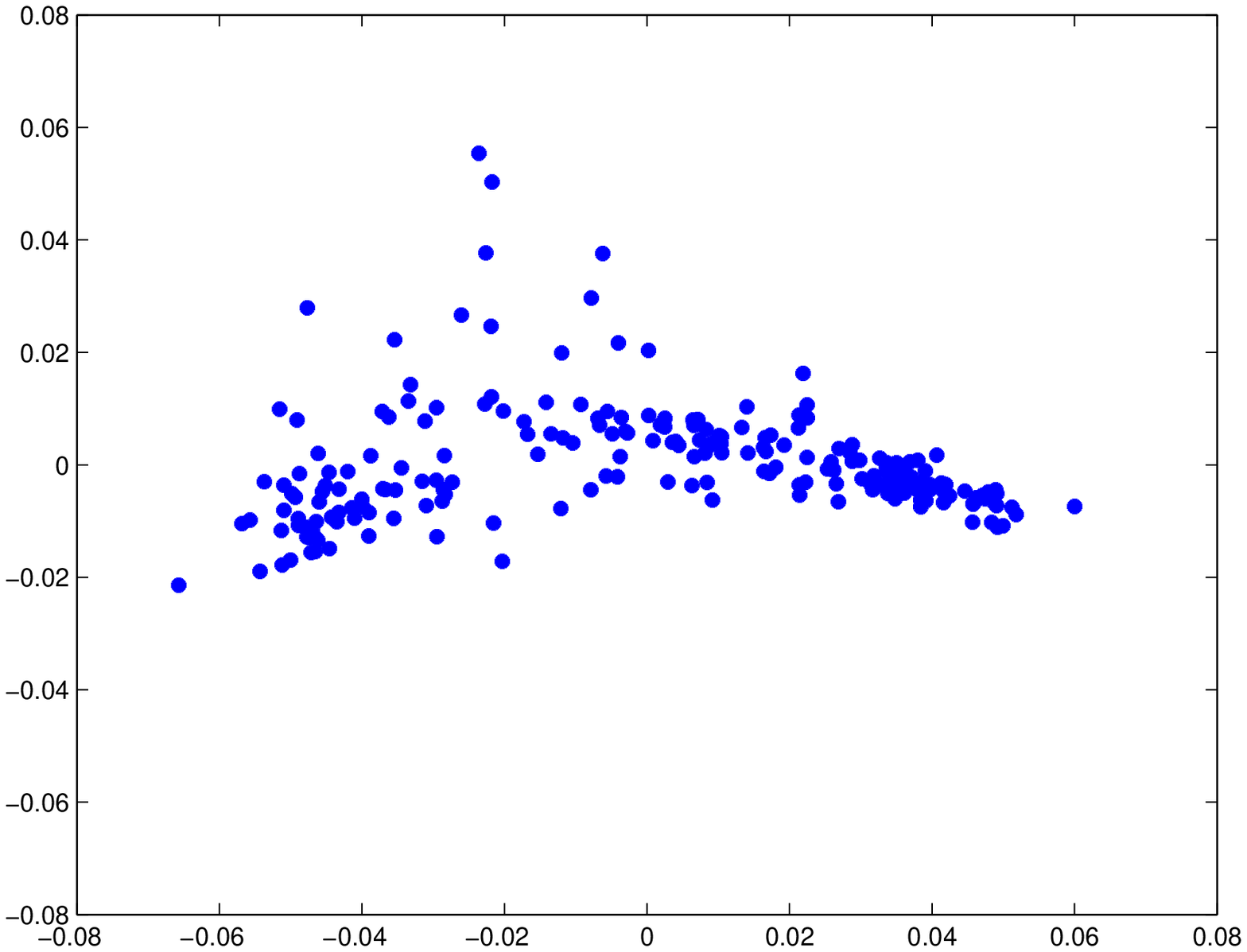}}
  \subfigure[GPCA]
{ \includegraphics[width=5cm,height=4cm]{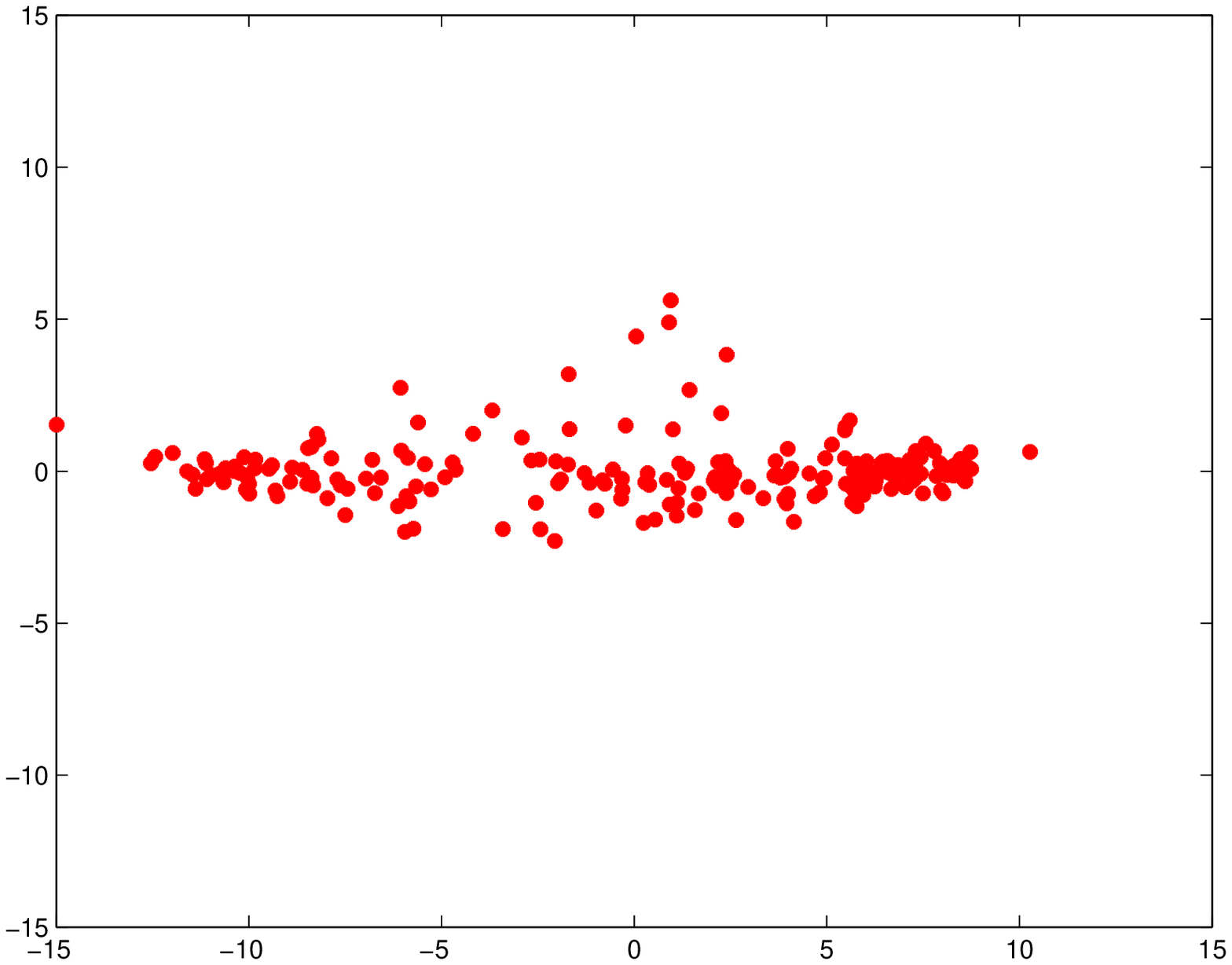} }
 \caption{Population pyramids of 223 countries. (a) Standard PCA: projection  of the data on  the first PC (81 \%) and second PC (8 \%). (b) GPCA: projection  of the data on  the first PC (96 \%) and second PC (2 \%)} \label{fig:FPCA_GPCA}
\end{figure}

\begin{figure}[h!]
\centering
 \subfigure[FPCA]
{ \includegraphics[width=5cm,height=4cm]{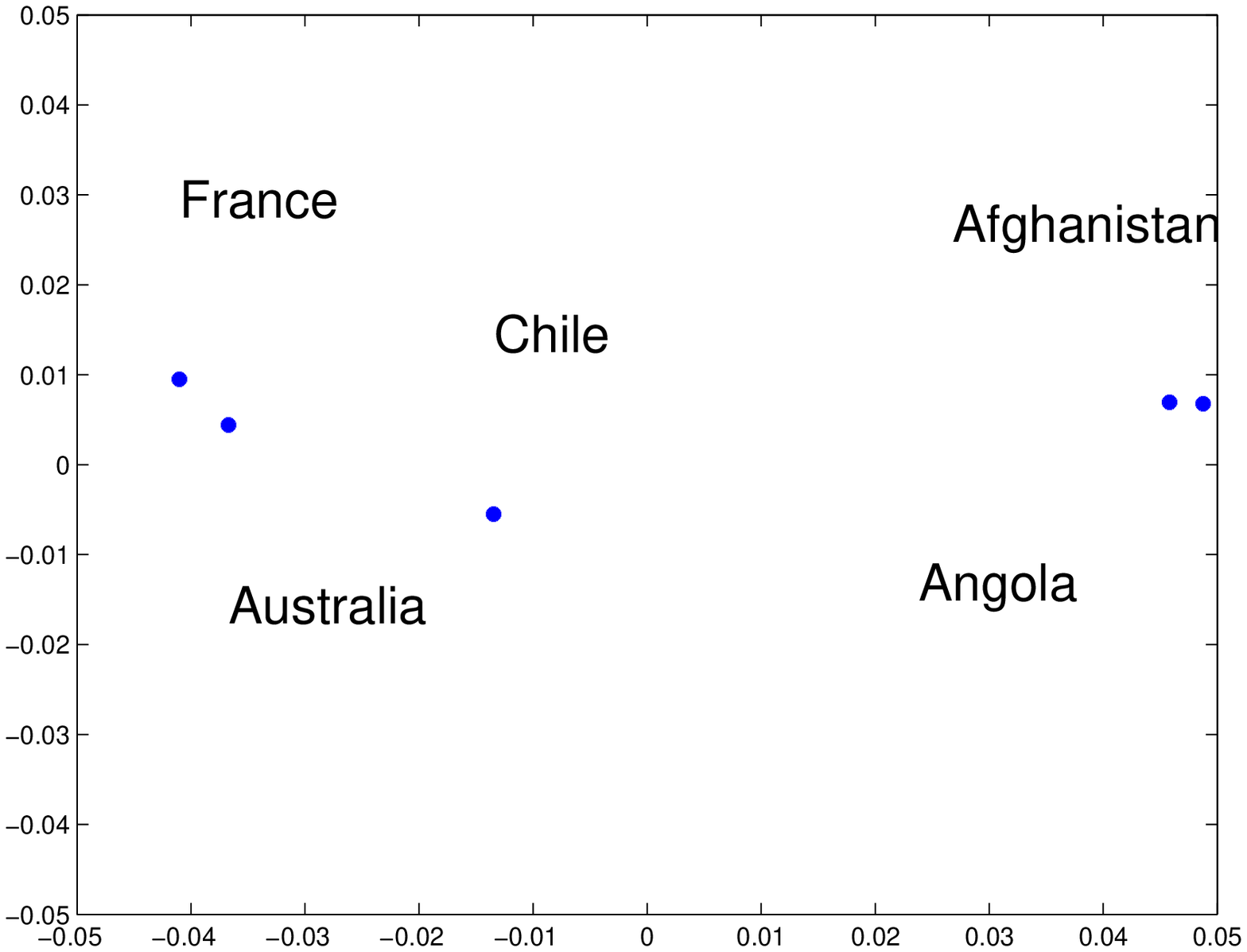}}
  \subfigure[GPCA]
{ \includegraphics[width=5cm,height=4cm]{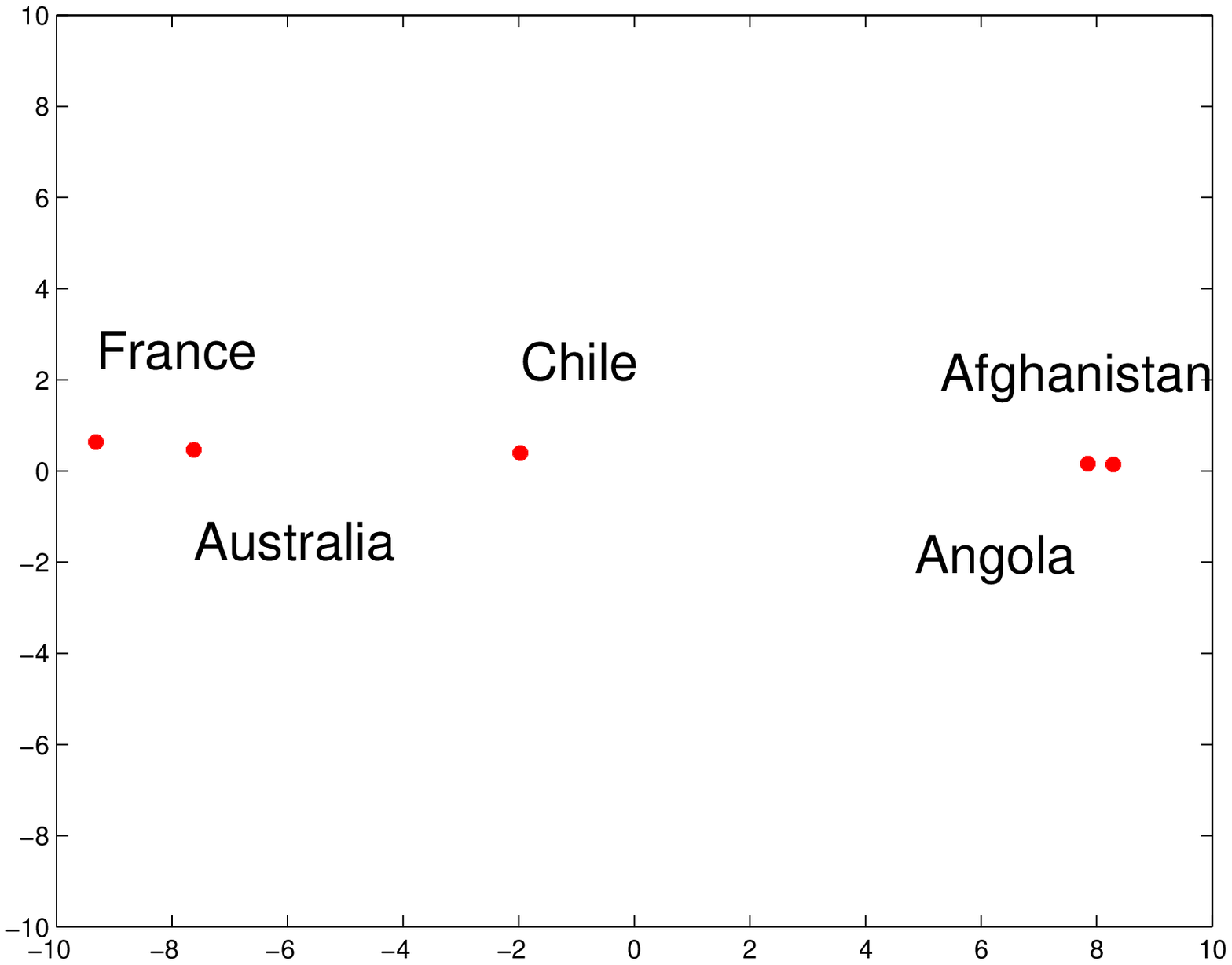} }
 \caption{Projection of the five histograms displayed in Figure \ref{fig:5hist}, using the whole dataset, on the first two PC, from (a) standard PCA and  (b) GPCA.} \label{fig:FPCA_GPCA_pays}
\end{figure}

\section{Analysis of consistency}\label{sec:consitency}

\subsection{Consistency of the empirical CPCA}\label{sec:consistencyCPCA}

Throughout this section we use the notation of Section \ref{sec:cpca}; limits are understood as $n\to\infty$. Let $x_0 = \E\bx$ and let $\bx_1, \ldots,\bx_n$ be independent, identically distributed (iid) copies of $\bx$. Denote by $\bar{\bx}_n := \sum_{i=1}^n \bx_i/n$ their arithmetic mean and observe that $\bar{\bx}_n\to x_0$ a.s., by the strong law of large numbers (SLLN) in a Hilbert space (see \cite{MR2814399}). Let also $\costhnb(C)=\frac{1}{n} \sum_{i=1}^n  d^2(\bx_i,C)$ be the random version of $\costhn$.

We prove in Theorem \ref{theo:GGXcon} that empirical GPCC based on $\bx_1, \ldots,\bx_n$ converge, in a sense defined below, to GPCC of $\bx$. The analogous result for NPCC is omitted.

Following Definition \ref{def:empPCC}, let $\GG_{x_0,k}(X)$ be the set of GPCC of $\bx$, with reference point $x_0=\E\bx$, and
${\cal G}_{n,k}(X):=\argmin_{C\in\mathrm{CC}_{\bar{\bx}_n,k}(X)}\costhnb(C)$ the (random) set
of empirical GPCC of $\bx_1, \ldots,\bx_n$, with $\bar{\bx}_n$ as reference point.

\begin{defin} The empirical GPCC are consistent, denoted
${\cal G}_{n,k}(X)\to\GG_{x_0,k}(X)$ a.s., if for every $C_{n}\in{\cal G}_{n,k}(X),n\ge1$, and $C\in\GG_{x_0,k}(X)$,\ \\
(a) $\costhnb(C_n) \to\costh(C)\;\text{a.s.}$, and\ \\
(b) the accumulation points of $(C_{n})$ belong to $\GG_{x_0,k}(X)$ a.s.

\end{defin}

In the following lemma we show that the indicators of $\mathrm{CC}_{x_n,k}(X)$ (denoted $\chi_{_{n,k}}$) $\Gamma$-converge to the indicator of  $\mathrm{CC}_{x,k}(X)$ (denoted $\chi_{_k}$) when $x_n\to x\in X$. We refer to Section \ref{sec:Gamma} for the definitions of $\Gamma$-convergence and indicator.

\begin{lem}\label{lem:CGkmuW}
Let $x_n\in X, n\ge1$, with $x_n\to x\in X$. If $X$ is compact then
$\Gamma \mbox{-} \lim_{n \to \infty} \chi_{_{n,k}} = \chi_{_k}$.
\end{lem}

\begin{proof}
Recall that under compactness of $X$, $h(C_n,C) \to 0$ is equivalent to $K$-$\lim C_n= C$ (see Section \ref{sec:Klim}). By Lemma \ref{lem:chi}, it is sufficient  to show that $\mathrm{CC}_{x_n,k}(X)$ converges to $\mathrm{CC}_{x,k}(X)$ in the sense of Kuratowski. That is, we have to show that:\ \\
(a) for every $C \in \mathrm{CC}_{x,k}(X)$ there exist $C_n\in \mathrm{CC}_{x_n,k}(X), n\ge1$, with $h(C_n,C)\to0$, and\ \\
(b) if $C$ is an accumulation point of  $C_n \in \mathrm{CC}_{x_n,k}(X), n \geq 1$, then $C \in \mathrm{CC}_{x,k}(X)$.

For (a) take $C \in \mathrm{CC}_{x,k}(X)$ and let $C_n := C+x_n-x\in\mathrm{CC}_{x_n,k}(X)$, $n \geq 1$. After some calculation we find that the deviations $d(C,C_n)$ and $d(C_n,C)$ (see Definition \ref{def:Haussdorff}) are bounded above by $\|x-x_n\|$. Therefore,  $h(C,C_n) \leq \|x-x_n\| \to 0$.

For  (b)  let  $C$ be an accumulation point of $(C_{n})$. Then, since $x_{n} \in C_{n}$ and $x_{n} \to x$,  it follows that $x \in C$, by (ii) in Definition \ref{def:Klim}. On the other hand, recall that $\mathrm{CC}_{k}(X)$ is compact, thanks to Proposition \ref{prop:CLXcompact}. Then, as $C_{n} \in \mathrm{CC}_{k}(X), n \geq 1$, we have $C \in \mathrm{CC}_{k}(X)$ and, since $x \in C$, we conclude that $C \in \mathrm{CC}_{x,k}(X)$.
\end{proof}

\begin{theo}\label{theo:GGXcon}
If $X$ is compact then ${\cal G}_{n,k}(X)\to\GG_{x_0,k}(X)$ a.s.
\end{theo}
\begin{proof}
Let $\chi_{_{0,k}}, \chi_{_{n,k}}$ be the indicators of $\mathrm{CC}_{x_0,k}(X), \mathrm{CC}_{\bar{\bx}_n,k}(X)$ respectively.
Note that
\begin{equation}\label{eq:GGk}
\GG_{x_0,k}(X) = \argmin_{C \in \CL(X)} \costh(C) + \chi_{_{0,k}}(C)\quad \text{and} \quad{\cal G}_{n,k}(X)=\argmin_{C\in\CL(X)}\costhnb(C)+\chi_{_{n,k}}(C).
\end{equation}
From Lemma \ref{lem:CGkmuW}, we have
\begin{equation}\label{eq:gamma1}
\Gamma \mbox{-} \lim_{n \to \infty} \chi_{_{n,k}} = \chi_{_{0,k}},\;\text{a.s.},
\end{equation}
where the $\Gamma$-convergence takes place in the space $\CL(X)$. From Proposition \ref{prop:dKcont} and recalling that $X$ is compact, we have that $d^2(x,C)$ is separately continuous in $x \in X$ and $C \in  \CL(X)$. Hence, $d^2(x,C)$ is measurable on the product space $X \times \CL(X)$; see \cite{Johnson69} or \cite{Rudin81}. Thus, from Theorem 2.3 in \cite{ArtsteinWets1995}, we have the following $\Gamma$-convergence in $\CL(X)$,
\begin{equation}\label{eq:gamma2}
\Gamma \mbox{-} \lim_{n \to \infty}\costhnb(\cdot) = \costh(\cdot) \mbox{ a.s.}
\end{equation}
On the other hand, as  $X$ is compact, there exists a constant $R>0$ such that $d^2(x,C) \leq R$, for all $x \in X$ and $C \in\CL(X)$. Also, by Proposition \ref{prop:CLXcompact} , $\CL(X)$ is a compact set. Therefore, by the uniform strong law of large number (see Lemma 2.4 in \cite{MR1315971}), $\costhnb(C) \to \costh(C)$ uniformly in $\CL(X)$ a.s., that is,
\begin{equation}\label{eq:unif}
\lim_{n \to \infty}\sup_{C \in \CL(X)} |\costh(C)-\costhnb(C)| = 0 \text{ a.s.}
\end{equation}
From \eqref{eq:gamma1} to \eqref{eq:unif} and by Proposition 6.24 in \cite{DalMaso93}, we obtain
\begin{equation}\label{eq:gamma3}
\Gamma \mbox{-} \lim_{n \to \infty} \costhnb + \chi_{_{n,k}} = \costh + \chi_{_{0,k}} \; \text{a.s.}
\end{equation}
Therefore, from \eqref{eq:GGk}, \eqref{eq:gamma3}, the compactness of $\CL(X)$  and Theorem \ref{theo:GammaM}, the conclusion follows.
\end{proof}

\subsection{Consistency of the empirical GPCA}\label{sec:consistencyGPCA}
In this section we use the notation of Section \ref{sec:GPCA}, with $\nu_0 =\nu^* $, the Fr\'echet mean of $\bnu$. Let $\bnu_1, \ldots,\bnu_n$ be iid copies of $\bnu$ and let $\bnu^*_n$ be their empirical Fr\'echet mean. Let $\bx=\log_{\mu}(\bnu), \bx_i=\log_{\mu}(\bnu_i), i=1,\ldots,n$ and $x_0=\log_\mu(\nu_0)$. Let also $\costnb(G)=\frac{1}{n} \sum_{i=1}^n  d^2_{W_2}(\bnu_i,G)$ be the random version of $\costn$. We show the convergence of $\bnu^*_n$ and of the empirical GPG to their population counterparts.

\begin{prop}\label{prop:constfm}
$d_W(\bnu^*_n,\nu_0)\to0$ a.s.
\end{prop}

\begin{proof}
From Proposition \ref{prop:frechetmean}({\it ii}), $x_0=\log_{\mu}(\nu_0)=\E\bx$ and $\log_{\mu}({\bnu}_n^*)={\bar\bx}_n:=\frac{1}{n} \sum_{i=1}^n{\boldsymbol x}_i$. By Theorem \ref{theo:exp} and the SLLN in a Hilbert space (see \cite{MR2814399}), $d_W^2(\bnu_n^*,\nu_0) =   \| {\bar\bx}_n- \E\bx\|_{\mu}^2 \to 0$, a.s.
\end{proof}
\begin{rem}
As pointed out by an anonymous reviewer, the result in Proposition \ref{prop:constfm} follows from Ziezold's strong law of large number \cite{Ziezold1977}.
\end{rem}

Recall that if $\Omega$ is compact then $\WS$ is compact. In this case $\CL(W)$ is also compact, as can be easily shown from  Theorem \ref{theo:exp} and Proposition \ref{prop:CLXcompact}. Therefore,  if $\Omega$ is compact, then  every sequence $G_{n}\in{\cal G}_{n,k}(W),n \geq 1$, has a convergent subsequence in $\CL(W)$.

Let $\GG_{\nu_0,k}(W)$ be the set of GPG of $\bnu$, with reference measure $\nu_0=\nu^*$. Let also  ${\cal G}_{n,k}(W):=\argmin_{G\in\mathrm{CG}_{{\bnu}_n^*,k}(W)}\costnb(G)$ the (random) set of empirical GPG of $\bnu_1, \ldots,\bnu_n$, with reference measure ${\bnu}_n^*$, and ${\cal G}_{n,k}(\VV):=\argmin_{C\in\mathrm{CC}_{\bar{\bx}_n,k}(\VV)}\costhnb(C)$.

\begin{defin} The empirical GPG are consistent, denoted
${\cal G}_{n,k}(W)\to\GG_{\nu_0,k}(W)$ a.s., if for every $G_{n}\in{\cal G}_{n,k}(W),n\ge1$, and $G\in\GG_{\nu_0,k}(W)$,\ \\
(a) $\costnb(G_n) \to\cost(G)\;\text{a.s.}$, and\ \\
(b) the accumulation points of $(G_{n})$ belong to $\GG_{\nu_0,k}(W)$ a.s.
\end{defin}

\begin{theo}\label{theo:GGcon}
If $\Omega$ is compact then ${\cal G}_{n,k}(W)\to\GG_{\nu_0,k}(W)$ a.s.
\end{theo}
\begin{proof}
From Proposition \ref{prop:CPCA2GPCA_G}
\begin{equation}\label{eq:GGcon1}
\GG_{x_0,k}(\VV) =\log_{\mu}\left(\GG_{\nu_0,k}(W)\right)\quad\text{and}\quad {\cal G}_{n,k}(\VV)=\log_{\mu}\left({\cal G}_{n,k}(W)\right) .
\end{equation}

On the other hand, from Theorem \ref{theo:exp}, it can be shown that $\log_\mu$  is an isometric bijection for the Hausdorff distance, between $\CL(W)$ and $\CL(\VV)$. Let $G_{n}\in{\cal G}_{n,k}(W),n \geq 1$, with a subsequence $(G_{n'})$ converging to $G \in \CL(W)$. Then, by the continuity of $\log_\mu$, $C_{n'}:=\log_\mu(G_{n'})\to C:=\log_\mu(G)$.

 From \eqref{eq:GGcon1}, we have $C_{n} \in {\cal G}_{n,k}(\VV),n\ge1$. Therefore, Theorem \ref{theo:GGXcon} implies that $C \in  \GG_{x_0,k}(\VV)$ and $\costhnb(C_n) \to \costh(C)$ a.s. Finally, the result follows from the equalities $\costh(C)=\cost(G),\costhnb(C_n)=\costnb(G_n), n \geq 1,$ and \eqref{eq:GGcon1}.
\end{proof}

\begin{rem}
From the proof of Theorem \ref{theo:GGcon} one can see that the sequence of reference measures $\bnu_n^*$ (the empirical Fr\'echet means) can be replaced by any other sequence of atomless measures in $\WS$, converging to an atomless limit, say, $\mu_{0}$. Of course, the limiting GPG would have $\mu_{0}$ as reference measure.
\end{rem}

\begin{rem}
As commented in Remark \ref{rem:randobs}, in practical situations we often have access only to, say, $n_i$ random observations from each $\nu_i$. In this context, consistency has to be redefined as $n$ and the $n_i$ tend to infinity.
\end{rem}

\section{GPCA in $\WS$ and PCA in a Riemannian manifold} \label{sec:conclusion}

As already mentioned in the introduction, nonlinear analogs of PCA  have been proposed in the literature for the analysis of data belonging to curved Riemannian manifolds \cite{geodesicPCA,citeulike:10314363}. To perform a PCA-like analysis two popular approaches are:
\begin{enumerate}
\item Standard PCA of the data projected onto the tangent space at their Fr\'echet mean, with back projection onto the manifold  and
\item Principal Geodesic Analysis (PGA), that is, a PCA along geodesics.
\end{enumerate}
Below, we briefly recall   the main ideas of these two approaches that generally lead  to different  directions  of geodesic variability in a curved manifold \cite{exactPGA}.

Consider $y_{1},\ldots,y_{n}$  belonging to a complete Riemannian manifold $\M$ admitting a geodesic distance $d_{\M}$. In order to define a PCA like analysis in  $\M$, one needs a notion of average. It has been suggested \cite{geodesicPCA} that  the appropriate notion is the Fr\'echet mean, defined as an element $z \in \M$ (not necessarily unique) minimizing the sum of squared distances to the data, namely
\begin{equation*}
z \in \argmin_{y \in \M} \frac{1}{n} \sum_{i=1}^{n} d_{\M}^{2}(y,y_{i}).
\end{equation*}
We refer to \cite{Bhattacharya03} for details and properties of the Fr\'echet mean in Riemannian manifolds.

Let $T_{z}\M$ be the tangent space to $\M$ at  $z$. If $v$ denotes a tangent vector in $T_{z}\M$, there exists a unique geodesic $\gamma_{v}(t)$ having $v$ as  its initial velocity, where $t \in \R$ is a time parameter.  The Riemannian exponential map $\exp_{z} : T_{z}\M \to \M$,
 defined by $\exp_{z}(v) = \gamma_{v}(1)$ is a diffeomorphism on a neighborhood of zero and its inverse is the Riemannian log map, denoted by $\log_{z}$. \\

\noindent {\bf (1) PCA via linearization in the tangent space:} in this approach the data $y_1,\ldots,y_n$ is first projected on  $T_{z}\M$ by means of the $\log_z$ map, thus obtaining $x_{i} = \log_{z}(y_{i}), \; i=1,\ldots,n$. Next, a standard PCA of  $x_{1},\ldots,x_{n}$ is performed in the linear space $(T_{z}\M, \langle \cdot, \cdot \rangle, \| \cdot \|)$, which leads to computing the first principal component $v^{lin}$, the eigenvector associated with the largest eigenvalue of the covariance operator
\begin{equation*}
K v = \frac{1}{n} \sum_{i=1}^{n} \langle x_{i} - \bar{x}_n, v \rangle (x_{i} - \bar{x}_n), \; v \in T_{z}\M,
\end{equation*}
where $\bar{x}_{n} = \frac{1}{n} \sum_{i=1}^{n} x_{i}$. Finally, $v^{lin}$ is  projected back onto $\M$ by means of the $\exp_z$ map to obtain $w^{lin} = \exp_{z}(v^{lin})$, which represents a first notion of principal direction  of geodesic variability. The main drawback of PGA via linearization is the fact that distances are generally not preserved by the projection step, that is, \ $\|x_{i} - x_{j} \| \neq d_{\M}(y_{i},y_{j})$. \\

\noindent {\bf (2) PGA  on $\M$: }  the notion of PCA along geodesics on $\M$ is motivated by formulation  \eqref{eq:pca}, which characterizes standard PCA. In a first step, one computes
\begin{equation*}
v^{geo} = \argmin_{v \in  T_{z}\M, \; \| v \| = 1}  \frac{1}{n} \sum_{i=1}^{n} d^2_{\M}\left( y_i, G_{v} \right),
\end{equation*}
where $G_{v} = \{ \exp_{z}( t v ), \; t \in \R \}$ and $d_{\M}(y,G) =\inf_{y' \in G}d_{\M}(y,y')$ for $y \in \M$ and $G \subset \M$. Then, in a second (and final) step, one projects the element   $v^{geo} \in T_{z}\M$ onto $\M$, by computing $w^{geo} = \exp_{z}(v^{geo})$. This yields another notion of principal direction  of geodesic variability of the data and generally one has that $w^{lin} \neq w^{geo}$, except if $\M$ is a Hilbert space. Therefore, PCA via linearization on the tangent space and PCA along geodesics may lead to different  directions  of geodesic variability in a curved manifold. A detailed analysis of the differences between these methods can be found in \cite{exactPGA}. In both methods it is also possible to define subsequent principal directions (second, third, and so on) of geodesic variability   in a recursive manner, and we refer to \cite{geodesicPCA} for further details. \\

In this paper, we have considered the analysis of data in the Wasserstein space $\WS$, which is not a  Riemannian manifold but has pseudo-Riemannian structure, rich enough to allow the definition of a notion of geodesic PCA. By means of the analogs of the logarithmic and  of the  exponential maps, we also introduce the corresponding version of the standard PCA in the tangent space, with back projection onto $\WS$, thus establishing a parallel to the methodological duality available for data in Riemannian manifolds, as presented above. Also, as could be expected, these two approaches yield, in general, different forms of geodesic variability.

There is however a significant distinguishing feature of our methodology, namely the possibility of performing a PCA in the tangent space under convexity restrictions, which is equivalent (after projection) to the geodesic PCA in $\WS$. This motivates the definition of Convex PCA (see Section \ref{sec:cpca}), a general PCA-like method for analyzing data on a closed convex subset of a Hilbert space, which can be of interest beyond its specific application in the context of GPCA.  The CPCA applied to the logarithms of the data measures is interesting because it is formally simpler than the geodesic PCA in $\WS$ although more complex than standard PCA. In this respect it is also worth noticing that if the data are ``sufficiently concentrated'', the standard and the restricted PCA in the tangent space yield the same results.

It should be  mentioned that the terminology geodesic PCA (GPCA) was used previously by Huckemann et al. in  \cite{citeulike:10314363} to denote a Riemannian manifold generalization of linear PCA.  Their approach shares similarities with the PGA method introduced in \cite{geodesicPCA} but optimizes additionally for the placement of the center point (not necessarily equal to the Fr\'echet mean). Furthermore, it does not use a linear approximation of the manifold and is only suited for Riemannian manifolds, where explicit formulas for geodesics exist. However, it is difficult to compare our approach to the GPCA in \cite{citeulike:10314363} since the notion of principal geodesic, that we propose in this paper, is defined with respect to a given reference measure $\nu_0$ (chosen to be either the population or the empirical Fr\'echet mean). For a precise comparison it would be necessary to carry out the optimization in Definition \ref{def:GG}(a), with respect to the reference measure $\nu_0$, a task which is beyond the scope of this paper.

Finally observe that, from Theorem \ref{theo:exp}, one can interpret $\WS$ as a space with no curvature, and hence the pseudo-Riemannian formalism, used in Section \ref{sec:Riemannian}, is not essential for our development. However, such a framework allows making a connection between our approach and PCA methods adapted to Riemannian manifolds.

\appendix
\section{Appendix}
\subsection{Increasing functions and quantiles}
We present some useful, well-known results about increasing functions and quantiles. For additional information, see  \cite{Embrechts10anote}, \cite{Rockafellar14}. In this section $\mu, \nu$ denote probability measures on $(\R,{\cal B}(\R))$, $F_\nu$ denotes the  (right-continuous) cdf of $\nu$ and $L^2(0,1)$ is the space of square-integrable functions, with respect to the Lebesgue measure on $(0,1)$.
\begin{defin}\label{def:incr}
Let $A \subseteq\R$ and $T:A\to\R$.
\ \\
(a) $T$ is increasing on $B\subseteq A$ if $\forall x,y\in B$, $x<y$ implies $T(x)\le T(y)$.
\ \\
(b) $T$ is $\mu$-a.e. increasing if there exists $B_\mu \in {\cal B}(\Omega)$, with $B_\mu\subseteq A$, $\mu(B_\mu)=1$ and $T$  increasing on $B_\mu$.
\end{defin}
\begin{rem}
A $\mu$-a.e. increasing function $T:A\to\R$ needs not to have a version increasing on $A$. A version of $T$ is a function $\tilde T:A\to\R$ such that $T=\tilde T$, $\mu$-a.e.
\end{rem}
\begin{defin}\label{def:quantile}
The quantile function of $\nu$ is defined as $F_\nu^-(y)=\inf\{x\in \R|F_\nu(x)\ge y\}$,  $y\in(0,1)$.
\end{defin}
\begin{prop}\label{prop:quantile}
\ \\
(a) $F_\nu^-$ is left-continuous and increasing on $(0,1)$.
\ \\
(b) Any left-continuous and increasing $T:(0,1)\to\R$ is the quantile of some probability $\nu$.
\ \\
(c) $\nu$ has finite second moment if and only if $\int_0^1(F_\nu^-(x))^2dx<\infty$.

\end{prop}
\begin{proof}
See  \cite{Embrechts10anote}, \cite{Rockafellar14}.
\end{proof}
\begin{lem}\label{lem:quantile}
Let $T\in L^2(0,1)$ a.e. increasing. Then there exists $\nu\in W_2(\R)$ such that $T=F_\nu^-$ a.e.
\end{lem}
\begin{proof}
Suppose $T$ is increasing on a full measure set $B\subseteq (0,1)$ (that is the Lebesgue measure of $B$ is one). Let $\tilde T:(0,1)\to\R$ be defined as $\tilde T(x)=T(x)$, for $x\in B$, and $\tilde T(x)=\inf_{y\in B,x<y}T(y)$, for $x\not\in B$. Then $\tilde T$ is increasing in $(0,1)$ and $\tilde T=T$ a.e. Finally, let $\hat T$ be the left-continuous version of $\tilde T$, that is, $\hat T(x):=\lim_{t\to x^-}\tilde T(t)$. So, as $\hat T$ is left-continuous and increasing on $(0,1)$, from Proposition \ref{prop:quantile}(b,c) there exists a probability $\nu\in W_2(\R)$, such that $F_\nu^-=\hat T$. Finally, since the number of discontinuities of any increasing function is countable, we have $\hat T=\tilde T$ a.e.
\end{proof}
\begin{prop}\label{prop:closedconvexquantile}
Let $\Omega$ be an interval of real numbers (not necessarily bounded). Then, the set of quantile functions $\{F_\nu^-|\nu\in\WS\}$ is closed and convex in $L^2(0,1)$.
\end{prop}
\begin{proof}
For convexity let $\alpha\in(0,1)$ and $\nu_1,\nu_2\in\WS$. Then $T_\alpha:=\alpha F_{\nu_1}^-+(1-\alpha)F_{\nu_2}^-$ is increasing, left-continuous and square integrable. Hence, by Proposition \ref{prop:quantile}(b,c), $T_\alpha$ is the quantile of some $\nu_\alpha\in\WS$. For closedness consider a sequence $(\nu_n)$ in $\WS$, such that $\int_0^1(F_{\nu_n}^-(x)-T(x))^2dx\to0$, as $n\to\infty$. Then, there exists a subsequence $(\nu_{k_n})$ of $(\nu_n)$ such that $F_{\nu_{k_n}}^-\to T$ a.e. and hence, $T$ is square-integrable and a.e. increasing. So, by Lemma \ref{lem:quantile}, $T$ is a quantile. As usual, the elements of $L^2(0,1)$ are understood as equivalence classes.
\end{proof}
\subsection{Geodesics in metric spaces}
We introduce the concept of geodesic in metric spaces. For notations, definitions and results, we follow \cite{Chodosh2011} and references therein. For convenience, without loss of generality, we consider $I$ such that $[0,1] \subset I$.
\begin{defin}
A curve in a metric space $(X,d)$ is a continuous function $\gamma : I \to X$, where $I \subset \R$ is a closed (not necessarily bounded) interval. Also
\ \\
(i) $\gamma$ is said to pass through $z \in X$ if $\gamma(t)=z$, for some $t \in I$;
\ \\
(ii) $\gamma$ joins $x,y \in X$ if there exists $a,b \in I$, such that $\gamma(a)=x$ and $\gamma(b)=y$ and
\ \\
(iii) $\gamma$ is rectifiable if its length $L(\gamma)$ is finite.
\end{defin}

\begin{defin}\label{def:geodesic}
A metric space $(X,d)$ is said to be geodesic if for every $x,y \in X$, there exists a rectifiable curve $\gamma$ joining $x$ and $y$, such that $d(x,y) = L(\gamma)$. Such minimum length curve $\gamma$ is called a shortest path between $x$ and $y$. A curve $\gamma : I \to X$ is a geodesic if for every $t \in I$, there exist $a,b \in I, a < b, a \leq t \leq b$ such that the restriction of $\gamma$ to $[a,b]$ is a shortest path between $\gamma(a)$ and $\gamma(b)$.
\end{defin}

The following is a useful characterization of shortest path (See \cite{Chodosh2011} for a proof).

\begin{lem}\label{lem:reparam}
For any shortest path, there exists a continuous reparametrization $\gamma$ on $[0,1]$ such that
\begin{equation*}
d(\gamma(s),\gamma(t)) = |t-s|d(\gamma(0),\gamma(1)) \mbox{ for all } s,t \in  [0,1] .
\end{equation*}
\end{lem}

\begin{lem}\label{lem:gammaH}
Let $H$ be a Hilbert space and $x,y \in H$. Then $\gamma$ is a shortest path joining $x$ and $y$ if and only if  $\gamma(t) = (1-t)x + ty$, for all $t \in  [0,1]$, up to a continuous reparametrization.
\end{lem}

\begin{proof}
Denote the inner product and the induced norm in H by $\langle \cdot,\cdot \rangle$ and $\|\cdot\|$ respectively. Let $\gamma$ be a shortest path between $x$ and $y$, and $t \in [0,1]$. After a reparametrization such that $\gamma(0)=x$ and $\gamma(1)=y$, from Lemma \ref{lem:reparam} we have $\|x-\gamma(t)\| = t \|x-y\|$ and $\|\gamma(t)-y\| = (1-t) \|x-y\|$, then
$\|x-\gamma(t)\| + \|\gamma(t)-y\| = \|x-y\|$.

Squaring and simplifying the former expression above, we obtain $\|x-\gamma(t)\| \|\gamma(t)-y\| = \langle x-\gamma(t) , \gamma(t)-y\rangle$.
Hence, by the Cauchy-Schwartz inequality, there exists $\lambda \geq 0$ such that $x-\gamma(t) = \lambda (\gamma(t)-y)$. Finally,  taking norm we find $\lambda = \frac{t}{1-t}$ and the result follows. The other implication is direct.
\end{proof}

From the previous lemma we deduce that, in Hilbert spaces, any geodesic is locally a segment and so, geodesics are straight lines. We state this in the following corollary.

\begin{coro}\label{coro:gamma}
Let $H$ be a Hilbert space and $\gamma : I \to H$ a curve, such that $\gamma(0)=x$ and $\gamma(1)=y$. Then $\gamma$ is a geodesic if and only if  $\gamma(t) = (1-t)x + ty$, for all $t \in  I$, up to a continuous reparametrization.
\end{coro}

\begin{defin}\label{def:geodesicspace}
Let $(X,d)$ be a geodesic space and  $Y \subset X$. We say that $Y$ is geodesic if the induced metric space $(Y,d)$ is geodesic. In other words, if for any $x,y \in Y$, there exists a shortest path joining $x$ and $y$, totally contained in $Y$.
\end{defin}

Note from Lemma \ref{lem:gammaH} that a Hilbert space $H$ is geodesic and $C \subset H$ is geodesic if and only if $C$ is convex.

\subsection{$K$-convergence}\label{sec:Klim}

In this section we present definitions and results that we use for proving the existence of principal geodesics (see Section \ref{sec:pgs}). In particular, we define  an appropriate concept of convergence for sequences of convex sets in a metric space $(X,d)$.

\begin{defin}\label{def:Klim}
Let $C,C_n \subset X, n \geq 1$. We say that the sequence $(C_n)$ converges to $C$ in the sense of Kuratowski, denoted by $K$-$\lim_{n \to \infty}C_n=C$, if
\ \\
(i) for all $x \in C$, there exist $x_n\in C_n, n\ge1$,  such that $x_n\to x$ and
\ \\
(ii) for all $x_n \in C_n, n \geq 1$, and for any accumulation point $x$ of $(x_n)$,  $x\in C$.
\end{defin}

\begin{defin} \label{def:Haussdorff}
The deviation from $x \in X$ to $B\subseteq X$ is defined by $d(x,B) := \inf_{x' \in B}d(x,x')$; the deviation from $A\subseteq X$ to $B$ is $d(A,B) := \sup_{x \in A}d(x,B)$ and the Hausdorff distance between the sets $A$ and $B$ is
\begin{equation}\label{eq:h}
h(A,B) := \max\{d(A,B),d(B,A)\}.
\end{equation}
\end{defin}

\begin{rem}\label{rem:kurato}
It is well known (see \cite{Price40,Beer85} and references therein) that convergence with respect to the Hausdorff distance is stronger than convergence in the sense of Kuratowski. Moreover,  if $X$ is compact both notions of convergence coincide.
\end{rem}
\begin{defin}\label{def:CLX}
 We define the metric space $\CL(X)$ as the set of nonempty, closed subsets of $X$,
endowed with the Hausdorff distance $h$.
\end{defin}

\begin{prop}\label{prop:dKcont}
For all $x\in X$, then $d(x,\cdot)$ is continuous on $\CL(X)$.
\end{prop}
\begin{proof}
Observe that, for all $x\in X$ and $A,B\in \CL(X)$, $d(x,A)\le d(x,B)+h(A,B)$. Then $|d(x,A) - d(x,B)|\le h(A,B)$ and the conclusion follows.
\end{proof}

\begin{lem}\label{lem:Klimaux}
Let $B,C,B_n,C_n \subset X$, with $B_n \subset C_n$, $n \geq 1$, such that $K$-$\lim_{n \to \infty}B_n=B$ and $K$-$\lim_{n \to \infty}C_n=C$. Then $B \subset C$.
\end{lem}
\begin{proof}
By Definition \ref{def:Klim}(i), for any $x\in B$ there exist $x_n \in B_n, n \geq 1$, such that $x_n\to x$. As $x_n \in B_n \subset C_n$, $n \geq 1$, from Definition \ref{def:Klim}(ii) we have $x \in C$.
\end{proof}

\subsection{$\Gamma$-convergence}\label{sec:Gamma}

The notion of $\Gamma$-convergence in a metric space $(X,d)$ (\cite{Attouch84,DalMaso93}) is used in the proof of Theorem \ref{theo:GGcon}.

\begin{defin}\label{def:GammaConv}
Let $F, F_n:X\mapsto \overline{\R}:=\R \cup \{+\infty,-\infty\}, n\ge1$, a sequence of functions. We say that $(F_n)$ $\Gamma$-converges to $F$, denoted $\Gamma$-$\lim_{n \to \infty}F_n = F$, if, for every $x \in X$, \ \\
(i) $F(x) \leq \liminf_{n \to \infty}F_n(x_n)$, for any $x_n\in X, n\ge1$, with $x_n\to x$, and\ \\
(ii) there exist $x_n\in X, n\ge1$, with $x_n\to x$, such that $F(x) = \lim_{n \to \infty}F_n(x_n)$.
\end{defin}

\begin{defin}\label{def:M}
For $F : X \to \overline{\R}$, let $M(F) := \{x \in X : F(x) = \inf_{y \in X}F(y) \}$.
\end{defin}

The following result (see \cite{DalMaso93}, Theorems 7.8 and 7.23) shows that $\Gamma$-convergence together with compactness (or more generally equicoercivity) implies convergence of minimum values and minimizers.

\begin{theo}\label{theo:GammaM}
Assume that $X$ is compact and let $F, F_n:X\mapsto \overline{\R}, n\ge1$, such that $\Gamma$-$\lim_{n \to \infty}F_n = F$. Then $M(F)$ is nonempty and
\begin{equation*}
\lim_{n \to \infty}\inf_{x \in X}F_n(x) = \min_{x \in X}F(x).
\end{equation*}
Moreover, if $x_{n} \in M(F_{n}), n \geq 1$, then the accumulation points of $(x_n)$ belong to $M(F)$.
\end{theo}

\begin{defin}\label{def:chi}
 The indicator  of $A \subset X$ is the function $\chi_{_{ A}} : X \to \R \cup \{+ \infty\}$ defined by $\chi_{_{ A}}(x)=0$, if $x \in A$, and $\chi_{_{ A}}(x)=+ \infty$, if $x \notin A$.
\end{defin}

The following Proposition (see  \cite{Attouch84}, Proposition 4.15.) shows the relation between $\mathrm{K}$-convergence (see Definition \ref{def:Klim}) and $\Gamma$-convergence.

\begin{lem}\label{lem:chi}
Let  $A, A_n  \subset X,\ n \geq 1$. Then $\mathrm{K}$-$\lim_{n \to \infty} A_n = A$ if and only if $\Gamma$-$\lim_{n \to \infty}\chi_{_{A_n}}=\chi_{_{ A}}$.
\end{lem}

\bibliographystyle{alpha}
\bibliography{wasserstein_gPCA}

\end{document}